%% file: rabisim-journal.tex
\documentclass{lmcs}
\pdfoutput=1
\usepackage[utf8]{inputenc}

\usepackage{lastpage}
\lmcsdoi{21}{1}{13}
\lmcsheading{}{\pageref{LastPage}}{}{}%
{May~14,~2020}{Feb.~06,~2025}{}

\usepackage{ifthen}
\usepackage[all]{xy}
\usepackage{bm}
\usepackage{graphicx}
\usepackage{amssymb}
\usepackage{setspace}

\usepackage{soul}

\newif\ifsubmitversion
\submitversiontrue
\newcommand\init{\mathit{init}}

\newcommand{\Act}{\mathcal{A}\mathit{ct}}
\newcommand\clg[1]{\mathcal{#1}}
\newcommand\uratwo{{URA}$_2$}
\newcommand\oprr[1][r]{Lab}
\newcommand\conf{\mathbb C}

\newcommand\sfre{\circledast}

\newcommand\sims{\overset{\mathsf{s}}{\sim}}
\newcommand\simi[1]{\overset{#1}{\sim}}
\newcommand\nsimi[1]{\mathrel{\not\!\overset{#1}{\sim}}}

\newcommand{\FRA}[1]{FRA($#1$)}
\newcommand{\RA}[1]{RA($#1$)}
\newcommand{\FRAbs}[1]{$\boldsymbol{\sim}$-\FRA{#1}}
\newcommand{\RAbs}[1]{$\boldsymbol{\sim}$-\RA{#1}}

\usepackage{listings}
\usepackage{tikz}
\usetikzlibrary{automata,positioning,arrows}
\tikzset{automaton/.style={node distance=2cm,on grid}}
\tikzset{every state/.style={draw=none,minimum size=5pt,inner sep=1pt}}
\tikzset{transition/.style={->,>=stealth',shorten >=1pt}}
\tikzset{every initial by arrow/.style={transition}}
\tikzset{initial text={}}
\tikzset{
    >=stealth',
    punkt/.style={
           rectangle,
           rounded corners,
           draw=black, thick,
           text width=6.5em,
           minimum height=2em,
           text centered},
    pil/.style={
           ->,
           shorten <=2pt,
           shorten >=2pt,},
    qwe/.style={
           -,
           shorten <=2pt,
           shorten >=2pt,}           
}

\newcommand{\cutout}[1]{}
\ifsubmitversion
\newcommand\colortext[2]{#2}
\newcommand\ustwo[1]{#1}
\else
\newcommand\colortext[2]{{\color{#1}#2}}
\newcommand\ustwo[1]{{\color{olive}#1}}
\fi
\newcommand\nt[1]{\colortext{blue}{#1}}

\newcommand\sr[1]{\colortext{olive}{#1}}
\newcommand\ntnote[1]{\sidenote{\nt{#1}}}

\newcommand\us[1]{#1}

\newcommand\mea[1]{\gamma(#1)}
\newcommand\uni[2][]{\mathcal{U}_{#2}^{#1}}

\newcommand\rarr\rightarrow
\newcommand\abra[1]{\langle #1 \rangle}

\newcommand\boldemph[1]{\textbf{\em #1}}
\newcommand\fre{^\bullet}
\newcommand{\D}{\mathcal{D}}

\newcommand{\CMD}[1][r]{\mathsf{OP}_{#1}}
\newcommand{\CMDm}[1][r]{\CMD^{-\!}}
\newcommand{\dom}{\mathsf{dom}}

\newcommand{\calA}{\mathcal{A}}
\newcommand{\calG}{\mathcal{G}}
\newcommand{\calP}{\mathcal{P}}
\newcommand{\calS}{\mathcal{S}}
\newcommand{\calM}{\mathcal{M}}
\newcommand{\calT}{\mathcal{T}}
\newcommand{\calU}{\mathcal{U}}

\newcommand{\trans}[1]{\mathrel{\xrightarrow{#1}}}

\newcommand\xr[1]{\xrightarrow{#1}}

\newcommand\myPP[2][]{\,P_{#2}^{#1}\,}

\renewcommand{\paragraph}[1]{\smallskip\noindent{\bf #1.}}

\newcommand\gen[1]{\mathit{gen}(#1)}
\newcommand\sym[1]{\clg{S}_{#1}}
\newcommand\id[1]{\mathsf{id}_{#1}}

\newcommand\is[1]{\mathcal{IS}_{#1}}

\renewcommand\dom[1]{\mathsf{dom}(#1)}
\newcommand\rng[1]{\mathsf{rng}(#1)}
\newcommand\gfre{^\circledast}

\newcommand\MAX{{\rm MAX}}

\newcommand{\mypara}[1]{\medskip\noindent\emph{#1}.}

\begin{document}

\title[Bisimilarity in Fresh-Register Automata]{Bisimilarity in Fresh-Register Automata}
\titlecomment{This is a revised and extended version of a paper that appeared in LICS'15~\cite{MRT15}}

\author[A.~S.~Murawski]{Andrzej S.\ Murawski \lmcsorcid{0000-0002-4725-410X}}[a]
\address{University of Oxford, UK}

\author[S.~J.~Ramsay]{Steven J. Ramsay \lmcsorcid{0000-0002-0825-8386}}[b]
\address{University of Bristol, UK}

 \author[N.~Tzevelekos]{Nikos Tzevelekos \lmcsorcid{0000-0001-8509-8059}}[c]
 \address{Queen Mary University of London, UK}


\thanks{This research was funded in whole or in part by the UK Engineering and Physical Sciences Research Council (EP/J019577/1, EP/L022478/1)
and the Royal Academy of Engineering (RF 10216/111).
For the purpose of Open Access, the author has applied a CC
BY public copyright licence to any Author Accepted Manuscript (AAM)
version arising from this submission.}

\begin{abstract}
%
Register automata are a basic model of computation over infinite alphabets.
Fresh-register automata extend register automata with the capability to generate fresh symbols in order to model computational scenarios involving name creation.
This paper investigates the complexity of the bisimilarity problem for classes of register and fresh-register automata.
%
%
%
We examine all main  disciplines that have appeared in the literature: general register assignments;
assignments where duplicate register values are disallowed; and assignments without duplicates in which registers cannot be empty. In the general case, we show that the problem is EXPTIME-complete.


However, the absence of duplicate values in registers enables us to
identify inherent symmetries inside the associated
bisimulation relations, which can be used to establish a polynomial
bound on the depth of Attacker-winning strategies. Furthermore, they
enable a highly succinct representation of the corresponding
bisimulations.
By exploiting results from group theory and computational group theory, we can then show {membership} in PSPACE and NP respectively for the latter two register disciplines.
In each case, we find that freshness does not affect the complexity class of the problem.
%
%
%

The results allow us to close a complexity gap for language equivalence of deterministic register automata. We show that deterministic language inequivalence for the no-duplicates fragment is NP-complete, which disproves an old conjecture of Sakamoto.
%

%
%
Finally, we discover that, unlike in the finite-alphabet case, the addition of pushdown store makes bisimilarity undecidable, even in the case of visibly pushdown storage.

\end{abstract}

\maketitle


\input{intro}
\input{defs}
\input{ram}

\input{rash}

\input{sakamoto}
\input{npmagic}
\input{frash}

\input{vpdra}
\input{conclusion}

\section*{Acknowledgments}

We would like to thank M. Jerrum, R. Gray, J. Mitchell and M. Beaudry for useful
 discussions regarding   computational group theory.  
 We are also grateful to the anonymous referees for many helpful suggestions. The research was supported by the Engineering
 and Physical Sciences Research Council (EP/J019577/1, EP/L022478/1) and the Royal Academy of 
Engineering (RF 10216/111).

\bibliographystyle{alphaurl}
\bibliography{my}

\input{apx-hierachy}
\input{apx-se-hardness}

\input{apx-rasihard}
\appendix
\input{apx-ram}
\input{apx-rash}

\input{apx-frash}
\input{apx-pspace}

\end{document}

%% file: intro.tex

\section{Introduction}


Register automata are one of the simplest models of computation over infinite alphabets.
They consist of finite-state control and finitely many registers for storing elements from the infinite alphabet.
Since their introduction by {Kaminski} and Francez~\cite{KF94}  as a candidate formalism for capturing regularity in the infinite-alphabet setting, they have been actively researched especially
in the database and verification communities: selected applications include
 the study of markup languages~\cite{NSV04} and run-time verification~\cite{GDPT13}.
While register automata can detect symbols that are currently not stored in registers (local freshness), the bounded number of registers means that they are not in general capable of recognising inputs that are genuinely fresh in the sense that they occur in the computation for the first time (global freshness).
Because such a feature is desirable in many contexts, notably dynamic resource allocation,
the formalism has been extended in~\cite{Tze11} to fresh-register automata, which do account for global freshness. This paper is concerned with the problem of \emph{bisimilarity testing} for 
register and fresh-register automata.

Bisimulation is a fundamental notion of equivalence in computer science.  
Its central role is, in part, derived from the fact that it is intensional and yet very robust.  
Consequently, the algorithmics of bisimilarity have attracted a lot of attention from researchers interested in the theory and practice of equivalence checking.  
When the set of observable actions available to a system is finite, a lot is already known about the complexity of the problem for specific classes of systems,
although tight bounds are often difficult to obtain in the infinite-state cases~\cite{Srba08}.
In this paper we prove a number of bounds on the complexity of bisimulation equivalence checking. 
We note that in this setting language equivalence is known to be undecidable~\cite{NSV04}.

Our results are expressed using a unified framework that comprises
all variations that have appeared in the literature. They differ in the allowed register assignment discipline, which turns out to affect complexity.
\newcommand\qweqwe[1]{{\bf (#1)}}
Assignments are allowed to be
\begin{description}
\item[\qweqwe{$S$}] \emph{single}, if the contents of all registers are required to be distinct; or
\item[\qweqwe{$M$}] \emph{multiple}, if we allow for duplicate values.
\end{description}
Furthermore, registers are required to
\begin{description}
\item[\qweqwe{$F$}] always be filled; or
\item[\qweqwe{$\#_0$}] initially allowed to be empty; or
\item[\qweqwe{$\#$}] allowed to be erased and filled during a run\footnote{Empty content is\! ``$\#$''. A full definition of each of the automaton variants is given in Section \ref{sec:prelims}.}.
\end{description}
The complexity of bisimilarity checking for each combination is summarised in the table below, where we use the suffix ``-c'' to denote completeness for this class and ``-s'' to denote solvability only.   
The results hold regardless of whether one considers register or fresh-register automata.

\bigskip

\begin{center}
\begin{tabular}{|c|c|c|c|c|c|}
\hline
             ($M\#$)  & ($M\#_0$)  & ($MF$)  & ($S\#$)  & ($S\#_0$) & ($SF$)  \\\hline 
 EXP-c & EXP-c& EXP-c & EXP-c & PSPACE-c & NP-s \\\hline
\end{tabular}
\end{center}
\cutout{
\begin{center}
\begin{tabular}{|@{\hspace{5pt}}c@{\hspace{5pt}}|@{\hspace{5pt}}c@{\hspace{5pt}}|@{\hspace{5pt}}c@{\hspace{5pt}}|@{\hspace{5pt}}c@{\hspace{5pt}}|@{\hspace{5pt}}c@{\hspace{5pt}}|@{\hspace{5pt}}c@{\hspace{5pt}}|}
\hline
             RA($M\#$)  & RA($M\#_0$)  & RA($MF$)  & RA($S\#$)  & RA($S\#_0$) & RA($SF$)  \\\hline 
 EXP & EXP& EXP & EXP & PSPACE & NP \\
\hline
\end{tabular}
\end{center}}

\bigskip
%
Our work thus provides a practical motivation for modelling systems with single assignment whenever possible --- if the system does not need to erase the contents of registers mid-run, the corresponding equivalence problems are lower in the complexity hierarchy.  

We start by giving coarse, exponential-time upper bounds for all the classes of system considered by showing how any such bisimilarity problem can be reduced to one for finite-state automata at exponential cost.  
For all the multiple assignment machines this bound is tight and, for single assignment, tightness depends upon whether or not erasing is allowed.  
The implied significance of being able to erase the contents of registers is explained by our proof that the bisimulation games associated with such systems can simulate the computations of alternating Turing machines running 
in polynomial space.
Here we set up an encoding of the tape, determined by the presence or absence of content in certain registers, and erasing of registers corresponds to writing of tape cells. 

Once erasure is forbidden under single assignments, we obtain better bounds by investigating the structure of the associated bisimulation relations.
Such relations are generally infinite,
but only the relationship between the register assignments in two configurations is relevant to bisimilarity, and so {we} work with a finite, though exponentially large, class of symbolic relations built over partial permutations (to link register indices).  
Due to the inherent symmetry and transitivity of bisimilarity, each such relation forms an inverse semigroup under function composition. 
Also, crucially, the relations are upward closed in the information order.  
Although, taken separately, neither of the preceding facts leads to an exponential leap in succinctness of representation, taken together they reveal an interconnected system of (total) permutation groups underlying each relation.  What is more, in any play of the associated bisimulation game, the number of registers that are empty must monotonically decrease.  
 This, together with an application of Babai's result on the length of subgroup chains in symmetric groups~\cite{B86}, allows us to show that any violation of bisimilarity can be detected after polynomially many rounds of the bisimulation game. Consequently, in this case, we are able to decide bisimilarity in polynomial space.

\ustwo{From a conceptual point of view,  the use of group theory helps us capture symmetries in bisimulation relations, express them in a succinct and structured way and manipulate them effectively.
We regard the use of group-theoretic techniques in this context to be the technical highlight of the paper, and hope that it will inspire further fruitful interplay between
automata over infinite alphabets and computational group theory.}

The polynomial bound mentioned above enables us to close a complexity gap (between NP and PSPACE) in the study of deterministic language equivalence. Namely, we show that 
the language inequivalence problem for \emph{deterministic} RA($S\#_0$) is in NP,
and thus NP-complete, refuting a conjecture by Sakamoto~\cite{S98}.

Further, if registers are additionally required to be filled $(SF)$, we can exhibit very compact representations of the relevant bisimulation relations.
The fact that permutation groups have small generating sets \cite{MN87} allows us then to design a representation for symbolic bisimulations that is at most polynomial in size. 
Furthermore, by exploiting polynomial-time membership testing for permutation groups given in terms of their generators \cite{FHL80}, we 	show that such a representation can be guessed and verified by a nondeterministic Turing machine in polynomial time.

Finally, we consider bisimilarity for visibly pushdown register automata ($\text{VPDRA}$) under the $SF$ register discipline, and we show that the problem here is already undecidable.  
Since VPDRA($SF$) are a particularly weak variant, this result implies undecidability for all PDRA considered in \cite{MRT14}.  
In contrast, for finite alphabets, (strong) bisimilarity of pushdown automata is known to be decidable~\cite{Sen05} 
but non-elementary~\cite{BGKM13}, \us{with ACKERMANN being the best upper bound~\cite{JS19}}.
In the visibly pushdown case, the problem is EXPTIME-complete\,\cite{Sr06}.
\medskip

\paragraph{Related Work}
The complexity of bisimilarity problems has been studied extensively in the finite-alphabet setting and the current state of the art for infinite-state systems is summarised nicely in \cite{Srba08}.  
Recent papers concerning the complexity of decision problems for register automata have, until now, not considered bisimulation equivalence. 
However, there are several related complexity results in the concurrency literature.

In his PhD thesis, Pistore~\cite{Pis99}, gives an exponential-time algorithm for bisimilarity of HD-automata~\cite{MP97}.  
Since Pistore shows that bisimulation relations for HD-automata have many of the
algebraic properties\footnote{{E.g. the \emph{active names} of \cite{Pis99} are comparable to our \emph{characteristic sets}.}} as the relations we study here, it seems likely that our algorithm could be adapted to show 
that the bisimilarity problem for HD-automata is in NP. 
{Indeed, a compact representation of symmetries using generators for such a purpose was envisaged by \cite{CM08}}.

Jonsson and Parrow~\cite{JP93} and Boreale and Trevisan~\cite{BT00} consider bisimilarity over a class of data-independent processes.
These processes are terms built over an infinite alphabet, but the behaviour of such a process does not depend upon the data from which it is built.
In the latter work, the authors also consider a class of value-passing processes, whose behaviour may depend upon the result of comparing data for equality.    
They show that if such processes can be defined recursively then the problem is EXPTIME-complete.
Since value passing can be seen as a purely functional proxy for multiple register assignments, this result neatly reflects our findings for RA($M\#$).
{Finally, decidability of bisimilarity for \FRA{S\#_0} was proven in~\cite{Tze11}, albeit without a proper study of its complexity (the procedure given in \emph{loc.\,cit.} can be shown to run in nondeterministic exponential time).}

Finally, in a recent follow-up paper~\cite{MRT18}, we showed that the  language equivalence problem for deterministic RA($SF$) is in P, in contrast to NP-completeness 
for RA($S\#_0$), established in the present paper. For RA($SF$), this still leaves a complexity gap between NL and P.

{It would be interesting to see to what extent our
decidability and complexity results can be generalised,
e.g. in settings with ordered infinite alphabets or nominal
automata~\cite{BKL14}}. 
\medskip

\paragraph{Structure}  
In Section \ref{sec:prelims} we introduce the preliminaries and  prove all of the EXPTIME 
bounds in Section \ref{sec:ram}. Then we start the  presentation of other results with register automata, as the addition of global freshness requires non-trivial modifications.
In Section \ref{sec:pspace-np-bounds} we show bounds for the ($S\#_0$) problems and apply the techniques to deterministic language equivalence in Section~\ref{sec:saka}. Section~\ref{sec:npmagic} covers further improvements for the ($SF$) case. 
In Section~\ref{sec:frash} we generalise our techniques to fresh-register automata and, finally,
consider the pushdown case in Section~\ref{sec:vpdra}. 

%% file: defs.tex
\newcommand\sw[2]{(#1\,#2)}
\newcommand\xsw[2]{[#1\leftrightarrow#2]}

\section{Preliminaries}\label{sec:prelims}

We introduce some basic notation.
Given a relation $R\subseteq X\times Y$, we define $\dom{R}=\{x\in X\,|\, \exists y. (x,y)\in R\}$ and $\rng{R}=\{y\in Y\,|\, \exists x.(x,y)\in R\}$.
For natural numbers $i\leq j$, we write $[i,j]$ for the set $\{i,i+1,\ldots, j\}$. \us{$\calP(X)$ stands for the powerset of $X$.}

\subsection{Bisimilarity}\label{sec:bisim-defns}

We define bisimulations generally with respect to a labelled transition system.  
As we shall see, the particular systems that we will be concerned with in this paper are the configuration graphs of various classes of {(fresh-)} register automata.

\begin{defi}
A \boldemph{labelled transition system} (LTS) is a \hbox{tuple 
$\clg{S}=(\conf,\Act,\rightarrow)$},
where $\conf$ is a set of {\em configurations},
$\Act$ is a set of {\em action labels},
\us{and ${\rightarrow}\subseteq \conf \times\Act\times\conf$ is a {\em transition relation}.
For $\ell\in\Act$, we use $\mathord{\xrightarrow{\ell}}$ to refer to $\rightarrow\, \cap\, (\conf\times\{\ell\}\times\conf)$.}
%

A binary relation $R\subseteq \conf\times \conf$ is a \boldemph{bisimulation}
if for each $(\kappa_1,\kappa_2)\in R$ and each $\ell\in\Act$,
we have:
\begin{enumerate}
\item if $\kappa_1\xrightarrow{\ell} \kappa_1'$, then there is some
$\kappa_2\xrightarrow{\ell}\kappa_2'$ with $(\kappa_1',\kappa_2')\in R$;
\item if $\kappa_2\xrightarrow{\ell}\kappa_2'$, then there is some
$\kappa_1\xrightarrow{\ell}\kappa_1'$ with $(\kappa_1',\kappa_2')\in R$.
\end{enumerate}
We say that $\kappa_1$ and $\kappa_2$ are \boldemph{bisimilar}, written $\kappa_1\sim \kappa_2$, just if there is some bisimulation
$R$ with $(\kappa_1,\kappa_2)\in R$.
\end{defi}

Let us recall that bisimilarity has a very natural game-theoretic account.
Given two configurations, one can consider a \emph{bisimulation game} involving
two players, traditionally called {\em Attacker} and {\em Defender} respectively.
They play rounds in which
Attacker fires a transition from one of the configurations and Defender has to follow with
an identically labelled transition from the other configuration.
In the first round, the chosen transitions must lead from the configurations to be tested
for bisimilarity, while, in each subsequent round, they must start at the configurations
reached after the preceding round.
Defender loses if he cannot find a matching transition.
In this framework, bisimilarity corresponds to the existence of a winning strategy
for Defender.
The process of playing a bisimulation game
naturally favours Attacker as the decision maker
but, thanks to the forcing technique of \cite{JS08},
it is possible to  construct
transition
systems in which Defender effectively ends up making choices.


\subsection{Fresh-register automata}

We will be interested in testing bisimilarity of configurations
generated by machines with registers and pushdown stack in the infinite-alphabet setting, i.e.
as $\Act$ we shall use the set $\Sigma\times {\D}$ for a finite alphabet $\Sigma$ \us{(with its elements sometimes called \emph{tags})} and an infinite alphabet ${\D}$ {(with its elements sometimes called \emph{names})},
\textit{cf.}\ data words~\cite{NSV04}.

\begin{defi}\label{def:ra}
An $r$-\boldemph{fresh-register automaton} ($r$-{F}RA) is a tuple $\calA=\abra{Q,\Sigma,\delta}$, where:
\begin{itemize}
\item $Q$ is a finite set of states; 
\item $\Sigma$ is a finite set of tags;
\item $\delta\subseteq Q\times\Sigma\times(\calP([1,r])\cup\{\sfre\})\times[0,r]\times\calP([1,r])\times Q $ is the transition relation, with elements written as $q\xrightarrow{t,X,i,Z}q'$.  We assume that in any {such} transition $i \notin Z$.
\end{itemize}
Finally an \boldemph{$r$-register automaton} {($r$-RA)} is a special case of an $r$-FRA such that all its transitions $q\xrightarrow{t,X,i,Z}q'$ satisfy $X\not=\sfre$.
\end{defi}

%
\us{An \emph{$r$-register assignment} is  a mapping of register indices to letters from the infinite alphabet $\D$ and the special symbol $\#$,
  i.e.\ a function:
  \[
    \rho:[1,r]\to{\D} \uplus \{\#\}.
  \]
  The $\#$ symbol} is used to represent the fact that a register is empty, i.e.\ contains no letter from $\D$.  
Consequently, by slight abuse of notation, for any $r$-register assignment $\rho$ we will be writing $\rng{\rho}$ for the set $\rho([1,r])\cap\D$, and
$\dom{\rho}$ for $\rho^{-1}(\rng{\rho})$, where $\rho^{-1}=\{(d,i)\ |\ d\in{\D}\land (i,d)\in \rho\}$.
\us{Finally, we shall use two kinds of assignment update. For any $d\in\D,i\in[0, \ustwo{r}],Z\subseteq[1,\ustwo{r}]$
  and assignment $\rho$ we set:
\begin{align*}
  \rho[i\mapsto d] &= \begin{cases}
    \{(i,d)\} \cup\{(j,\rho(j))\mid j\in[1,n]\setminus\{i\}\} &\text{if }i\neq 0\\
    \rho &\text{otherwise},
  \end{cases}\\
  \rho[Z\mapsto\#] &= \{(j,\#)\mid j\in Z\} \cup\{(j,\rho(j))\mid j\in[1,n]\setminus Z\}.
\end{align*}
Note that, in the former case, no update takes place when $i=0$ but we keep the update notation for notational convenience.}

The meaning of a transition $q \trans{t,X,i,Z} q'$ is described as follows.  The components $t$ and $X$ are a precondition: for the transition to be applicable, it must be that the next letter of the input has shape $(t,a)$ for some $a\in\D$ and, moreover:
\begin{itemize}
\item if $X\subseteq[1,r]$ then $a$ is already stored in exactly those registers named by $X$;
\item if $X=\sfre$ then $a$ is (globally) \emph{fresh}: it has so far not appeared in the computation of $\calA$.
\end{itemize}
If the transition applies then taking it results in changes being made to the current register assignment, namely: $a$ is written into register $i$ (unless $i=0$, in which case it is not written at all) and all registers named by $Z$ have their contents erased.


\begin{defi}\label{d:FRA}
A \boldemph{configuration} $\kappa$ of an $r$-FRA $\calA$ is a triple $(q,\rho,H)$ consisting of a state $q \in Q$, an $r$-register assignment $\rho$ and a finite set $H\subseteq\D$, called the \emph{history}, such that $\rng{\rho}\subseteq H$.
If $q_1 \xrightarrow{t,X,i,Z} q_2$ is a transition of $\calA$, then a configuration $(q_1,\rho_1,H_1)$ can make a transition to a configuration $(q_2,\rho_2,H_2)$
accepting input $(t,d)$, written $(q_1,\rho_1,H_1) \trans{(t,d)} (q_2,\rho_2,H_2)$, just if:
\begin{itemize} 
\item $X = \{j \,|\, \rho_1(j) = d\}$, {or $X=\sfre$ and $d\notin H$};
\item \us{$\rho_2=\rho_1[i\mapsto d][Z\mapsto \#]$};
\item $H_2=H_1\cup\{d\}$.
\end{itemize}
We will sometimes write the set of configurations of $\calA$ by $\mathbb{C}_\calA$ and the induced transition relation by $\rightarrow_\calA$.  We let $\calS(\calA)$ be the LTS $\abra{\mathbb{C}_\calA,\,\Sigma \times \D,\,\rightarrow_\calA}$.

On the other hand, a {configuration} $\kappa$ of an $r$-RA $\calA$ is a pair $(q,\rho)$ of a state $q \in Q$ and an $r$-register assignment $\rho$. The LTS $\abra{\mathbb{C}_\calA,\,\Sigma \times \D,\,\rightarrow_\calA}$ is defined precisely as above, \us{albeit excluding histories and fresh transitions.}
\us{More precisely, if $q_1 \xrightarrow{t,X,i,Z} q_2$ is a transition of $\calA$, then $(q_1,\rho_1) \trans{(t,d)} (q_2,\rho_2)$ just if 
  $X = \{j \,|\, \rho_1(j) = d\}$  
  \us{and $\rho_2=\rho_1[i\mapsto d][Z\mapsto \#]$.}
}
\end{defi}

\us{We define several specific classes of fresh-register automata that we will study in this work 
by considering configurations and transitions restricted according to the  register assignment discipline followed.}
\smallskip

\noindent {\em Duplication in assignment.} We consider two register storage policies, namely single assignment ($S$) or multiple assignment ($M$).  In single assignment, 
we restrict register assignments to be injective on non-empty registers, i.e. for all $i,j \in [1,r]$, $\rho(i) = \rho(j)$
just if $i = j$ or $\rho(i) = \# = \rho(j)$.  In multiple assignment
there is no such restriction.  To ensure that all configurations respect the register assignment discipline, in the ($S$) case every transition $q_1 \xrightarrow{t,X,i,Z} q_2$ is required to 
\us{satisfy the following condition: if $X\subseteq [1,r]$ then \ustwo{$|X|\leq 1$} and if $X\neq\emptyset$ then $i=0$. This simply corresponds to the fact that $d\in\D$ matches
the content of at most one register and, if $d$ is already stored in a register, it will not be written back to any (other) register.}
\smallskip

\noindent {\em Emptiness of registers.} We consider the automaton's ability to process empty registers.  
We say that either all registers must always be filled ($F$), that registers may be initially empty ($\#_0$) or that the contents of registers may be erased ($\#$) during a run.  
Under condition ($F$), {$r$-register assignments are restricted so that $\# \notin \rho([1,r])$}.
Under conditions ($F$) and ($\#_0$), every transition $q_1 \xrightarrow{t,X,i,Z} q_2$ must have $Z = \emptyset$\cutout{and $i\not=0$}.  
Condition ($\#$) imposes no specific restrictions.
\medskip

We describe particular classes by the acronym FRA($XY$) in which 
\[
X\,{\in}\,\{M,\,S\} \textrm{ and } Y\,{\in}\,\{F,\,\#_0,\,\#\}.
\]  
\us{The class FRA($XY$) refers to specialisations of Definitions~\ref{def:ra},~\ref{d:FRA} to transitions and register assignments satisfying the constraints imposed by $X$ and $Y$.  
For instance, FRA($S\#_0$)-configurations are functions from $[1,r]$ to ${\D}\,{\cup}\,\{\#\}$ that are injective on non-empty registers, 
and every transition of such a machine is of the form $q_1 \xrightarrow{t,X,i,Z} q_2$ with 
$X\in\{\sfre,\emptyset\}\cup\{\{j\}\,|\, j\in [1,r]\}$ and $Z=\emptyset$ such that $X=\{j\}$ implies $i=0$.} In a similar manner, 
we define the classes RA($XY$).

\begin{rem}
The class RA($MF$) follows the register assignment discipline of the register automata defined by Segoufin \cite{Seg06}.
The class {RA}($M\#_0$) follow the register assignment discipline of the $M$-Automata defined by Kaminski and Francez \cite{KF94} and the class of RA($S\#_0$) follows the assignment discipline of the finite memory automata considered in the same paper. 
The class RA($SF$) contain automata that follow the register assignment discipline of the machines considered by Nevin, Schwentick and Vianu \cite{NSV04}.
The class \FRA{S\#_0} follow the register assignment discipline of the automata defined in~\cite{Tze11}.
\us{We note that the automata from~\cite{KF94,NSV04,Tze11} mentioned above are a little more restrictive in that
every name encountered by the automaton must be stored in some register, i.e. $i\neq 0$.\footnote{\us{In the conference version of the paper, 
we added this restriction to the definitions of $F$ and $\#_0$. Also, the definition of $S$ was slightly different therein: we stipulated that $X\subseteq \{i\}$,
i.e.\ \us{we allowed an input letter already present in a register to} be unnecessarily overwritten with itself rather than simply preserved (as in the current version).}
\us{These differences between the conference version and the current one were triggered by reviewers' suggestions and do not affect any of the results.}}}
\end{rem}

In this paper we are concerned with the following family of decision problems.

\begin{defi}
Let $X\,{\in}\,\{M,\,S\}$ and $Y\,{\in}\,\{F,\,\#_0,\,\#\}$.
\begin{itemize}
\item The problem \FRAbs{XY} is: given an \FRA{XY} $\calA$ and configurations $\kappa_1=(q_1,\rho_1,H)$ and $\kappa_2=(q_2,\rho_2,H)$,
does $\kappa_1 \sim \kappa_2$ hold in $\clg{S}(\clg{A})$?
\item The problem \RAbs{XY} is: given an \RA{XY} $\calA$ and configurations $\kappa_1$ and $\kappa_2$	,
does $\kappa_1\sim \kappa_2$ hold in $\clg{S}(\clg{A})$?
\end{itemize}
\end{defi}


We shall relate the various classes of bisimilarity problems that we study by their complexity.  We write $P_1 \leq P_2$ to denote that there is a polynomial-time many-one reduction from problem $P_1$ to problem $P_2$.

\begin{lem}\label{lem:hierachy}
The considered bisimilarity problems can be related as in Figure \ref{fig:prob-rel}.
\end{lem}
\begin{proof}
First note that, for all $XY$, any RA($XY$) $\calA$ can be trivially seen as an FRA($XY$) $\calA'$ (i.e.\ $\calA'$ has the same components as $\calA$). We claim that, for any pair $(q_1,\rho_1),(q_2,\rho_2)$ of RA-configurations of $\calA$, 
\[\tag{$*$}
(q_1,\rho_1)\sim(q_2,\rho_2)\iff (q_1,\rho_1,H)\sim(q_2,\rho_2,H)\] 
where $H=\rng{\rho_1}\cup\rng{\rho_2}$ and $(q_1,\rho_1,H),(q_2,\rho_2,H)$ are configurations of $\calA'$. Indeed, we can show that the relation between $\calA$- and $\calA'$-configurations given by:
\[
R =\{\ ((q,\rho),(q,\rho,H))\ |\ \rng{\rho}\subseteq H\ \}
\]
is a bisimulation, from which we obtain~($*$).

We next show the FRA-bisimilarity \us{reductions}; the RA-bisimilarity \us{reductions} are shown in a similar (simpler) way.

Observe that, for any $X \in \{S,M\}$, \FRAbs{XF} $\leq$ \FRAbs{X\#_0} $\leq$ \FRAbs{X\#}.  
This is because any FRA($XF$) can be viewed trivially as an FRA($X\#_0$) in which all registers begin filled and, similarly, any FRA($X\#_0$) can be viewed trivially as an FRA($X\#$) in which no registers are ever erased.  

Now, given an $r$-FRA($S\#$) $\calA$ and two configurations $\kappa_1$ and $\kappa_2$ we construct a $2r$-FRA($MF$) $\calA'$ and configurations $\widehat{\kappa_1}$ and $\widehat{\kappa_2}$ in which every register $k$ of $\calA$ is simulated by two registers $2k-1$ and $2k$ of $\calA'$.  
The representation scheme is as follows: if registers $2k-1$ and $2k$ of $\calA'$ contain the same letter then register $k$ of $\calA$ is empty, otherwise the register $k$ in $\calA$ contains exactly the contents of register $2k$ in $\calA'$.  \us{Additionally, the content of odd-numbered registers in $\widehat{\kappa_1}$, $\widehat{\kappa_2}$ will be the same, which  will make it easy to simulate erasures:
to simulate the erasure of register $k$ in $\calA$ it will suffice to copy the content of register $2k-1$ into $2k$ in $\calA'$.}

The states of $\calA'$ are the states of $\calA$ augmented by an additional state $q^i_\tau$ for every $q \in Q$, $i \in [1,r]$ and every $\tau \in \delta$. 
\us{The extra states subscripted with $\tau$ will be used to simulate potential erasures caused by $\tau$.}

Each transition $\tau = q \trans{t,X,i,Z} q'$ of $\calA$, in which $X\subseteq [1,r]$ and $|X|\le 1$, is simulated by a sequence of transitions of $\calA'$ with the following shape:

\noindent
\begin{center}
\begin{tikzpicture}[automaton]
  \node[state] (q) {$q\vphantom{^0}$};
  \node[state,right=2cm of q] (q1) {$q_\tau^1$}; 
  \node[state,right=2.5cm of q1] (q2) {$q_\tau^2$};
  \node[state,right=.65cm of q2] (dots) {$\cdots$};
  \node[state,right=.5cm of dots] (qr) {$q_\tau^r$};
  \node[state,right=3.5cm of qr] (q') {$q'$};
  \path[transition]
    (q)  edge[above] node{$\scriptstyle t,2X,2i,\emptyset$} (q1)
    (q1) edge[above,out=45,in=135] node{$\scriptstyle t,\{1\},2,\emptyset$} (q2)
    (q1) edge[above,dashed,out=-15,in=195] node{$\scriptstyle t,\{1\},0,\emptyset$} (q2)
    (q1) edge[below,out=-45,in=-135] node{$\scriptstyle t,\{1,2\},0,\emptyset$} (q2)
    (qr) edge[above,out=45,in=135] node{$\scriptstyle t,\{2r-1\},2r,\emptyset$} (q')
    (qr) edge[above,dashed,out=-15,in=195] node{$\scriptstyle t,\{2r-1\},0,\emptyset$} (q')
    (qr) edge[below,out=-45,in=-135] node{$\scriptstyle t,\{2r-1,2r\},0,\emptyset$} (q');
\end{tikzpicture}
\end{center}

\noindent where $2X$ is a shorthand for $\{2x \,|\, x \in X\}$. \ustwo{For each $j \in [1,r]$ the solid (upper) arrow labelled $(t,\{2k-1\},2k,\emptyset)$ 
exists just if $k \in Z$: this transition models erasure of a non-empty register.
The dashed arrow labelled $(t,\{2k-1\},0,\emptyset)$ exists just if $k \notin Z$: it models lack of erasure for non-empty register $k$, but we add these transitions so that $\calA'$
can behave uniformly regardless of whether erasures are needed or not.
The solid (lower) arrow labelled $(t,\{2k-1,2k\},0,\emptyset)$ applies in case register $k$ is empty (we do nothing, regardless of whether $k \in Z$ or not).}
On the other hand, each transition $\tau = q \trans{t,\sfre,i,Z} q'$ of $\calA$ is simulated by the following sequence of transitions of $\calA'$:

\noindent
\begin{center}
\begin{tikzpicture}[automaton]
  \node[state] (q) {$q\vphantom{^0}$};
  \node[state,right=2cm of q] (q1) {$q_\tau^1$}; 
  \node[state,right=2.5cm of q1] (q2) {$q_\tau^2$};
  \node[state,right=.65cm of q2] (dots) {$\cdots$};
  \node[state,right=.5cm of dots] (qr) {$q_\tau^r$};
  \node[state,right=3.5cm of qr] (q') {$q'$};
  \path[transition]
    (q)  edge[above] node{$\scriptstyle t,\sfre,2i,\emptyset$} (q1)
    (q1) edge[above,out=45,in=135] node{$\scriptstyle t,\{1\},2,\emptyset$} (q2)
    (q1) edge[above,dashed,out=-15,in=195] node{$\scriptstyle t,\{1\},0,\emptyset$} (q2)
    (q1) edge[below,out=-45,in=-135] node{$\scriptstyle t,\{1,2\},0,\emptyset$} (q2)
    (qr) edge[above,out=45,in=135] node{$\scriptstyle t,\{2r-1\},2r,\emptyset$} (q')
    (qr) edge[above,dashed,out=-15,in=195] node{$\scriptstyle t,\{2r-1\},0,\emptyset$} (q')
    (qr) edge[below,out=-45,in=-135] node{$\scriptstyle t,\{2r-1,2r\},0,\emptyset$} (q');
\end{tikzpicture}
\end{center}

\noindent where solid and dashed arrows are as above.

We say that a pair of configurations $(q_1,\widehat{\rho_1})$, $(q_2,\widehat{\rho_2})$ of $\calA'$ \emph{represents} a pair of configurations $(q_1,\rho_1)$, $(q_2,\rho_2)$ of $\calA$ just if $\widehat{\rho_1}$ is a representation of $\rho_1$ and $\widehat{\rho_2}$ is a representation of ${\rho_2}$ as discussed above and, furthermore: 
\begin{itemize}
\item for all $k \in [1,r]$, $i \in [1,2r]$, $j \in \{1,2\}$: if $\widehat{\rho_j}(2k-1) = \widehat{\rho_j}(i)$ then $i \in \{2k-1,2k\}$
\item for all $k \in [1,r]$: $\widehat{\rho_1}(2k-1) = \widehat{\rho_2}(2k-1)$
\end{itemize}
These latter two properties can easily be seen to be an invariant of configurations reachable from any pair that initially satisfy it, since transitions of $\calA'$ only write to even numbered registers $2k$ \ustwo{and only} with a fresh letter or the contents of the adjacent register $2k-1$.

By construction, the automaton $\calA'$ faithfully simulates the original in the following sense, given configurations $(q_1,\rho_1)$, $(q_2,\rho_2)$ of $\calA$ and $\calA'$ representations $\widehat{\rho_1}$ of $\rho_1$ and $\widehat{\rho_2}$ of $\rho_2$: $(q_1,\rho_1) \sim (q_2,\rho_2)$ in $\calS(\calA)$ iff $(q_1,\widehat{\rho_1}) \sim (q_2,\widehat{\rho_2})$ in $\calS(\calA')$.
\end{proof}

\newcommand{\deq}{\rotatebox[origin=c]{90}{$\leq$}}
\newcommand\leqe{\hspace*{-3mm}\leq\hspace*{-3mm}}
\newcommand{\FRAbse}[1]{$\hspace*{-1.5mm}\boldsymbol{\sim}$-\FRA{#1}\hspace*{-1.5mm}}
\newcommand{\RAbse}[1]{$\hspace*{-1.5mm}\boldsymbol{\sim}$-\RA{#1}\hspace*{-1.5mm}}

\begin{figure}
\[
\begin{array}{ccccccccccc}
  \text{\FRAbse{SF}} &\leqe& \text{\FRAbse{S\#_0}} &\leqe& \textrm{\FRAbse{S\#}} &\leqe& \text{\FRAbse{MF}} &\leqe& \text{\FRAbse{M\#_0}} &\leqe& \text{\FRAbse{M\#}} \\
  \deq&&\deq&&\deq&&\deq&&\deq&&\deq\\
  \text{\RAbse{SF}} &\leqe& \text{\RAbse{S\#_0}} &\leqe& \text{\RAbse{S\#}} &\leqe& \text{\RAbse{MF}} &\leqe& \text{\RAbse{M\#_0}} &\leqe& \text{\RAbse{M\#}}
\end{array}
\]
\caption{Relationship between the main bisimilarity problems considered in this work.}\label{fig:prob-rel}
\end{figure}

\subsection{Groups and permutations}


\ustwo{Next we introduce notation related to groups and semigroups. Their use will be intstrumental to improving upon our initial EXPTIME bounds. Group-theoretic arguments
and computational procedures based on them will be employed in Sections~\ref{s:SI},~\ref{sec:saka},~\ref{sec:npmagic} to study register automata, and in Section~\ref{sec:frash} in the fresh-register case.}

For any $S\subseteq[1,n]$,
we shall write  $\sym{S}$ for the group of permutations on $S$, and
$\is{S}$ for the inverse semigroup of partial permutations on $S$. For economy, we write $\sym{n}$ for $\sym{[1,n]}$; and $\is{n}$ for $\is{[1,n]}$.
For partial permutations $\sigma$ and $\tau$, we write $\sigma;\tau$ for their relational composition:
\[
  \sigma;\tau=\{\,(i,j)\ |\ \exists k.\sigma(i)=k\land\tau(k)=j\,\}.
\]
\us{Given $i,j\in[1,n]$, we write $(i\ j)$ for the permutation swapping $i$ and $j$, that is, $(i\ j)=\{(i,j),(j,i)\}\cup\{(k,k)\in[1,n]^2\mid k\neq i,j\}$.}
\subsection{Update notation}\label{sec:updates}

\us{%
  We shall be applying updates to partial permutations $\sigma\in\is{n}$, by adding new mappings $[i\mapsto j]$ or pre- or post-composing them with swappings $(i\ j)$. For notational convenience it is useful to have $i,j\in[0,n]$, 
  but extra care is needed when $i=0$ or $j=0$. 
  Given $\sigma\in\is{n}$ and $i,j\in[{0},n]$, we
  let:
\begin{align*}
  \sigma[i\mapsto j] &= \begin{cases}\{(i,j)\}\cup\{(i',j')\in\sigma\ |\ i'\neq i\land j'\neq j\} & \text{if }i,j\in[1,n]\\
\{(i',j')\in\sigma\ |\ j'\neq j\} & \text{if }i=0\text{ and }j\neq0\\
\{(i',j')\in\sigma\ |\ i'\neq i\} & \text{if }i\neq0\text{ and }j=0\\
\sigma & \text{if }i=j=0
\end{cases}\\
  \sigma[i\leftrightarrow j] &= \begin{cases}(i\ j);\sigma & \text{if }i,j\in[1,n]\\
\sigma & \text{if }i=0\text{ or }j=0
\end{cases}\\
  [i\leftrightarrow j]\sigma &= \begin{cases}\sigma;(i\ j) & \text{if }i,j\in[1,n]\\
\sigma & \text{if }i=0\text{ or }j=0
\end{cases}
\end{align*}
Similarly, given $S\subseteq[1,n]$ and $i,j\in[0,n]$, we let:
\[
  S[i\leftrightarrow j] =\begin{cases}\{(i\ j)(k)\mid k\in S\} & \text{if $i,j\in[1,n]$}\\ S & \text{otherwise}
      \end{cases} \quad\vrule
    \quad 
      S[j]= \begin{cases}S\cup\{j\} & \text{if $j\in[1,n]$}\\ S & \text{otherwise}
  \end{cases}
\]
\begin{lem}\label{lem:updates1}
  Given $\sigma,\tau\in\is{n}$ and $i,j,i_x,i_x'\in[0,n]$ (for $x=1,2,3$):
  \begin{itemize}
  \item
    $  \sigma[i\mapsto j]^{-1} = \sigma^{-1}[j\mapsto i]$ \ and \ $(\sigma[{i}\leftrightarrow{j}])^{-1}=[{i}\leftrightarrow{j}]\sigma^{-1}$
    \item $\dom{\sigma\xsw{i}{j}} = \dom{\sigma}\xsw{i}{j}$ \ and \ $\rng{\xsw{i}{j}\sigma}=\rng{\sigma}\xsw{i}{j}$ 
    \item
      $(\xsw{i_2}{i_2'}\sigma)\xsw{i_1}{i_1'}=\xsw{i_2}{i_2'}(\sigma\xsw{i_1}{i_1'})$
      \item
      $(\xsw{i_2}{i_2'}\sigma\xsw{i_1}{i_1'});(\xsw{i_3}{i_3'}\tau\xsw{i_2}{i_2'})=\xsw{i_3}{i_3'}(\sigma;\tau)\xsw{i_1}{i_1'}$
      \item
      $(\sigma[i_1\mapsto i_2]);(\tau[i_2\mapsto i_3])\subseteq(\sigma;\tau)[i_1\mapsto i_3]$.
    \end{itemize}
  \end{lem}
  \begin{proof}
    We only look at the last claim and leave the remaining ones as exercises.
    Given a partial permutation $\pi$ on an arbitrary finite set $X$, and $x,y\in X$, let us write:
    \[
      \pi\langle x\mapsto y\rangle = \{(x,y)\}\cup\{(x',y')\mid x\neq x'\land x\neq x'\}.
    \]
    Given $\pi,\pi'$ and $x,y,z\in X$, we can show that
    \begin{equation}\label{eq:some}
      (\pi\langle x\mapsto y\rangle);(\pi'\langle y\mapsto z\rangle)
      \subseteq
      (\pi;\pi')\langle x\mapsto z\rangle.
    \end{equation}
    Back to the claim,
for any $\sigma\in\is{n}$ and $i,j\in[0,n]$,
    setting $\hat\sigma=\sigma\cup\{(0,0)\}$ and viewing it as a partial permutation on $[0,n]$, we have that
    \
$\sigma[i\mapsto j] = (\hat\sigma\langle i\mapsto j\rangle)\cap[1,n]^2$.
Hence:
\begin{align*}
  (\sigma[i_1\mapsto i_2]);(\tau[i_2\mapsto i_3])
  &= (\hat\sigma\langle i_1\mapsto i_2\rangle\cap[1,n]^2);(\hat\tau\langle i_2\mapsto i_3\rangle\cap[1,n]^2)\\
  &\subseteq (\hat\sigma\langle i_1\mapsto i_2\rangle;\hat\tau\langle i_2\mapsto i_3\rangle)\cap[1,n]^2\\
  &\subseteq (\hat\sigma;\hat\tau)\langle i_1\mapsto i_3\rangle\cap[1,n]^2\quad\text{by }\eqref{eq:some}
\end{align*}
and the latter is $(\sigma;\tau)[i_1\mapsto i_3]$, as required.
\end{proof}
}


%% file: ram.tex

\newcommand{\retag}{\mathsf{retag}}
\newcommand{\DF}{\mathsf{DF}}
\newcommand{\dum}{\heartsuit}
\newcommand{\suc}{\mathsf{suc}}
\newcommand{\A}{\mathsf{A}}
\newcommand{\AO}{\mathsf{A}{\downarrow}}
\newcommand{\AL}{\mathsf{AL}}
\newcommand{\AR}{\mathsf{AR}}
\newcommand{\DL}{\mathsf{DL}}
\newcommand{\DR}{\mathsf{DR}}
\newcommand{\calB}{\mathcal{B}}

\section{Bisimilarity problems complete for EXPTIME}\label{sec:ram}

In this section we show that the upper four classes in our two hierachies of automata all have bisimilarity problems that are complete for exponential time.
\begin{thm}
All of the problems 
\RAbs{S\#}, \RAbs{MF}, \RAbs{M\#_0}, \RAbs{M\#}, 
\FRAbs{S\#}, \FRAbs{MF}, \hbox{\FRAbs{M\#_0}} and \FRAbs{M\#}
are EXPTIME-complete.
\end{thm}
\begin{proof}
The result follows immediately from Propositions \ref{prop:me-solv} and \ref{prop:se-hard} and Lemma \ref{lem:hierachy}.
\end{proof}

Our argument proceeds by showing that \FRAbs{M\#} is in EXPTIME (Proposition \ref{prop:me-solv}) and \RAbs{S\#} is already EXPTIME-hard (Proposition \ref{prop:se-hard}). 
{\ustwo{In the latter case}, we shall rely on alternating linear bounded automata, whose acceptance problem is known to be EXPTIME-complete~\cite{CKS81}.}
\begin{defi}\label{def:ALBA}
An \boldemph{alternating linear bounded automaton} (ALBA) is a tuple 
\[
\calA = \abra{\Gamma,Q_\forall,Q_\exists,q_0,q_\text{acc},q_\text{rej},\delta}.
\]
\us{We let $Q = Q_\forall \uplus Q_\exists\uplus \{ q_\text{acc}\}\uplus \{q_\text{rej}\}$ and call it the set of states, 
assuming the four constituent subsets are pairwise disjoint.}
The components are:
\begin{itemize}
\item a finite tape alphabet $\Gamma$ containing end-of-tape markers $\triangleleft$ and $\triangleright$;
\item disjoint finite sets of universal states $Q_\forall$ and  existential states $Q_\exists$;
\item distinguished initial state $q_0 \in Q$;
\item distinct accepting and rejecting states $q_\text{acc}\neq q_\text{rej}$;
\item a transition function $\delta : \us{(Q\setminus \{q_\text{acc},q_\text{rej}\})} \times \Gamma \to \us{\calP}(Q \times \Gamma \times \{-1,\,+1\})$, satisfying the following properties:
  \begin{enumerate}[(i)]
\item if $(q',a,z) \in \ustwo{\delta(q,\triangleright)}$ then $a = \triangleright$ and $z = +1$; 
\item if $(q',a,z) \in \ustwo{\delta(q,\triangleleft)}$ then $a = \triangleleft$ and $z = -1$; 
\item if $(q',a,z) \in \delta(q,b)$ then $b \in \Gamma \setminus \{\triangleleft,\triangleright\}$ implies $a \in \Gamma \setminus \{\triangleleft,\triangleright\}$.
\end{enumerate}
\end{itemize}
A \boldemph{configuration} of such a machine is a triple $c = (q,\,k,\,t)$ with $q$ a state, $t$ the current tape contents and \ustwo{$k\geq0$} the index of the cell currently under the head of the machine.  
We assume that the tape contents \us{are} of the form $$\triangleright \ a_1 \cdots a_n \ \triangleleft$$ for some letters $a_i \in \Gamma \setminus \{\triangleleft,\triangleright\}$.  
We write $t(k)$ for the content of cell $k$ of tape $t$.  
We say that a configuration $(q,k,t)$ is \boldemph{accepting} (respectively \boldemph{rejecting}, \boldemph{universal}, \boldemph{existential}) just if $q = q_\text{acc}$ (respectively $q = q_\text{rej}$, $q \in Q_\forall$, $q \in Q_\exists$).

A configuration $(q_1,k_1,t_1)$ can make a transition to a \boldemph{successor} $(q_2,k_2,t_2)$ just if there is $a \in \Gamma$ and $z \in \{-1,+1\}$ such that $(q_2,a,z) \in \ustwo{\delta(q_1,t_1(k_1))}$ and $k_2 = k_1 + z$ and $t_2 = t_1[k_1 \mapsto a]$.

Given an input $w \in \Gamma \setminus \{\triangleleft,\triangleright\}$, a \boldemph{computation tree on $w$} for such a machine is an unordered tree labelled by configurations which additionally satisfies the following conditions:
\begin{itemize}
\item The tree is rooted at $(q_0,\ustwo{0},\triangleright w \triangleleft)$.
\item If a universal configuration $c$ labels some node of the tree then this node has one child for each possible successor to $c$.
\item If an existential configuration $c$ labels some node of the tree then this node has exactly one child which can be any successor to $c$.
\end{itemize}
\us{A computation tree is \emph{accepting} if it is finite and all of its leaves are accepting. 
We say that an input $w$ is \boldemph{accepted} just if there is an accepting computation tree on $w$.}
\end{defi}

\begin{defi}
The problem \textsc{ALBA-Mem} is, given an ALBA $\calM$ and an input $w$, to determine whether $w$ is accepted by $\calM$.
\end{defi}

\us{As mentioned above, \textsc{ALBA-Mem} is EXPTIME-complete~\cite{CKS81}.}

\subsection{EXPTIME algorithm}

Given an instance of the $r$-register \FRA{M\#} bisimilarity problem, the main idea is to
\ustwo{consider a bounded version of the associated bisimulation game that uses a finite subset $N \subseteq \D$ of size $2r+2$ as the alphabet. 
One can then determine the winner using an alternating algorithm running in polynomial space.}
This finite set of names is sufficient in order to faithfully capture the full bisimulation game, though a careful discipline is required when making moves with names that are not in the current sets of registers. Such names need to be sourced from the set $N$, in effect re-using names that have appeared before in the game. The crux of the argument is showing that such re-use does not affect the outcome of the (full) game.

\begin{prop}\label{prop:me-solv}
\FRAbs{M\#} is in EXPTIME.
\end{prop}

Given an instance $\abra{\calA,(q_{01},\rho_{01},H_0),(q_{02},\rho_{02},H_0)}$ of the bisimilarity problem for \linebreak[5] \FRA{M\#}, 
where $\calA = \abra{Q,\Sigma,\delta}$ has $r$ registers, 
we first consider a \emph{restricted} bisimilarity problem concerning configurations that contain names from a bounded subset of $\D$.
Let us pick a set $N\subseteq\D$ of cardinality $2r+2$, with a fixed enumeration $N=\{d_1,d_2,\dots,d_{2r+2}\}$, such that:
\begin{enumerate}
\item $H_0\subseteq N$, if $|H_0|< 2r+2$;
\item $\rng{\rho_{01}}\cup\rng{\rho_{02}}\subseteq N\subseteq H_{\us{0}}$, otherwise.  
  \end{enumerate}
  In the former case, $N$ is a superset of  $H_0$, while in the latter it is a subset.  In either case, $N$ includes all names in $\rho_{01},\rho_{02}$.
  We also let the set of \boldemph{$N$-configurations}:
  \[
\mathbb{C}_{\calA,N} = \{ (q,\rho,H)\in\mathbb{C}_\calA\mid H\subsetneq N \}
  \]
  contain all configurations involving names from $N$ and whose histories are strictly included in $N$.
  Given $\rho_1,\rho_2,H$ with \us{${\rng{\rho_{1}}\cup\rng{\rho_{2}}}\subseteq H\subseteq N$ (and hence $\rng{\rho_1}\cup\rng{\rho_2}\subsetneq N$)}, we will sometimes refer to the following trimmed version of $H$:
\[
  \lceil H \rceil_{\rho_1,\rho_2}^N   = \begin{cases} H & \text{ if }H\subsetneq N\\ \us{H}\setminus\{\min(N\setminus(\rng{\rho_1}\cup\rng{\rho_2})\} & \text{ otherwise (i.e.\ if $H=N$)}\end{cases}
  \]
  In the second case, $ \lceil H \rceil_{\rho_1,\rho_2}^N$ is obtained from $H$ by deleting the first name (according to the enumeration of $N$) that  is not present in $\rho_1$ or $\rho_2$.
  Intuitively, the removed name will be recycled and available to simulate global freshness later.
  
We can now define a notion of bisimilarity adapted to $N$-configurations.

  \begin{defi}\label{d:Nbisim}
    Given $\calA$ and $N$ as above, a binary relation $R \subseteq\mathbb{C}_{\calA,\us{N}}\times\mathbb{C}_{\calA,\us{N}}$ is an \boldemph{$N$-bisimulation} if for each $((q_1,\rho_1,H_1),(q_2,\rho_2,H_2)) \in R$ we have $H_1=H_2 (=H)$ and for all $(t,d)$ with $d\in N$:
    \begin{enumerate}[(1)]
    \item if $(q_1,\rho_1,H)\xrightarrow{(t,d)}(q_1',\rho_1',H')$ and one of the following conditions holds:
      \begin{enumerate}[(a)]
\item $d\in\rng{\rho_1}\cup\rng{\rho_2}$,
\item $\rng{\rho_1}\cup\rng{\rho_2}\subsetneq H$ and $d=\min(H\setminus(\rng{\rho_1}\cup\rng{\rho_2}))$, 
\item $d=\min(N\setminus H)$,
\end{enumerate}
then  $(q_2,\rho_2,H)\xrightarrow{(t,d)}(q_2',\rho_2',H')$ and $(({q_1',\rho_1'}, \lceil H' \rceil_{\rho_1',\rho_2'}^N),(q_2',\rho_2',\lceil H' \rceil_{\rho_1',\rho_2'}^N))\in R$;
\item dual conditions hold for $(q_2,\rho_2,H)\xrightarrow{(t,d)}(q_2',\rho_2',H')$.
\end{enumerate}
We say that $\kappa_1$ and $\kappa_2$ are \boldemph{$N$-bisimilar}, written $\kappa_1\sim_N\kappa_2$, just if there is some $N$-bisimulation $R$ with $(\kappa_1 ,\kappa_2) \in R$.
  \end{defi}

  \begin{rem}
    The idea behind
    $N$-bisimulations is that $2r+2$ names suffice in order to decide the bisimilarity problem. Given a pair of configurations $((q_1,\rho_1,H),(q_2,\rho_2,H))$, the specific names in $\rho_1,\rho_2,H$ are immaterial; instead, of importance are:
    \begin{itemize}
    \item the sets of the registers in $\rho_1$ and $\rho_2$ containing the same names;
    \item whether the register assignments contain all names that are included in $H$.
    \end{itemize}
    $2r+1$ names are sufficient for encoding the above information. By allowing $2r+2$ names in total, we are then able to represent the full bisimulation game using only configurations from $\mathbb{C}_{A,N}$.

    To see this, suppose we are at a pair $((q_1,\rho_1,H),(q_2,\rho_2,H))\in\mathbb{C}_{A,N}$ in the (full) bisimulation game and WLOG Attacker chooses to play on the Left, say some $(q_1,\rho_1,H)\xrightarrow{(t,d)}(q_1',\rho_1',H')$. While there may be infinitely many possible choices for $d$, we can narrow them down to finitely many. We can partition $\D$ as:
    \[
\D = (\rng{\rho_1}\cup\rng{\rho_2})\uplus (H\setminus(\rng{\rho_1}\cup\rng{\rho_2}))\uplus (\D\setminus H)
    \]
    and, for each block, only consider a finite number of representatives:
      \begin{enumerate}[(a)]
\item For $\rng{\rho_1}\cup\rng{\rho_2}$ we consider all elements. 
\item For $H\setminus(\rng{\rho_1}\cup\rng{\rho_2})$ we can restrict our attention to the least $d$ in $H\setminus(\rng{\rho_1}\cup\rng{\rho_2})$ and ignore all others, as the specific choice of $d$ from this set has no bearing on the outcome of the bisimulation game. 
\item For $\D\setminus H$, similarly to the previous case, the specific choice of $d$ is not important, so we may as well pick $d$ to be the least element in $N\setminus H$ (which is not empty as $H\subsetneq N$). 
\end{enumerate}
    These three cases precisely correspond to cases (a-c) in Definition~\ref{d:Nbisim}. Our analysis above would allow us to capture bisimilarity using $N$-configurations, if 
    target configurations like $(q_1',\rho_1',H')$ were still in $\mathbb{C}_{A,N}$. This does not always hold, as case~(c) can lead us to $H'=N$. 
   In this case, we use $\lceil H' \rceil_{\rho_1',\rho_2'}^N$ instead of $H'$ so as to remain in $\mathbb{C}_{A,N}$.
  Since $N$ has at least 2 more names than $\rho_1'$ and $\rho_2'$ combined, we can always pick a name from 
  $N\setminus (\rng{\rho_1'}\cup\rng{\rho_2'})$ to remove from $H'=N$ so that
  $\lceil H' \rceil_{\rho_1',\rho_2'}^N\setminus (\rng{\rho_1'}\cup\rng{\rho_2'})$  remains  non-empty. 
  Such a choice will not affect the outcome of the bisimulation game.
  \end{rem}

  \begin{lem}\label{l:Nbisim}
    Given $\calA,(q_{01},\rho_{01},H_0),(q_{02},\rho_{02},H_0)$ and $N$ as above, let
$\hat H_0 = \lceil  H_0\cap N \rceil_{\rho_{01},\rho_{02}}^N$.
    Then, $(q_{01},\rho_{01},H_0)\sim (q_{02},\rho_{02},H_0)$ iff $(q_{01},\rho_{01},\hat H_0)\sim_N (q_{02},\rho_{02},\hat H_0)$.
  \end{lem}
\cutout{
\begin{figure}[t]
\begin{lstlisting}
def $\sf N$_${\sf bisimulation}(\mathcal A,N,q_{01},\rho_{01},q_{02},\rho_{02},\hat H_0)$:
    $q_1,\rho_1,q_2,\rho_2,H$ = $q_{01},\rho_{01},q_{02},\rho_{02},\hat H_0$
    c = 0
    while (c $< \MAX$):
        c += 1
        $\forall\, i\in\{1,2\}$:
            $\forall\,\rm valid$ $(q_i,\rho_i,H)\xrightarrow{(t,d)}(q_i',\rho_i',H')$:
                $\exists\,\rm valid$ $(q_{3-i},\rho_{3-i},H)\xrightarrow{(t,d)}(q_{3-i}',\rho_{3-i}',H')$:
                    $q_1,\rho_1,q_2,\rho_2,H$ = $q_{1}',\rho_{1}',q_{2}',\rho_{2}',\lceil H' \rceil_{\rho_1',\rho_2'}^N$
                    continue
                return NO
    return YES    
\end{lstlisting}
  \caption{\ustwo{Alternating algorithm deciding whether Defender can win an $N$-bisimulation game.}}\label{fig:rev1}
\end{figure}}
\begin{figure}[t]
\ustwo{
\begin{flushleft}
\textbf{INPUT}:  \FRA{M\#} $\mathcal A, N,q_{01},\rho_{01},q_{02},\rho_{02},\hat H_0$\\
\end{flushleft}
\bigskip

\begin{flushleft}
$q_1,\rho_1,q_2,\rho_2,H := q_{01},\rho_{01},q_{02},\rho_{02},\hat H_0$\\[2mm]
\textbf{repeat}\\
\begin{itemize}
\item existentially choose $i\in\{1,2\}$ and valid $(q_i,\rho_i,H)\xrightarrow{(t,d)}(q_i',\rho_i',H')$,\\ or REJECT in the absence of any such choice;
\item universally choose valid $(q_{3-i},\rho_{3-i},H)\xrightarrow{(t,d)}(q_{3-i}',\rho_{3-i}',H')$,\\ or ACCEPT in the absence of any such choice;
\item $q_1,\rho_1,q_2,\rho_2,H := q_{1}',\rho_{1}',q_{2}',\rho_{2}',\lceil H' \rceil_{\rho_1',\rho_2'}^N$
\end{itemize}
\end{flushleft}}

  \caption{\ustwo{Alternating algorithm determining whether Attacker wins the $N$-bisimulation game.}}\label{fig:rev1}
\end{figure}

\ustwo{
It suffices to demonstrate that $N$-bisimilarity can be decided in alternating polynomial space, using the fact that APSPACE = EXPTIME. 
\cutout{
Since the $N$-bisimulation game only involves $N$-configurations, the number of distinct pairs $((q_1,\rho_1,H),(q_2,\rho_2,H))$ that can be produced in a game is bounded by
      \[
        \MAX = |Q|^2\,(|N|+1)^{2r}\,2^{|N|}= |Q|^2\,(2r+3)^{2r}\,2^{2r+2}.
\]
We can therefore bound the number of rounds in an $N$-bisimulation game by $\MAX$.}

\begin{lem}
    Given $\calA,N$ and $(q_{01},\rho_{01},\hat H_0),(q_{02},\rho_{02},\hat H_0)$ as above, we can decide 
    \[
    (q_{01},\rho_{01},\hat H_0)\not\sim_N (q_{02},\rho_{02},\hat H_0)
    \]
    with an alternating algorithm using space $O(r\log r+\log(|Q|))$.
  \end{lem}
  \begin{proof}
\cutout{    We use the algorithm in Figure~\ref{fig:rev1}, which simply plays the $N$-bisimulation game. 
    We can see that the algorithm always terminates and it accepts iff $(q_{01},\rho_{01},\hat H_0)\sim_N (q_{02},\rho_{02},\hat H_0)$. Moreover, the space it uses consists of: the counter $c$ (bounded by $\MAX$); $q_1$ and $q_2$; the assignments $\rho_1,\rho_2$
(each bounded in space by $r\log(2r+2)$); and the 
history $H$ (bounded in space by $(2r+2)\log(2r+2)$). Thus, the overall space used is $O(r\log r+\log(|Q|))$.}
We use the algorithm in Figure~\ref{fig:rev1}, which simply plays the $N$-bisimulation game, exploring existentially a strategy for Attacker. It accepts as soon as Defender cannot defend himself.
Consequently, the algorithm accepts iff $(q_{01},\rho_{01},\hat H_0)\not \sim_N (q_{02},\rho_{02},\hat H_0)$. 
Moreover, the space it uses consists of $q_1, q_2$, the assignments $\rho_1,\rho_2$
(each bounded in space by $r\log(2r+2)$), and the 
history $H$ (bounded in space by $(2r+2)$). Thus, the overall space used is $O(r\log r+\log(|Q|))$.
\end{proof}

}

\subsection{EXPTIME hardness}\label{sec:EXPTIMEhard}

\us{Further down the hierachy, we show that \RAbs{S\#} is EXPTIME-hard by reduction from \textsc{ALBA-Mem}.
The idea is to use the registers of this class of automata to represent the tape content of ALBA's.

For the purposes of the argument, we will assume without loss of generality that we examine ALBA's such that 
$\Gamma \setminus \{\triangleleft,\triangleright\} = \{0,1\}$ and, for all $(q,a)$, $|\delta(q,a)| \leq 2$.
Thus all choices presented by the alternation are binary.
Starting from an instance of the \textsc{ALBA-Mem} problem $\abra{\calM,w}$, we construct a bisimulation problem for RA($S\#$) in which two configurations are bisimilar iff $\calM$ accepts $w$.
From the ALBA $\calM$ we construct an RA($S\#$) $\calA$ that simulates it, with the binary tape content of $\calM$ encoded by the register assignment of $\calA$.
\us{We assume that cells numbered \ustwo{$0$} and \ustwo{$|w|+1$} contain end-markers
and, to each  tape cell \ustwo{$k\in [1,|w|]$}, assign a corresponding pair of registers (\ustwo{$2k$ and $2k+1$}, to be exact)
with exactly one of them being full and the other one being empty (i.e.\ containing $\#$). Then,
cell $k$ will have $0$ written on it iff register \ustwo{$2k$} is empty, and it has $1$ written on it iff register \ustwo{$2k+1$} is empty. This is depicted in Figure~\ref{fig:ALBAtoRA}.
At every step of the bisimulation game, we arrange for Defender to choose transitions from existential states (using Defender forcing \cite{JS08}) and for Attacker to make choices from universal states.}
\us{Without loss of generality, for technical convenience, we will assume that the given ALBA does not diverge, i.e.\ it generates only finite computation paths (Theorem 2.6(b)~\cite{CKS81}).}

\newcommand\padd[1]{\mspace{2mu}#1\mspace{5mu}}

\begin{figure}\ustwo{%
  \begin{gather*}
    \hspace{-2mm}\begin{array}{|p{3mm}|p{3mm}|p{3mm}|p{3mm}|p{3mm}|p{3mm}|p{3mm}|p{3mm}|p{3mm}|p{3mm}|p{3mm}|p{3mm}|p{3mm}|p{3mm}|p{3mm}|}
      \hline
      $\triangleright$ & $0$& $0$&$1$ & $0$& $1$& $1$& $1$& $0$& $0$& $\triangleleft$\\\hline
      \end{array}\\[0mm]
    {\xymatrix@R=1.5em@C=1.33mm{
        &&&&
        \scriptstyle 0&
        \scriptstyle 1\ar@{-}@[gray][dllll]\ar@{-}@[gray][dlll]&
        \scriptstyle2\ar@{-}@[gray][dlll]\ar@{-}@[gray][dll]&
        \scriptstyle3\ar@{-}@[gray][dll]\ar@{-}@[gray][dl]&
        \scriptstyle4\ar@{-}@[gray][dl]\ar@{-}@[gray][d]&
        \scriptstyle5\ar@{-}@[gray][d]\ar@{-}@[gray][dr]&
        \scriptstyle6\ar@{-}@[gray][dr]\ar@{-}@[gray][drr]&
        \scriptstyle7\ar@{-}@[gray][drr]\ar@{-}@[gray][drrr]&
        \scriptstyle8\ar@{-}@[gray][drrr]\ar@{-}@[gray][drrrr]&
        \scriptstyle 9\ar@{-}@[gray][drrrr]\ar@{-}@[gray][drrrrr]&
        \scriptstyle 10\ar@{-}@[gray]\\
                \scriptstyle\padd{1}&
        \scriptstyle \padd{2}&
        \scriptstyle\padd3&
        \scriptstyle\padd4&
        \scriptstyle\padd5&
        \scriptstyle\padd6&
        \scriptstyle\padd7&
        \scriptstyle\padd8&
        \scriptstyle\padd9&
        \scriptstyle 10&
        \scriptstyle 11&
        \scriptstyle 12&
        \scriptstyle13&
        \scriptstyle14&
        \scriptstyle15&
        \scriptstyle16&
        \scriptstyle17&
        \scriptstyle18&
        \scriptstyle19
      }}\\[0mm]
        \begin{array}{|p{3mm}|p{3mm}|p{3mm}|p{3mm}|p{3mm}|p{3mm}|p{3mm}|p{3mm}|p{3mm}|p{3mm}|p{3mm}|p{3mm}|p{3mm}|p{3mm}|p{3mm}|p{3mm}|p{3mm}|p{3mm}|p{3mm}|p{3mm}|p{3mm}|p{3mm}|}
      \hline
      $\scriptstyle d_0$& $\scriptstyle \#$&$\scriptstyle d_1$ &$\scriptstyle  \#$&$\scriptstyle  d_2$&$\scriptstyle  d_3$&$ \scriptstyle \#$&$ \scriptstyle \#$&$ \scriptstyle d_4$&$\scriptstyle  d_5$&$\scriptstyle \#$&$\scriptstyle d_6$&$\scriptstyle \#$&$\scriptstyle d_{7}$&$\scriptstyle \#$&$\scriptstyle \#$&$\scriptstyle d_{8}$&$\scriptstyle \#$&$\scriptstyle d_{9}$\\\hline
      \end{array}
  \end{gather*}\caption{Encoding of a bounded tape of length $9$ (top) using $18$ registers (bottom, registers 2-19). The first register stores an auxiliary name which is used in the reduction of \textsc{ALBA-Mem} to \RAbs{S\#}.}\label{fig:ALBAtoRA}}
\end{figure}


\begin{prop}\label{prop:se-hard}
\RAbs{S\#} is EXPTIME-hard.
\end{prop}

Given an instance $\abra{\calM,\,w}$ of the \textsc{ALBA-Mem} problem, we construct a $2|w|+\ustwo{1}$ register RA($S\#$) $\calA_\calM^w$ whose induced bisimulation game simulates the computations of $\calM$.
A configuration of a computation of $\calM$ will be represented, in duplicate, by a pair of configurations of $\calA_\calM^w$, which together make up a single configuration of the bisimulation game.  
These configurations will track the current state of $\calM$ and the current position of the head of $\calM$ in their state and the current tape contents of $\calM$ will be represented by their current register assignment $\calA_\calM^w$.
We will not require the use of any tags (\textit{cf.} data words) in our construction, so we assume that $\Sigma$ is a unary alphabet and omit this component in transitions. 

\paragraph{Tape encoding}
The \ustwo{first register} is used to help implement a simulation of alternation and will never be empty.
The last $2|w|$ registers of $\calA_\calM^w$ will be used to encode the (non-endmarker) tape content of $\calM$
according to the following scheme: the tape cell   \ustwo{$k \in [1,|w|]$}
\us{\begin{itemize}
\item contains $0$ iff register \ustwo{$2k$} is empty iff register \ustwo{$2k+1$} contains a name;
\item and it contains $1$ iff register \ustwo{$2k$} contains a name iff register \ustwo{$2k+1$} is empty.
  \end{itemize}}

\paragraph{States}
\ustwo{The set of states of  $\calA_\calM^w$ is built from the states of $\calM$, tape cell indices, tape letters, and special tags $L$, $R$:
\begin{align*}
  Q' &= (Q\times[0,|w|+1]\times\{L,R\})\uplus Q_{\rm aux},
\end{align*}
where $Q_{\rm aux}$ is a polynomially-sized set of auxiliary states whose role will be explained  later on.
Thus, each state $p\in Q'\setminus Q_{\rm aux}$
is a tuple $(q,k,C)$ where $q\in Q$ and:
\begin{itemize}
\item $k$ is an index representing the position of the head of the tape of $\calM$,
\item and $C\in\{L,R\}$ is a tag allowing us to have two copies of each state.
\end{itemize}
Given $p\in Q'\setminus Q_{\rm aux}$ and $x \in \{L,R\}$, we write $p[x]$ for the tuple $p$ with its final component replaced by $x$.
Taking an encoding $\rho_I$ of $w$,
our construction of $\calA_\calM^w$ shall ensure that configurations $((q_0,0,L),\rho_I)$ and $((q_0,0,R),\rho_I)$ are bisimilar iff $\calM$ accepts $w$.}



We motivate the construction by looking at the bisimulation game that it induces. A configuration
in that game is a pair of configurations $((p_1,\rho_1),\,(p_2,\rho_2))$ of $\calA_\calM^w$. Our construction shall impose the following invariant at each round of the induced bisimulation game.
If the game is at configuration $((p_1,\rho_1),\,(p_2,\rho_2))$ then:
\[
\rho_1=\rho_2\land (p_1=p_2\lor \exists p. p_1=p[L]\land p_2=p[R]).
\]
The idea is that when the two configurations are of the form $((q,k,C),\rho)$, with $C \in \{L,R\}$, the play is simulating a configuration of $\calM$ which is in state $q$, with the head over tape cell $k$ and the tape contents itself encoded by the last $2|w|$ registers of $\rho$.  

\paragraph{Defender forcing}
In order to describe the transition relation of the automaton we will make use of a gadget to implement defender forcing.
Since the configurations of the induced bisimulation game are guaranteed, by the invariant, to have the same register contents, we are able to instantiate the general construction of \cite{JS08}, in which Attacker is punished for making choices inconsistent with Defender's wishes by allowing Defender to move his configuration into a configuration identical with that of Attacker.

\begin{figure}
\[\xymatrix@R=3em@C=2.7em{
       &       &  p_1 \ar[ld]_{\ell}\ar[rd]^{\ell}\ar[rrrd]^(0.3){\ell}   &        & p_2\ar[ld]_(0.4){\ell}\ar[rd]^{\ell} & \\
       & \cdot\ar[ld]_{\ell_1}\ar[rd]^{\ell_2} &        & \cdot\ar[ld]_(0.3){\ell_2} \ar[rd]^(0.3){\ell_1} &    & \cdot\ar[llllld]_(0.15){\ell_1}\ar[rd]^(0.4){\ell_2} &\\
q_1 &       & q_1' &        &q_2 &     & q_2'\\
}\]
\caption{\ustwo{Defender forcing gadget $\DF(p_1,p_2,\ell,\ell_1,\ell_2,q_1,q_2,q'_1,q'_2)$. Labels $\ell_1$ and $\ell_2$ must be semantically distinct.}}\label{fig:dforcing}
\end{figure}

The gadget is shown in Figure~\ref{fig:dforcing}. The states denoted by dots are the ones constituting the set $Q_{\rm aux}$.
The gadget \ustwo{$\DF(p[L],p[R],\ell,\ell_1,\ell_2,p'[L],p'[R],p''[L],p''[R])$} ensures that, when the game configuration consists of two automata configurations of shape $(p[L],\rho)$ and $(p[R],\rho)$, then 
Defender can force the play so that the game enters a configuration consisting of either two automata configurations \ustwo{of shape $(p'[L],\rho')$ and $(p'[R],\rho')$, or two automata configurations of shape $(p''[L],\rho'')$ and $(p''[R],\rho'')$, where $\rho'$ (respectively $\rho''$) is determined by transition labels $\ell$ and $\ell_1$ (respectively $\ell$ and $\ell_2$)}.  
It is by this defender forcing gadget that we will be able to ensure that the two players correctly simulate existential choices made by $\calM$, essentially by allowing Defender to make the choice.

\paragraph{Transitions}
\ustwo{We describe the transitions of $\calA_\calM^w$ as part of a general description of how the induced bisimulation game simulates $\calM$.  Recall that a configuration of the form $((q,k,C),\rho)$ is used to simulate $\calM$ in operating in state $q$ with the head over cell $k$ of tape encoded by $\rho$.  Simulating a transition of $\calM$ from this configuration
requires reading and updating the tape, but also universally/existentially choosing the successor state.

Given a state $(q,k,C)$ of $\calA_\calM^w$ and $a\in\Gamma$, we do a case analysis on $|\delta(q,a)|$.
If $|\delta(q,a)|=0$ then there are no transitions to add. Otherwise, we proceed as follows.
Let us fix transition labels $A=(\{1\},0,\emptyset)$ and $B=(\emptyset,1,\emptyset)$; these only involve the auxiliary register 1 and are semantically disjoint (from any given configuration, they cannot accept the same $d$). Below, where we use states denoted by dots, these are sourced from $Q_{\rm aux}$.

\paragraph{I.~$\delta(q,a)=\{(q',b,z)\}$}
In this case, it suffices to decode the tape content, update it, and move to the next state. For reasons of uniformity, we will always employ two transitions at this step.
The decoding and updating of the tape is split into two cases, depending on whether the head of the machine being simulated is over an endmarker or not.
If $k \in \{0,|w|+1\}$ then the head is over an endmarker, and the content of cell $k$ is completely determined by $k$, and not updated.
Hence, in such cases we use transitions of the shape
\[
   (q,k,C) \trans{A}\cdot\trans{A} (q',k+z,C).
 \]
It is also useful to define labels $\ell_{\triangleright}=\ell_{\triangleright}'=\ell_{\triangleleft}=\ell_{\triangleleft}'=A$.
 
Otherwise, $k \in [1,|w|]$ and the head is over a cell which is encoded in the way described above.  To decode and update it, we use transitions:
\[(q,k,C) \trans{\ell_a}\cdot\trans{\ell_b'} (q',k+z,C)\]
where $\ell_a$ allows us to decode $a$ from the simulating registers (and reset them), and $\ell_b$ to update them with $b$. According to our encoding scheme:
\begin{align*}
\ell_0 &= (\{2k+1\},0,\{2k+1\}) & \ell_0' &= (\emptyset,2k+1,\emptyset)\\
\ell_1 &= (\{2k\},0,\{2k\}) & \ell_1' &= (\emptyset,2k,\emptyset)
  \end{align*}
Thus, for instance, if $ab=00$ then we use transitions
\[(q,k,C) \trans{\{2k+1\},0,\{2k+1\}}\cdot\trans{\emptyset,2k+1,\emptyset} (q',k+z,C)\]
so the first transition will read a name from register $2k+1$ (representing 0 in position $k$ of the tape) and set that register to $\#$. The next transition will update register $2k+1$ storing a new name $d'$ (representing 0 again).

  \paragraph{II.~$\delta(q,a)=\{(q_1,b_1,z_1),(q_2,b_2,z_2)\}$}
  We consider whether $q$ is a universal or existential move. In the former case,  we add transitions:
  \begin{align*}
(q_1,k+z_1,C) \xleftarrow{\ell_{b_1}'}\cdot\xleftarrow{A}\cdot\xleftarrow{\ell_a}(q,k,C) \trans{\ell_{a}}\cdot\trans{B}\cdot\trans{\ell_{b_2}'} (q_2,k+z_2,C)
    \end{align*}
If, on the other hand, $q$ is existential, 
we use an instance of the Defender forcing gadget: 
\[
\DF((q,k,L),(q,k,R),\ell_a,\vec\ell_{b_1},\vec\ell_{b_2},(q_1,k+z_1,L),(q_1,k+z_1,R),(q_2,k+z_2,L),(q_2,k+z_2,R))
\]
where, by abuse of notation, $\vec\ell_{b_1},\vec\ell_{b_2}$ are sequences of labels defined below.
\[ \vec\ell_{b_1}=A;\ell_{b_1}'\qquad\qquad \vec\ell_{b_2}=B;\ell_{b_2}'
\]
We note that the use of $A$ and $B$ ensures disjointness so that the gadget can be applied.
This ensures that Defender can steer the simulation into her choice whilst maintaining the invariant about the shape of configurations.
}

\paragraph{Accepting and rejecting states}
If the simulation reaches an accepting state then Defender should win.  We organise for this to happen by forbidding any transition out of any state of shape $(q_\text{acc},k,C)$.  In this way, any two configurations that are both in states of this form are trivially bisimilar since neither can perform an action.
Conversely, Attacker should win if the simulation reaches a rejecting state.  We organise for this to happen by transitions of the following shape:
\[
  (q_\text{rej},k,L) \trans{\{1\},0,\emptyset} (q_\text{rej},k,L)
\]
Notice that such transitions only occur in those states that are tagged $L$.  By construction, when the simulation arrives at a rejecting state, one configuration will in such a state tagged with $L$ and the other with $R$ and it follows that the two configurations will not be bisimilar.

\begin{lem}
\us{Given an ALBA $\calM$ and input $w$, 
$\calM$ accepts $w$ iff 
$((q_0,0,L),\rho_I)\sim ((q_0,0,R),\allowbreak \rho_I)$ in 
$\calS({\calA_\calM^w})$, where $\rho_I$ is a register assignment encoding $w$ in the way described above.}
\end{lem}
\begin{proof}
By construction \us{and our assumption that all ALBA computations terminate}, there are only two ways Defender can win a play of the associated bisimulation game.
\begin{enumerate}[(i)]
\item By Attacker choosing a move in the Defender forcing gadget that results in a punishment response from Defender so that every game configuration that follows in the play is of shape $((p,\rho),(p,\rho))$,
i.e. the components are trivially bisimilar.
\item By the play reaching a game configuration in which the two component configurations are of the shape $((q,k,L),\rho)$ and $((q,k,R),\rho)$ for $q = q_\text{acc}$, which are bisimilar by construction. 
\end{enumerate}

In the forward direction, assume that $\calM$ accepts $w$.  Then there is a computation tree $T$ for $w$ in which every leaf is accepting.  Hence Defender can win every play of the corresponding bisimulation game by using $T$ as a representation of a winning strategy.  In particular, for any given play there are two possibilities.  If Attacker plays badly inside a Defender forcing gadget and is punished then the result is (i) above.  Otherwise, as long as  Defender makes choices consistent with $T$ then every play will eventually reach a configuration which simulates $\calM$ in accepting state $q_\text{acc}$.  By construction, the corresponding game configuration must have component configurations of shape $((q_\text{acc},k,L),\rho)$ and $((q_\text{acc},k,R),\rho)$ and Defender wins as described in (ii).  

In the backward direction, assume that Defender has a winning strategy $\us{W}$ for the bisimulation game.  Then, since this strategy must specify which transition to choose when simulating a computation from an existential state and \us{because we assume that the given ALBA terminates}, the strategy can be used to build a \us{finite} computation tree $T$ for $\calM$ on $w$.
Since, by construction, Attacker can always avoid being punished whilst playing in a defender forcing gadget, it follows that $W$ must allow Defender to win any such play by the criterion (ii).  Hence, every simulation which follows $W$ ends in an accepting state and it follows that every leaf of $T$ is accepting.
\end{proof}}


%% file: rash.tex

\newcommand\myR[1][]{\,R_{#1}\,}
\newcommand\myRR[2][]{\,R_{#2}^{#1}\,}
\newcommand\reps[1]{\mathit{rep}(#1)}
\newcommand\ray[2]{\mathsf{ray}^{#1}_{#2}}
\newcommand\Cl[2][]{\textit{Cl}^{#1}(#2)}
\newcommand\SyS{\textsc{(SyS)}}
\newcommand\base[1]{\textsc{BASE}_{#1}}
\newcommand\calL{\mathcal{L}}
\newcommand\fcirc{\bullet}

\newcommand\h{x}

\section{PSPACE-completeness for RAs with single assignment without erasure (RA($S\#_0$)) \label{s:SI}\label{sec:pspace-np-bounds}}

We next prove that the EXPTIME bound can be improved if duplicate values and erasures are forbidden. We handle register automata first to expose the flavour of our technique.
The main result is given below, it follows from Propositions~\ref{psolv} and~\ref{phard}.
\begin{thm}
\RAbs{S\#_0} is PSPACE-complete.
\end{thm}

\subsubsection*{Simplified notation}
Recall that, in any transition $q_1\xr{t,X,i,Z}q_2$ of an $r$-RA($S\#_0$), we have that 
\us{$X\subseteq [1,r]$, $|X|\le 1$, $Z=\emptyset$, and $X\neq\emptyset$ implies $i=0$.}
These restrictions allow for a simpler notation for transitions, with {$\delta\subseteq Q\times\Sigma\times([1,r]\cup\{i\,\fre\ |\ i\in[\us{0},r]\,\})\times Q$:}
\begin{enumerate}[\;(a)]
\item we write each transition $q_1\xr{t,\{i\},0,\emptyset}q_2$ as \ $q_1\xr{t,i}q_2$, where \us{$i\in [1,r]$};
\item and each transition $q_1\xr{t,\emptyset,i,\emptyset}q_2$  as \ $q_1\xr{t,i\fre}q_2$, where \us{$i\in [0,r]$}.
\end{enumerate}
Thus, transitions of type~(a) correspond to the automaton reading an input $(t,a)$ where $a$ is the name in the $i$-th register; while in (b)~transitions the automaton reads $(t,a)$ if $a$ is \emph{locally fresh}, that is, it does not appear in the registers, and in this case $a$ will be stored in register $i$ \us{(for $i\in [1,r]$) or not stored in any register ($i=0$).}

\us{
\subsubsection*{Composition of assignments}

Recall that register assignments in the $S$ case are injective on non-empty registers, we will refer to them as \emph{assignments of type $S$}.
In what follows we will be composing $r$-register assignments $\rho_1,\rho_2$ of type $S$ to obtain partial permutations
capturing the positions of their common names.
\begin{defi}
  Given an $r$-register assignment $\rho$ of type $S$, let us define its inverse by
   \[
\rho^{-1} = \{ (d,i)\in\D\times[1,r]\mid \rho(i)=d\},
\]
i.e. as  the inverse of $\rho\cap ([1,r]\times\D)$.
\end{defi}
We can observe that, if $\rho_1,\rho_2$ are $r$-register assignments of type $S$
then $\rho_1;\rho_2^{-1}$ is a partial permutation. One can show that updates of assignments and permutations are related as follows.
\begin{lem}\label{lem:updates}
  Given $r$-register assignments $\rho_1,\rho_2$ of type $S$,  $d\in\D$ and $i,j\in[0,r]$ such that
  \[
    (d\in\rng{\rho_1}\implies d=\rho_1(i)) \land
    (d\in\rng{\rho_2}\implies d=\rho_2(j)),
  \]
we have $(\rho_1;\rho_2^{-1})[i\mapsto j] = \rho_1[i\mapsto d];\rho_2[j\mapsto d]^{-1}$.
\end{lem}}

\subsection{Symbolic bisimulations}\label{sec:symbolic}

We attack the bisimulation problem \emph{symbolically}, i.e.\ by abstracting actual names in the bisimulation game to the indices of the registers where these names reside. This will lead us to consider groups of finite permutations and inverse semigroups of partial finite permutations.
{\us{In symbolic bisimulations we shall consider pairs} $(q,S)$ of a state $q$ and a set of register indices $S \subseteq [1,r]$, \us{as representing configurations of the form $(q,\rho)$ where $\dom\rho=S$}.  In this way, the locations of the empty registers $[1,r] \setminus S$ are made explicit. \us{Configurations in a symbolic bisimulation relation will consist of triples of the form $(q_1,S_1,\sigma,q_2,S_2)$ where $(q_i,S_i)$ will be as above, while $\sigma\in\is{r}$ shall be a partial permutation matching register indices in $S_1$ to indices in $S_2$. Such tuples will represent concrete configuration pairs of the form $((q_1,\rho_1),(q_2,\rho_2))$ where the $\sigma=\rho_1;\rho_2^{-1}$: in words, $\sigma$ contains all pairs of registers that contain the same name in $\rho_1$ and $\rho_2$ respectively.}

\begin{defi}
Let $\mathcal{A}=\abra{Q,\Sigma,\delta}$ be an $r$-RA($S\#_0$). 
We first set: 
\[\begin{aligned}
\calU_0\ &=\ Q\times\calP([1,r])\times \is{r}\times Q\times \calP([1,r])\\
\calU\ &=\ \{\,
(q_1,S_1,\sigma,q_2,S_2)\in\calU_0\ |\ \sigma\subseteq S_1\times S_2\,\}
\end{aligned}\]
A \emph{symbolic simulation} on $\calA$ is a relation 
$R\subseteq\calU$,
with membership
$(q_1,S_1,\sigma,q_2,S_2)\in R$ often written infix $(q_1,S_1)\myR[\sigma](q_2,S_2)$, such that
all $(q_1,S_1,\sigma,q_2,S_2)\us{\in R}$ satisfy the following \emph{symbolic simulation conditions} \SyS:\footnote{We say that 
\emph{$(q_1,S_1,\sigma,q_2,S_2)$ satisfies the \SyS\ conditions in $R$}.}
\begin{itemize}
\raggedright
\item for all $q_1\xr{t,i}q_1'$,
\begin{itemize}
\item if $i\in\dom{\sigma}$ then there is some $q_2\xr{t,\sigma(i)}q_2'$ with $(q_1',S_1)\myR[\sigma](q_2',S_2)$,
\item if $i\in S_1\setminus\dom{\sigma}$ then there is some $q_2\xr{t,j\fre}q_2'$ with $(q_1',S_1)\myR[{\sigma[i\mapsto j]}](q_2',S_2[j])$;
\end{itemize}
\item for all $q_1\xr{t,i\fre}q_1'$,
\begin{itemize}
\item there is some \hbox{$q_2\xr{t,j\fre}q_2'$} with $(q_1',S_1[i])\myR[{\sigma[i\mapsto j]}](q_2',S_2[j])$,
\item for all $j\in S_2\setminus\rng{\sigma}$, there is some $q_2\xr{t,j}q_2'$ with $(q_1',S_1[i])\myR[{\sigma[i\mapsto j]}](q_2',S_2)$.
\end{itemize}
\end{itemize}\vspace{5pt}
We let the inverse of $R$ be
\[
R^{-1} = \{\,(q_2,S_2,\sigma^{-1},q_1,S_1)\ |\ (q_1,S_1,\sigma,q_2,S_2)\in R\,\}
\]
and call $R$ a \boldemph{symbolic bisimulation} if both $R$ and $R^{-1}$ are symbolic simulations. 
We let \us{\emph{symbolic bisimilarity}}, denoted $\sims$, be the union of all symbolic bisimulations.
We say that $(q_1,\rho_1)$ and $(q_2,\rho_2)$ are \us{\emph{symbolic bisimilar}}
\us{if $(q_1,\dom{\rho_1},\rho_1;\rho_2^{-1},q_2,\dom{\rho_2})\in{\sims}$, i.e.\
$(q_1,\dom{\rho_1})\sims_{\rho_1;\rho_2^{-1}} (q_2,\dom{\rho_2})$.
We will then also write $(q_1,\rho_1)\sims(q_2,\rho_2)$.}
\end{defi} 

\ustwo{Symbolic bisimulation provides a means to finitely represent an otherwise infinite bisimulation relation. The following result proves that this representation is precise. Its proof is based on a case analysis showing that symbolic bisimulation rules capture concrete ones, and vice versa.}

\begin{lem}\label{l:sym1}
Given configurations $(q_1,\rho_1)$, $(q_2,\rho_2)$ of an $r$-RA($S\#_0$), $(q_1,\rho_1)\sim(q_2,\rho_2)\iff(q_1,\rho_1)\sims(q_2,\rho_2)$.  
\end{lem}

It will be useful to approximate symbolic bisimilarity by a sequence of \boldemph{indexed bisimilarity} relations ${\simi{i}} \subseteq\calU$ {defined} inductively as follows.  
First, we let $\simi{0}$ be the whole of $\calU$.  Then, for all $i \in \omega$, 
\us{$(q_1,S_1,\tau, q_2,S_2)\in{\simi{i+1}}$} just if $(q_1,S_1,\tau,q_2,S_2)$ and $(q_2,S_2,\tau^{-1}\!\!,q_1,S_1)$ both satisfy the \SyS\ conditions in $\simi{i}$. We can show the following.

\begin{lem}\label{l:sym1b}
{For all $i\in\omega$, ${\simi{i+1}}\subseteq{\simi{i}}$ and $(\bigcap_{i\in\omega}\simi{i})= {\sims}$.}
\end{lem}

\us{
  \begin{rem}
    Given Lemmata~\ref{l:sym1} and~\ref{l:sym1b}, to obtain a polynomial-space algorithm for bisimilarity, it suffices to obtain a polynomial-space algorithm for symbolic bisimilarity. For the latter, it is enough to establish that symbolic bisimulation games can be decided in polynomially many rounds. In other words, it suffices to show that there is polynomial bound $B$ (dependent on the examined $\mathcal{A}$) such that  ${\simi{B}}= {\sims}$.
  \end{rem}
Our next aim is to show that $\sims$ and each $\simi{i}$ are closed under composition and extension of partial permutations.
Such a closure for $\simi{i}$ will allow us to polynomially bound the convergence of indexed bisimilarities by finding within them strict chains of subgroups (cf.~Lemma \ref{lem:play-length}).
The closure of $\sims$, on the other hand, will help us represent $\sims$ succinctly by appropriate choices of representatives (cf.~Section \ref{sec:npmagic}).}

Given $S_1,S_2\subseteq[1,r]$ and $\sigma,\sigma'\in\is{r}$ we write
$\sigma\leq_{S_1,S_2}\sigma'$ just if $\sigma\subseteq\sigma'\subseteq S_1\times S_2$.
%
%
{Moroever, given $X\subseteq S\subseteq [1,r]$, we write $\id{X}$ for the partial map from $S$ to $S$ that
acts as identity on $X$ (and is undefined otherwise).}
For any 
$R\subseteq\calU$,
we define its \boldemph{closure}  $\Cl{R}$  to be
the smallest relation $R'$ containing $R$ and closed under the following rules.
\begin{gather*}
  \frac{}{ {(q,S,\id{S}, q, S)\in R'}}\;(\textsc{Id}) \qquad
   \frac{(q_1,S_1,\sigma_1,q_2, S_2)\in R'\qquad (q_2,S_2,\sigma_2, q_3, S_3)\in R'}{(q_1, S_1, \sigma_1;\sigma_2, q_3, S_3)\in R'}\;(\textsc{Tr})
\\[2mm]
 \frac{(q_1,S_1,\sigma, q_2, S_2)\in R'}{ (q_2,S_2,\sigma^{-1}, q_1, S_1)\in R'}\;(\textsc{Sym}) \qquad
 \frac{(q_1,S_1,\sigma, q_2, S_2)\in R'\qquad \sigma\le_{S_1,S_2} \sigma'}{ (q_1,S_1,\sigma', q_2, S_2)\in R'}\;(\textsc{Ext})
\end{gather*}
\cutout{
\begin{align*}
\Cl{R}\,=\,\{\, (q_1,S_1,\sigma',q_{n+1},S_{n+1})\mid\ 
&n\geq 1\ 
\land\ \sigma_1;\sigma_2;\cdots;\sigma_n\leq_{S_1,S_{n+1}}\sigma'
\\
&\land\ \forall 1\leq i\leq n.\ (q_i,S_i,\sigma_i,q_{i+1},S_{i+1})\in R\cup R^{-1}
\,\}
\end{align*}
}%
{We say that $R$ is \emph{closed} in case {$\Cl{R} = R$}.}

\ustwo{Much of the following development relies upon the fact that bisimilarity and indexed bisimilarity are closed. Intuitively, this amounts to showing that the \SyS\ conditions are compatible with the rules above, i.e.\ if their premises satisfy the conditions then so do the conclusions. The interesting cases are $(\textsc{Tr})$ and $(\textsc{Ext})$. For the former, the argument is a symbolic version of showing that (bi)simulation is transitive. The case of $(\textsc{Ext})$ is subtler, as we need to argue that it is sound to relate previously unrelated registers.
}

\begin{lem}\label{l:sym2}
{Let $P,R\subseteq\calU$. If all \us{$g\in R\cup R^{-1}$} satisfy the \SyS\ conditions in $P$ then all $g\in\Cl{R}$ satisfy the \SyS\ conditions in $\Cl{P}$.}
\end{lem}


\begin{cor}(Closures)\label{c:closures}
Bisimilarity and indexed bisimilarity for RA($S\#_0$) are both closed:
\begin{enumerate}
\item $\sims \;\;=\: \Cl{\sims}$\,;\quad
\item for all $i \in \omega$: $\simi{i} \;\;=\: \Cl{\simi{i}}$.
\end{enumerate}
\end{cor}
\begin{proof}
For~1 note that ${\sims}={(\sims)}^{-1}$ and all its elements satisfy the \SyS\ conditions in $\sims$. Hence,
by Lemma~\ref{l:sym2} 
we have that $\Cl{\sims}$ is a symbolic bisimulation, i.e.\ $\Cl{\sims}\subseteq{\sims}$.
The result then follows.
For~2 we proceed by induction on $i$.  When $i = 0$ then the result follows from the fact that $\simi{0}$ is the universal relation.
For the inductive case, note first that $\simi{i+1}$ is symmetric by construction and all $g\in{\simi{i+1}}$ satisfy the \SyS\ conditions in $\simi{i}$. Hence, by Lemma~\ref{l:sym2}, all elements of $\Cl{\simi{i+1}}$ satisfy the \SyS\ conditions in $\Cl{\simi{i}}$. By IH, $\Cl{\simi{i}}={\simi{i}}$ so $\Cl{\simi{i+1}}\subseteq{\simi{i+1}}$, as required.	
\end{proof}

\newcommand{\relR}{\mathrel{{R}}}

\subsection{\us{Bounding indexed bisimilarity convergence using permutation groups}}\label{sec:groups}

\us{
  To bound the rate of convergence of indexed bisimilarities we study the strict sub-chains:
  \begin{equation}\label{eq:chains}
\{ {\simi{i}} \,|\, (\simi{i+1} \cap\ \calU_{S_1,S_2}) \subsetneq (\simi{i} \cap\ \calU_{S_1,S_2})\}
    \end{equation}
    that we obtain for a given pair of sets $S_1,S_2 \subseteq [1,r]$, where:
    \[\calU_{S_1,S_2} = \{ (q_1,S_1',\sigma,q_2,S_2')\in\calU\mid S_1=S_1', S_2=S_2'\}.\]
Our aim is to find a bound for $i$ in~\eqref{eq:chains}, independent of $S_1,S_2$.
\ustwo{To this end, below we introduce two auxiliary notions that will help us identify some structure within the $\simi{i}$ relations.
In particular, we shall study self-symmetries, which lead to group-theoretic considerations and enable us to relate 
the evolution of $\simi{i}$ to descending subgroup chains.}
\begin{defi}
  Let $p\in Q,S\subseteq[1,r]$ and $R\subseteq\calU$ be closed. We define:
  \begin{itemize}
  \item the \emph{\textbf{characteristic set} of $(p,S)$ in $R$} as: \ $X_S^p(R) = \bigcap\{ X\subseteq S\mid (p,S)\relR_{\id{X}}(p,S)\}$,
  \item the \emph{\textbf{characteristic group} of $(p,S)$ in $R$} as: \ $\clg{G}_S^p{(R)}=\{ \sigma\subseteq X_S^p{(R)}\times X_S^p{(R)}\,|\, (p,S)\relR_\sigma (p,S)\}$.
  \end{itemize}
\end{defi}

\ustwo{Note that $R_1\subseteq R_2$ implies  $X_S^p(R_1) \supseteq  X_S^p(R_2)$.}
We are going to show (in Lemma~\ref{lem:play-length}) that 
changes in $\simi{j}\cap\ {\mathcal{U}_{S_1,S_2}}$ (as $j$ increases) can be traced back to 
either \us{expansion} of a characteristic set $X^p_S({\simi{j}})$ ($S\in\{S_1,S_2\}$), 
or shrinkage of some $\clg{G}_S^p({\simi{j}})$ ($S\in\{S_1,S_2\}$)
or disappearance of all tuples $(q_1, S_1,\sigma, q_2, S_2)$ for some $q_1,q_2\in Q$. The number of changes of each kind can be bounded by a polynomial. In the second case, we shall rely on the fact that each $\clg{G}_S^p({\simi{j}})$ is indeed a group (Lemma~\ref{lem:group}) and on the following result which concerns subgroup chains in a group $G$:
\[  G = G_0 > G_1 > \cdots{} > G_m = I 
\]
in which $I$ is the trivial identity group and, for all $i \in [0,m-1]$, $G_{i+1}$ is a strict subgroup of $G_i$.

\begin{thmC}[\cite{B86}]\label{thm:babai}
For $n \geq 2$, the length of every subgroup chain in $\sym{[1,n]}$ is at most $2n-3$.
\end{thmC}

\us{
\begin{rem}
The above result provides a linear bound, which we will be using in subsequent calculations. Note, though, that the existence of a quadratic bound 
follows easily from Lagrange's theorem. In particular, it  implies $|G_i| \ge 2 |G_{i+1}|$ ($0\le i < m$) and, thus, $|G| \ge 2^m$. 
Consequently, $m\le \log_2( |G| ) \le \log_2(n!) \le n\log_2(n)\le n^2$. 
\end{rem}
}
  




Before tackling Lemma~\ref{lem:group}, we prove an auxiliary lemma.
\begin{lem}\label{lem:bijections}
Let $p,S,R$ be as above.  Suppose $(p,S)\relR_{\sigma} (q,S)$, then:
  \begin{itemize}
  \item $\dom{\sigma}\supseteq X_S^p{(R)}$ and $\rng{\sigma}\supseteq X_S^q{(R)}$.
  \item Setting $\sigma'=\sigma \cap (X_S^p{(R)}\times X_S^q{(R)})$, we have $\dom{\sigma'}=X_S^p{(R)}$, $\rng{\sigma'}=X_S^q{(R)}$ and $(p,S)\relR_{\sigma'}(q,S)$. In particular, $(p,S)\relR_{\id{X_S^p{(R)}}}(p,S)$.
      \end{itemize}
\end{lem}

\begin{rem}
The above Lemma shows that  {$R \cap \calU_{S,S}$}
can be generated from 
elements of the form
$(p,S)\relR_\sigma (q,S)$,
where $\sigma$ is a bijection between $X_S^p{(R)}$ and $X_S^q{(R)}$, using up-closure under $\le_{S,S}$.
That is, $(p,S)\relR_{\sigma'} (q,S)$ iff there exists 
a bijection $\sigma:X_S^p{(R)}\rarr X_S^q{(R)}$ such that $\sigma\le_{S,S} \sigma'$ and $(p,S)\relR_{\sigma} (q,S)$.
\end{rem}

\begin{lem}\label{lem:group}
$\clg{G}_S^p{(R)}$ is a group (under composition).
In particular, it is a subgroup of $\sym{X_S^p{(R)}}$.
\end{lem}
\begin{proof}
By the last part of Lemma~\ref{lem:bijections}, we  have $\id{X_S^p{(R)}}\in \clg{G}_S^p$.
Now let $\sigma\in\clg{G}_S^p{(R)}$, i.e.\ $(p,S)\relR_\sigma (p,S)$ with $\sigma\subseteq X_S^p{(R)}\times X_S^p{(R)}$.
By first part of Lemma~\ref{lem:bijections}, we have $\sigma\in\sym{X_S^p{(R)}}$.
Moreover, $(p,S)\relR_{\sigma^{-1}} (p,S)$, by closure of $R$, hence $\sigma^{-1}\in\clg{G}_S^p{(R)}$.
Finally, if $\sigma'\in\clg{G}_S^p{(R)}$, again using closure of $R$, we get $\sigma;\sigma'\in\clg{G}_S^p{(R)}$.
\end{proof}
}

{We can now use the above structure in indexed bisimilarities to bound their rate of convergence.}

\begin{lem}\label{lem:play-length}
Given $S_1,S_2 \subseteq [1,r]$,
 the sub-chain $\{ \simi{i} \,|\, (\simi{i+1} \cap\ \calU_{S_1,S_2}) \subsetneq (\simi{i} \cap\ \calU_{S_1,S_2})\}$ has size $O(|Q|^2 + r^2|Q|)$.	
\end{lem}
\begin{proof}
Fix $S_1,S_2 \subseteq [1,r]$.
We argue that $\{\simi{i} \,|\, (\simi{i+1} \cap\ \calU_{S_1,S_2}) \subsetneq (\simi{i} \cap\ \calU_{S_1,S_2})\}$ has length at most $|Q|^2 + 4r^2|Q| - 2r|Q|$.

Let us say that two configurations $(q_1,S_1)$ and $(q_2,S_2)$ are \emph{separated} in $\simi{i}$ just if there is no $\sigma$ such that $(q_1,S_1) \simi{i}_\sigma (q_2,S_2)$; we say they are \emph{unseparated} otherwise.
We claim that if $(\simi{i+1} \cap\ \calU_{S_1,S_2}) \subsetneq (\simi{i} \cap\ \calU_{S_1,S_2})$ then: 
\begin{enumerate}[(i)]
\item there is some $q \in Q$ and $S \in \{S_1,S_2\}$ such that $X^q_S({\simi{i+1}}) \supsetneq X^q_S({\simi{i}})$,
\item or there is some $q \in Q$ and $S \in \{S_1,S_2\}$ such that $\calG^q_S({\simi{i+1}})$ is a strict subgroup of $\calG^q_S({\simi{i}})$,
\item or there are configurations $(q_1,S_1),(q_2,S_2)$ that are unseparated in $\simi{i}$ and become separated in $\simi{i+1}$.
\end{enumerate}
We argue as follows.  
If $(\simi{i+1} \cap\ \calU_{S_1,S_2}) \subsetneq (\simi{i} \cap\ \calU_{S_1,S_2})$ then there are some $p,q \in Q$ and $\sigma$ such that $(q_1,S_1) \simi{i}_\sigma (q_2,S_2)$ but $(q_1,S_1) \nsimi{i+1}_\sigma (q_2,S_2)$.  
Note that, in such a case it follows that also $(q_1,S_1) \simi{i}_{\sigma'} (q_2,S_2)$ and $(q_1,S_1) \nsimi{i+1}_{\sigma'} (q_2,S_2)$, where $\sigma' = \sigma \cap (X^{q_1}_{S_1}({\simi{i}}) \times X^{q_2}_{S_2}({\simi{i}}))$, by closure of $\simi{i}$ \us{(\textsc{Tr} for $\id{X^{q_1}_{S_1}(\simi{i})}$ and $\id{X^{q_2}_{S_2}(\simi{i})}$)} and $\simi{i+1}$ \us{(contraposition with (\textsc{Ext}))}. 
Hence, we assume wlog that $\dom{\sigma} = X^{q_1}_{S_1}({\simi{i}})$ and $\rng{\sigma} = X^{q_2}_{S_2}({\simi{i}})$.  
Now, suppose that, for all $q \in Q$, $S \in \{S_1,S_2\}$, $X^q_S({\simi{i+1}}) = X^q_S({\simi{i}})$ and no previously unseparated pair of configurations become separated in $\simi{i+1}$.  
It follows \us{from $(q_1,S_1) \nsimi{i+1}_\sigma (q_2,S_2)$}
that there is some $\tau$ such that $(q_1,S_1) \simi{i+1}_\tau (q_2,S_2)$ and thus $\sigma;\tau^{-1} \in \calG^{q_1}_{S_1}({\simi{i}})$ but $\sigma;\tau^{-1} \notin \calG^{q_1}_{S_1}({\simi{i+1}})$. Hence $\calG^{q_1}_{S_1}({\simi{i}}) > \calG^{q_1}_{S_1}({\simi{i+1}})$.

\us{To bound the length of the chain $\{ \simi{i} \,|\, (\simi{i+1} \cap\ \calU_{S_1,S_2}) \subsetneq (\simi{i} \cap\ \calU_{S_1,S_2})\}$, observe that
we always have $X^q_S({\simi{i+1}}) \supseteq X^q_S({\simi{i}})$ because of ${\simi{i+1}}\subseteq{\simi{i}}$.
Thus, (i) may happen at most $2r|Q|$ times inside the chain. If (i) does not hold then $X^q_S({\simi{i+1}})= X^q_S({\simi{i}})$ for all $q$ and $S\in\{S_1,S_2\}$.
For fixed $X^q_S({\simi{i}})$, by Theorem \ref{thm:babai},
(ii) may happen at most $2r-2$ times (we include the case $r=1$), which gives an upper bound of $2 r |Q| (2r-2)$ for the number of such changes inside the whole chain
(under the assumption that the changes are not of type (i), which have already been counted).
Finally, the remaining changes must be of type (iii) and may happen at most $|Q|^2$ times across the whole chain.
Overall, we obtain $2r| Q|+2r |Q| (2r-2)+|Q|^2=|Q|^2 + 4r^2|Q| - 2r|Q|$ as a bound on the length of the given chain.}
\end{proof}

Note that it does not quite follow from the above result that the sequence $(\simi{i})$ converges in polynomially many steps, because there are exponentially many pairs $(S_1,S_2)$. Next we shall establish such a bound by studying more closely the overlap in evolutions of different $(S_1, S_2)$. 

%
\us{
\begin{lem}\label{lem:bound}
Let $\ell$ be the bound from Lemma~\ref{lem:play-length} and $B=(2r+1)\ell$.
Then, for any $S_1,S_2$,  $\simi{B} \cap\ \uni{S_1,S_2} =\ \sims\cap\ \uni{S_1,S_2}$.
\cutout{{Let $c$ be the constant of $O(|Q|^2 + r^2|Q|)$ in Lemma~\ref{lem:play-length}.}
\ntnote{proof moved to the appendix}
\begin{enumerate}
\item
  For any $S_1,S_2$, and
 $j=c (2r-\mea{S_1,S_2}+1)(|Q|^2 + r^2|Q|)$,
  we have
$\simi{j} \cap\ \uni{S_1,S_2} =\ \sims\cap\ \uni{S_1,S_2}$.
\item Let $B=c(2r+1)(|Q|^2 + r^2|Q|)$. For any $S_1,S_2$, 
$\simi{B} \cap\ \uni{S_1,S_2} =\ \sims\cap\ \uni{S_1,S_2}$.
\end{enumerate}}
\end{lem}}

\begin{prop}\label{p:bstrategy}
For any RA($S\#_0$) bisimulation problem, if there is a winning strategy for Attacker then there is one of depth $O(r|Q|^2+r^3|Q|)$.
\end{prop}
\begin{proof}
We first observe that bisimulation strategies and their corresponding symbolic bisimulation strategies have the same depth. Thus, it suffices to bound symbolic strategies for Attacker.
The $O(r|Q|^2+r^3|Q|)$ bound follows from the preceding Lemma.
\cutout{
Let $\cal S$ be such a strategy of minimal depth and suppose $\cal S$ has root $(q_1,S_1,\tau,q_2,S_2)$. 
By the previous lemma, the subtree $\cal S'$ of $\cal S$ whose nodes have sets exclusively from $\{S_1,S_2\}$ has depth bounded by $O(|Q|^2+r^2|Q|)$. Applying the same reasoning to the children $(q_1',S_1',\tau',q_2',S_2')$ of the leaves of $\cal S'$ with $\{S_1',S_2'\}\not\subseteq\{S_1,S_2\}$, and so on, we obtain a directed tree
where each edge abstracts a path of length $O(|Q|^2+r^2|Q|)$
{and whose depth is bounded by $2r$.} 
We now combine the latter with our previous bound.}
\end{proof}
\begin{prop}\label{psolv}
\RAbs{S\#_0} is in PSPACE.
\end{prop}
\begin{proof}
  \us{
Thanks to the bound obtained in Lemma~\ref{lem:bound}, to decide symbolic bisimilarity it suffices to play the corresponding symbolic bisimulation game for polynomially many steps. The existence of a winning strategy can then be established by an alternating Turing machine running in polynomial time, \ustwo{analogously to Figure~\ref{fig:rev1}}.  The PSPACE bound follows from APTIME\,$=$\,PSPACE.}
\end{proof}

\subsection{PSPACE hardness}

\us{
For PSPACE-hardness, we reduce from the well-known PSPACE-complete problem of checking validity of totally quantified boolean formulas in prenex conjunctive normal form.
\sr{
  One possibility is to decompose this reduction via the acceptance problem for ALBA that are not allowed to overwrite non-blank tape cells\,--\,\emph{write-once ALBA}.  Given an instance of QBF, one can construct a write-once ALBA with enough space on its tape to store the formula and a truth assignment, which it guesses by alternating moves according to the quantifiers, and then verifies deterministically.  Then our reduction of Section~\ref{sec:EXPTIMEhard} applies to obtain an instance of the bisimilarity problem for RA($S\#$) but, because the ALBA is write-once, so the corresponding RA obeys $S\#_0$.
  However, there is a more straightforward, direct reduction, which we present below.
}

\cutout{
The reduction is similar to that from {\sc ALBA-Mem} to RA($S\#$) presented in Section~\ref{sec:EXPTIMEhard}. While in the latter case we were representing a (read-write) linear bounded tape using the registers of an RA, for quantified boolean formulas we need to represent a truth-value assignment, which can be seen as a write-once tape. 
This blends well with the fact that RA($S\#_0$)'s can overwite an empty register but not empty it again. \ustwo{In effect, PSPACE-harness can be obtained by the same construction as in Section~\ref{sec:EXPTIMEhard} albeit for a ALBA's that utilise a write-once tape~\cite{} ({\it NT: is there a reference for write-once ALBAs?}). Here, instead, we present a simpler construction directly adapted to QBF.}
}

In our construction,
{universal quantification and selection of conjuncts is performed by Attacker.
For existential quantification and disjunctions, we rely on Defender Forcing.
The choices of truth values by both players are recorded in registers by using, for each variable $x_i$,  \ustwo{registers $2i, 2i+1$, both initialised to $\#$. 
If a player chooses \emph{true} for $x_i$, we fill register $2i$ leaving $2i+1$ empty;}
we do the opposite otherwise. This makes it possible to arrange for bisimilarity/non-bisimilarity
(as appropriate) in the final stage of the game, depending on whether the resulting literal is negated.
 \begin{prop}\label{p:rasihard}\label{phard}
\RAbs{S\#_0} is PSPACE-hard.
\end{prop}

}
\begin{proof}
We reduce from TQBF, i.e. the problem of deciding whether a formula $\Phi$ of the shape $\square_1 x_1\cdots \square_h x_h. \phi(x_0,\cdots,x_h)$ (with $\phi$ in conjunctive normal form \ustwo{and each $\square$ a quantifier}) is true. 

We shall construct a $(2h+\ustwo{1})$-register RA($S\#_0$)  and configurations $\kappa_L,\kappa_R$ such that $\kappa_L\sim\kappa_R$ if and only if $\Phi$ is true.
We will not require the use of any tags in our construction, so we assume that $\Sigma$ is a unary alphabet and omit this component in transitions. 
\ustwo{We pick some name $d_0$.
For $C\in \{L,R\}$, we shall have
\[\kappa_C=((q_1,C),\rho_0)\]
with $\rho_0(1)=d_0$ and $\rho_0(i)=\#$ for all other $i$.}

\us{The \ustwo{first register is used} to let Attacker/Defender make choices.
\ustwo{Registers $2,\cdots, 2h+1$ will represent truth-value assignments. Registers $2i,2i+1$} will be used to represent the value of $x_i$ ($i=1,\cdots,h$) subject to the following conditions:
\begin{itemize}
\item register \ustwo{$2i$} is filled if and only if the value of $x_i$ is true, 
\item register \ustwo{$2i+1$}  is filled if and only if the value of $x_i$ is false.
\end{itemize}}
The values will be selected by Attacker (when $\square_i=\forall$) or Defender (when  $\square_i=\exists$). 
Formally,
if $\square_i=\forall$ then
we add the following transitions, where \ustwo{$A=1$ and $B=1\fre$},
\[
  \xymatrix{
    &(q_{i},L) \ar[ld]_{\ustwo A} \ar[dr]^{\ustwo B} &&& (q_{i},R) \ar[ld]_{\ustwo A} \ar[dr]^{\ustwo B}\\
(q_{i}^T,L) && (q_i^F,L)  & (q_{i}^T,R) && (q_i^F,R)
}
\]
which allows Attacker to force the play from $((q_{i},L),(q_{i},R))$ into
either $((q_{i}^T,L),(q_{i}^T,R))$ or $((q_{i}^F,L),(q_{i}^F,R))$.
On the other hand, if $\square_i=\exists$ then
we add (cf.\ Figure~\ref{fig:dforcing}):
\ustwo{\[
    \DF((q_{i},L),(q_{i},R),A,A,B,(q_{i}^T,L),(q_{i}^T,R),(q_{i}^F,L),(q_{i}^F,R))
    \]}
We follow up the above transitions with register-setting ones:
\[
\xymatrix{(q_i^T,L) \ar[dr]^{{\ustwo{(2i)}}^\fcirc} && (q_i^F,L) \ar[dl]_{\ustwo{(2i+1)}^\fcirc} &(q_i^T,R) \ar[dr]^{{\ustwo{(2i)}}^\fcirc}& & (q_i^F,R) \ar[dl]_{\ustwo{(2i+1)}^\fcirc}\\ 
& (q_{i+1},L)  &&& (q_{i+1},R)
}
\]

The above handles  quantification. 
To represent the formula $\phi=\phi_1\wedge\cdots\wedge \phi_k$, we allow
Attacker to force the play from $((q_{h+1},L),(q_{h+1},R))$ into any of $((q_{(h+1)l},L),(q_{(h+1)l},R))$ for $l=1,\cdots, k$ \ustwo{using e.g.\ transition sequences with labels from $\{A,B\}^{k-1}$}.


Now assume $\phi_l = \phi_{l1}\vee\cdots\vee \phi_{l n_l}$ , where $\phi_{l m}=X_i$ or $\phi_{l m}=\neg X_i$ ($m=1,\cdots,n_l$).
To represent $\phi_l$, we iterate the $\DF$ circuit $n_l-1$ times so that Defender can force the play from $((q_{(h+2)l},L), (q_{(h+2)l},R))$ into 
any of $((q_{(h+3) l m},L),( q_{(h+3) l m},R))$ for $m=1,\cdots, n_l$.


Finally, we need to handle the formulas $\phi_{lm}$. 
\begin{itemize}
\item If $\phi_{l m}=X_i$ we add
\[
\xymatrix{(q_{(h+3) l m},L) \ar[d]^{\ustwo{2i+1}} & (q_{(h+3) l m},R)\\
 \mathit{end} &&}
\]
\item If $\phi_{l m} =\neg X_i$ we add
\[
\xymatrix{(q_{(h+3) l m},L) \ar[d]^{\ustwo{2i}} & (q_{(h+3) l m},R)\\
 \mathit{end} &&}
\]
\end{itemize}
Note that the outgoing transitions are added only for states tagged with $L$. They give Attacker a chance to win  if $\phi_{lm}$ does \emph{not} hold after Defender's choices.


Overall the construction yields a winning strategy for Defender if and only if the given formula is true. 
\end{proof}
}

%% file: sakamoto.tex
\section{Language equivalence for $RA(S\#_0)$\label{sec:saka}}


The results of the previous section can be used to close an existing complexity gap for deterministic language equivalence of register automata. Recall that, in the non-deterministic case, language equivalence
(even universality) is undecidable~\cite{NSV04}. In the deterministic case, however, the problem can be solved in polynomial space. Sakamoto~\cite{S98} conjectured that the language inequivalence 
problem is not in NP.
Below we refute the conjecture, showing that, for RA($S\#_0$), the complexity of deterministic language inequivalence actually matches that of nonemptiness~\cite{SI00}.
\us{Because we discuss language equivalence, in this section we assume that RA($S\#_0$) are given as $\abra{Q,\Sigma,q_0, \rho_0,\delta,F}$, where $q_0\in Q$ is the initial state,
$\rho_0$ is an initial register assignment conforming to the $S\#_0$ policy, and $F\subseteq Q$ is a set of accepting states.}

We call an $r$-RA($S\#_0$) $\calA$ \emph{deterministic} if, for all states $q$ of $\calA$:
\begin{enumerate}[(i)]
\item for all $(t,i)\in\Sigma\times[1,r]$ there is at most one transition of the form $q\xr{t,i}q'$, and
\item for all $t\in\Sigma$ there is at most one transition of the form $q\xr{t,i\fre}q'$ \us{for $i\in [0,r]$}.
\end{enumerate}
On the other hand, an LTS is {deterministic} if, for all $\kappa\in\conf$ and $\ell\in\Act$, there is at most one transition $\kappa\xr{\ell}\kappa'$.
Note that if $\calA$ is deterministic then so is its transition system $\calS(\calA)$.\footnote{The converse may fail due to transitions of $\calA$ {not being fireable} in $\mathcal{S}(\calA)$.}
Then, from Proposition~\ref{p:bstrategy}, one obtains the following.
\begin{lem}\label{lem:dra-word}
Let $\calA_i=\abra{Q_i,\Sigma,q_{0i},\rho_{0i},\delta_i,F_i}$ be a deterministic $r_i$-RA($S\#_0$) ($i=1,2$), $r=\max(r_1,r_2)$ and $N=|Q_1|+|Q_2|$.
If $\calL(\calA_1)\neq\calL(\calA_2)$ then there is some
$w\in(\calL(\calA_1)\cup\calL(\calA_2))\setminus(\calL(\calA_1)\cap\calL(\calA_2))$
with $|w|\in O(rN^2+r^3N)$.
\end{lem}
\begin{proof}
We view $\calA_1,\calA_2$ as $r$-RA($S\#_0$)s with some unused registers and consider the $r$-RA($S\#_0$) 
\[
\calA=\abra{Q_1\uplus Q_2\uplus\{q_0\}\uplus\{q_s\},q_0,\Sigma,\{(i,\#)\ |\ i\in[1,r]\},\delta_1\cup\delta_2\cup\delta_s\cup\delta_s'\cup\delta_F,\emptyset},
\]
where $q_0$ is a ``blind'' initial state, 
$q_s$ is a sink state,
$\delta_s=\{q\xr{t,i}q_s\ |\ \delta(q)\upharpoonright(t,i)=\nolinebreak\emptyset\}\cup\{q\xr{t,1\fre}q_s\ |\ \delta(q)\upharpoonright(t,i\fre)=\emptyset\}$ adds any missing outgoing transitions to $\delta=\delta_1\cup\delta_2$,
$\delta'_s=\{q_s\xr{t,i}q_s\ |\ t\in\Sigma\land i\in[1,r]\}\cup
\{q_s\xr{t,1\fre}q_s\ |\ t\in\Sigma\}$ is a set of sink transitions,
and $\delta_F=\{q\xr{t_F,i}q_0\ |\ i\in[1,r]\land q\in F_1\cup F_1 \}\cup\{q\xr{t_F,1\fre}q_0\ |\ q\in F_1\cup F_2\}$
is a set of ``final'' transitions for some newly introduced constant $t_F$.

Assume WLOG that $\calL(\calA_1)\not\subseteq\calL(\calA_2)$. Then, there is some transition path for $\calA_1$ from $(q_{01},\rho_{01})$ to some $q_1\in F_1$ that, when simulated by $\calA_2$ from $(q_{02},\rho_{02})$, does not lead in $F_2$. For $\calA$, this means that $(q_{01},\rho_{01})$ and $(q_{02},\rho_{02})$ are not bisimilar: Attacker can lead the game to a configuration pair $((q_1,\rho_1),(q_2,\rho_2))$, with $q_2\in (Q_2\setminus F_2)\cup\{q_s\}$, where he wins by playing some $(t_F,a)$ from $(q_1,\rho_1)$. By Proposition~\ref{p:bstrategy}, Attacker has some strategy $\calT$ of depth $O(rN^2+r^3N)$ for winning the same game. We observe that, because $\calA$ is  saturated with sink transitions, the latter can only be achieved by Attacker being able to play a final transition with label $(t_F,a)$ in one part of the game. Suppose the happens in the part starting from $(q_{01},\rho_{01})$ and let $w\,(t_F,a)$ be the string accepted by the corresponding transition path, so $w\in\calL(\calA_1)$. By determinacy of $\calA_2$, $w\notin\calL(\calA_2)$. 
\end{proof}

\begin{thm}
Language inequivalence for deterministic RA($S\#_0$) is NP-complete.
\end{thm}
\begin{proof}
Membership in NP is achieved via Lemma~\ref{lem:dra-word}.
NP-hardness follows from NP-completeness of language non-emptiness for 
deterministic RA($S\#_0$)~\cite{SI00}.
\end{proof}


%% file: npmagic.tex

\section{NP bound for single assignment with filled registers (RA($SF$))}\label{sec:npmagic}

In Section~\ref{s:SI} we showed, in the setting with single assignment and no erasures
(denoted by RA($S\#_0$)) the bisimilarity problem was solvable in polynomial space. Here we show that a further improvement is possible in the RA($SF$) case, i.e.\ 
if the registers are required to be filled from the very start. We shall show an NP upper bound.

We start off with a series of results aiming to identify succinct (polynomial-size) sets of generators for $\sims$, which we shall call \emph{generating systems}.
In Section~\ref{s:SI} we already found that parts of $\sims$ exhibit group-theoretic structure.
Namely, Lemma~\ref{lem:group} shows that, for any $p\in Q$ and $S\subseteq [1,r]$, $\clg{G}_S^p{(\sims)}=\{ \sigma\cap (X_S^p\times X_S^p)\,|\, (p,S)\sims_\sigma (p,S)\}$ is a group, where $X_S^p(\sims)\subseteq S$ is the characteristic set of $(p,S)$. 

Note that, for  RA($SF$), we only have the case $S=[1,r]$.  
Furthermore, $\sims$ will be the only closed relation that we shall consider.
For these reasons, we write simply $X^p$ for characteristic set $X^p_{[1,r]}(\sims)$ and $\clg{G}^p$ for group $\clg{G}^p_{[1,r]}(\sims)$.

The group-theoretic structure implies that ${\clg{G}^p}$ can be generated by linearly many generators with respect to $r$.

\begin{lemC}[\cite{MN87}]\label{lem:subgen}
\us{Every subgroup of $\sym{n}$  has a  generating set with at most $\max(2,\lfloor{\frac{n}{2}}\rfloor)$ elements.}
\end{lemC}
To handle the more general case $(p,S)\sims_\sigma (q,S)$ of different states, consider 
\[
{\clg{K}^{p,q}=\{ \sigma\cap (X^p\times X^q)\,\,|\,\, (p,[1,r])\sims_\sigma (q,[1,r])\}.}
\]
Observe that, for $\sigma_1,\sigma_2\in {\clg{K}^{p,q}}$, we have
$\sigma_2=(\sigma_2;\sigma_1^{-1}); \sigma_1$,
because $\sigma_1^{-1};\sigma_1=\id{{X^q}}$. Moreover, $\sigma_2;\sigma_1^{-1}\in {\clg{G}^p}$, {so $\sigma_2$ has been obtained from $\sigma_1$ and an element of $\clg{G}^p$}.
Consequently, in presence of generators of ${\clg{G}^p}$, one member of ${\clg{K}^{p,q}}$ suffices
to generate the whole of ${\clg{K}^{p,q}}$ by composition. This observation motivates the following definition of a generating system.

\begin{defi}
A \boldemph{generating system} ${\clg{G}}$ consists of:
\begin{itemize}[\;$\bullet$]
\item a partitioning of $Q$ into $P_1, \cdots, P_k$;
\item for each partition $P_i$, a single representative $p_i\in P_i$ {and:}
\begin{itemize}
\item a characteristic set ${X^{p_i}}\subseteq {[1,r]}$;
\item a set ${G^{p_i}}$, of up to $\max(2,\lfloor{\frac{{r}}{2}}\rfloor)$ permutations $\sigma\in\sym{{X^{p_i}}}$;
\item for each $q\in P_i\setminus\{p_i\}$, a partial permutation $\ray{p_i}{q}\in\is{{[1,r]}}$ such that
 $\dom{\ray{p_i}{q}}={X^{p_i}}$; for technical convenience, we also add $\ray{p_i}{p_i}= \id{X^{p_i}}$.
 \end{itemize}
 \end{itemize}
 We write $\reps{{\clg{G}}}$ for the set $\{p_1,\cdots, p_k\}$ of representatives.
 \end{defi}

A generating system is used to generate a relation $\gen{{\clg{G}}}\subseteq (Q\times\{{[1,r]}\}\times\is{r}\times Q\times \{{[1,r]}\})$ as follows. 
First, set
 \[\begin{aligned}
\base{{\clg{G}}} &=
 \{ (p_i,{[1,r]},\sigma,p_i,{[1,r]})\,|\, p_i\in\reps{{\clg{G}}}, \sigma\in {G^{p_i}}\} \\
&\;\quad\cup\, 
\{(p_i,{[1,r]},\ray{p_i}{q}, q, {[1,r]})\,|\, p_i\in \reps{{\clg{G}}}, q\in P_i\}
\end{aligned} \]
 and then take $\gen{{\clg{G}}}=\Cl{\base{{\clg{G}}}}$.

\begin{lem}\label{lem:genex}
There exists a generating system ${\clg{G}}$ such that
$\gen{{\clg{G}}}=\ \sims$.
\end{lem}
\begin{proof}
We partition $Q$ into equivalence classes defined by: $p\sim q$ if and only if there exists
$\sigma$ such that $(p,{[1,r]},\sigma,q,{[1,r]}) \in {\sims}$.
For each equivalence class $P_i$, we pick a single member $p_i$ arbitrarily and let ${G^{p_i}}$ consist
of  the generators of ${\clg{G}^{p_i}}$ provided by Lemma~\ref{lem:subgen}.
Consider $q\in P_i\setminus\{ p_i\}$. Because $q \in P_i$, there exists $\sigma$ such that $(p_i,{[1,r]},\sigma,q,{[1,r]})\in {\sims}$.
Then we can take $\ray{p_i}{q}=\sigma\cap ({X^{p_i}}\times[1,r])$.
By the previous discussion, this delivers the sought generating system.
\end{proof}
\begin{lem}\label{lem:polymem}
For any generating system ${\clg{G}}$, membership in $\gen{{\clg{G}}}$ can be determined in polynomial time.
\end{lem}
\begin{proof}
To determine whether $(q_1,{[1,r]},\sigma, q_2,{[1,r]})\in \gen{{\clg{G}}}$, we proceed as follows.
If $q_1,q_2$ belong to different partitions we return NO.
Suppose $q_1,q_2\in P_i$. Recall that $\base{{\clg{G}}}$ contains $(p_i,{[1,r]}, \ray{p_i}{q_j}, q_j,{[1,r]})$ with $\dom{\ray{p_i}{q_j}}=X^{p_i}$.
Then $(q_1,{[1,r]},\sigma,q_2,{[1,r]})\in \gen{{\clg{G}}}$ is equivalent to $(p_i, {[1,r]},\sigma', p_i,{[1,r]})\in \gen{{\clg{G}}}$, where $\sigma'=\ray{p_i}{q_1};\sigma;(\ray{p_i}{q_2})^{-1}$.
This  is in turn equivalent to $\sigma'\cap(X^{p_i}\times X^{p_i})$ being generated from permutations in ${G^{p_i}}$.
That the latter problem is solvable in polynomial time is a well-known result in computational group theory~\cite{FHL80}.
\end{proof}

\begin{thm}\label{t:SFisNP}
\RAbs{SF} is in NP.
\end{thm}
\begin{proof}
First we guess a generating system ${\clg{G}}$ and verify whether 
$\gen{{\clg{G}}}$ is a bisimulation. By Lemma~\ref{lem:genex}, there exists at least one generating system with this property. Because generating systems involve polynomially many
components of polynomial size, they can be guessed in polynomial time. 
Next, in order to check whether the guessed generating system generates a bisimulation, we need to verify the \SyS\ conditions (for $S_1=S_2=[1,r]$) for {each of the polynomially many elements of $\base{{\clg{G}}}$.} Note that this will involve polynomially many membership tests for 
$\gen{{\clg{G}}}$, each of which can be performed in polynomial time by 
Lemma~\ref{lem:polymem}. If the guess leads to a non-bisimulation, we return NO.
Otherwise, we use another membership test for $\gen{{\clg{G}}}$ to check whether
the given instance of the bisimilarity problem belongs to $\gen{{\clg{G}}}$. We return the outcome of that test as the final result.
\end{proof}

\begin{rem}
Note that symbolic bisimulations are based on \emph{partial finite permutations}, which form 
inverse semigroups. Consequently, inverse semigroup-theoretic structure could seem the most natural kind of structure with which to approach our problems.
\cutout{
Note that Lemma~\ref{lem:polymem} crucially relies on results from group theory: the availability of polynomially many generators for any subgroup of $\sym{n}$ and
the associated membership problem being solvable in polynomial time.}
Unfortunately, inverse semigroups do not admit analogous results.
\begin{itemize}
\item There exist inverse subsemigroups of $\is{n}$ that require 
${n \choose \frac{n}{2}}\approx 2^n \sqrt{\frac{2}{{\pi n}}}$ generators, e.g. $\{\id{X} | X\!\subseteq\! [1,n],|X|\!=\!\frac{n}{2}\}$.
\item {It is possible to} show that the membership problem for inverse subsemigroups of $\is{n}$ is PSPACE-complete, {sharpening a result of Kozen \cite{K77}}.
We present the argument in Appendix~\ref{apx-pspace}.
\end{itemize}
Consequently, we were forced to look a bit deeper{,} and base generating systems on groups.
\end{rem}

\begin{rem}
Note that we do not have a matching lower bound for RA($SF$),
which raises the intriguing prospect that  there may still be scope for improvement in this case.
\ustwo{A closely related problem to $\sim$-RA($SF$) is \emph{graph automorphism (GA)}, i.e.\ given a graph ${G}$ decide whether it has a non-trivial automorphism. While it is easy to see that GA is in NP, it is not known whether it is in P or, for that matter, in coNP. We can reduce graph-automorphism to the following problem in our setting: given a DRA($SF$) $\mathcal A$ (without locally fresh transitions) and a configuration $(q,\rho)$, is there a non-identity permutation $\pi$ such that $(q,\rho)\sim(q,\rho\circ\pi)$? This observation introduces a possible barrier to methods we can pursue to efficiently solve $\sim$-RA($SF$), such as partition refinement, which aim to construct a representation of the whole bisimilarity relation.}
\end{rem}


%% file: frash.tex

\newcommand\myP[2][]{\,P_{#2}^{#1}\,}
\newcommand\rank{\mathsf{rank}}
\newcommand\symb{\mathsf{symb}}
\newcommand\pplus{\textrm{++}}
\newcommand\Pfin{\calP_{\mathsf{fin}}}
\newcommand\FSyS{\textsc{(FSyS)}}
\newcommand\xtimes{{\times}\,}
\newcommand\fmea[1]{\hat{\gamma}(#1)}
\newcommand\simih[2][]{\overset{#2}{\sim}^{\text{\raisebox{-5pt}{$#1$}}}}

\newcommand\simihn[3]{(\overset{#1}{\sim})^{#2}_{#3}}

\section{Fresh-register automata with single assignment without erasure (FRA($S\#_0$))}\label{sec:frash}

In this section we examine the problems tackled in Sections~\ref{sec:pspace-np-bounds}-\ref{sec:npmagic} albeit in the general case of FRAs. 
We would like to apply the same techniques, aiming to produce the same upper bounds, yet the FRA setting raises {significant} additional challenges.
Our approach for RAs relied on symbolic bisimulations and the group-theoretic structure that emanated from them. 
While we can express bisimilarity in FRAs symbolically {following}~\cite{Tze11}, we shall see that {such} symbolic bisimulations do not support the group-theoretic representations.
The reason is the treatment of the history of the computation, which affects 
{bisimilarity}
in subtle ways, especially in the initial stages of the bisimulation game. 
In those stages, global  and local freshness can inter-simulate another, under certain conditions, which leads us to extending our symbolic representations beyond the $r$ names that each {system} can have in its registers.

\subsubsection*{Simplified notation}
We extend the simplified notation for RA($S\#_0$) by including transition labels for global freshness.
Recall that, in any transition $q_1\xr{t,X,i,Z}q_2$ of an $r$-FRA($S\#_0$), we have that 
\us{$X\in\{\sfre,\emptyset\}\cup\{\{j\}\,|\, j\in [1,r]\}$, $Z=\emptyset$ and $X=\{j\}$ implies $i=0$}. 
We thus follow a simpler notation for transitions, with \us{$\delta\subseteq Q\times\Sigma\times([1,r]\cup \{i\fre,i\gfre\ |\ i\in[0,r]\,\})\times Q$}:
\begin{enumerate}[\;\;(a)]
\item we write each transition $q_1\xr{t,\{i\},0,\emptyset}q_2$ as $q_1\xr{t,i}q_2$;
\item and each $q_1\xr{t,\emptyset,i,\emptyset}q_2$ as $q_1\xr{t,i\fre}q_2$;
\item and each $q_1\xr{t,\circledast,i,\emptyset}q_2$ as \ $q_1\xr{t,i\gfre}q_2$.
\end{enumerate}
(a),(b) are as in RA($S\#_0$). 
In (c), the automaton reads $(t,a)$ if $a$ is \emph{globally fresh}, i.e.\ it has not appeared in the history so far, and stores it in register $i$.
{Formally, $q\xrightarrow{t,i\gfre}q'$ can induce a transition $(q,\rho,H)\xrightarrow{t,a}(q',\rho[i\mapsto a],H\cup\{a\})$ just if  $a\notin H$.\footnote{The latter condition above is slightly different but equivalent to that used in~\cite{Tze11}. In \emph{loc.\,cit.}, the names of $\rho$ are not necessarily included in $H$ and hence in this rule one stipulates that $a\notin\rng{\rho}\cup H$.}}

\us{
  \subsubsection*{Assignment pre-updates}
  Recall the operations we introduced in Section~\ref{sec:updates} on partial bijections and in particular the pre-composing of generalised swaps (i.e.\ $\xsw{i}{j}$ with $i,j\in[0,r]$). We extend this operation to register assignments by setting:
  \[
    \rho\xsw{i}{j} = \begin{cases}
\sw{i}{j};\rho & \text{if }i,j\in[1,r]\\ \rho &\text{otherwise}
      \end{cases}
    \]
    We can then show the following.

    \begin{lem}\label{lem:updates2}
      Given $r$-register assignments $\rho_1,\rho_2$ (of $S$-type) and $i,i',j,j'\in[0,r]$:
      \begin{enumerate}
      \item
$\xsw{j}{j'}(\rho_1;\rho_2^{-1})\xsw{i}{i'} = \rho_1\xsw{i}{i'};\rho_2\xsw{j}{j'}^{-1}$;
        \item for any $a\in\D$ such that
  \
    $a\in\rng{\rho_1}\implies a=\rho_1(i')$ \ and
   \  $a\in\rng{\rho_2}\implies a=\rho_2(j')$ \
  we have \
  $\xsw{j}{j'}((\rho_1;\rho_2^{-1})[i'\mapsto j'])\xsw{i}{i'} = \rho_1[i'\mapsto a]\xsw{i}{i'};\rho_2[j'\mapsto a]\xsw{j}{j'}^{-1}$.
  \end{enumerate}
\end{lem}
  }


\subsection{Symbolic bisimulation}

Recall that, in the case of RAs, we were able to capture bisimilarity symbolically by using \us{tuples of the form $(q_1,S_1,\sigma,q_2,S_2)$, whereby $S_k$ represented $\dom{\rho_k}$ of the actual configuration $(q_k,\rho_k)$ being represented (for $k=1,2$)}, and partial bijection $\sigma:S_1\to S_2$ captured the matching names of $\rho_1$ and $\rho_2$. Moving to FRAs, the first obstacle we face is that actual configurations contain the full history of names and have therefore unbounded size. For bisimulation purposes, though, keeping track of the whole history, or its size, is not necessary. In fact, history only plays a role in globally fresh transitions and one can easily see that the rule
\begin{quote}
``Every globally fresh transition from $q_1$ must be matched by a globally or a locally fresh transition from $q_2$.''
\end{quote}
is sound for simulation of globally fresh transitions.

\us{However, global freshness leads to complications in the simulation of locally fresh transitions. For example, consider configurations $(q_1,\rho_1,H),(q_2,\rho_2,H)$ with
$H=\{d_1,d_2\}$ and a transition $q_1\xr{t,1\fre}q_1'$. We look at three scenarios:
\begin{enumerate}
  \item If
    $\rng{\rho_1}=\{d_1,d_2\}$, then the transition from $q_1$ can be matched by some $q_2\xr{t,1\gfre}q_2'$, as the local names of $q_1$ coincide with all the names in $H$.
    \item 
      If $\rng{\rho_1}=\{d_1\}$ and $\rho_2=\{(1,d_2)\}$,  then the transition from $q_1$ cannot be matched by some $q_2\xr{t,1\gfre}q_2'$ alone, unless there is also a transition $q_2\xr{t,1}q_2''$ (to capture the fact that $q_1\xr{t,1\fre}q_1'$ can produce $d_2$).
      \item 
        On the other hand, if $\rng{\rho_1}=\rng{\rho_2}=\{d_1\}$  then $q_2$ must use a locally fresh transition in order to match the transition from $q_1$ (as the latter can produce $d_2$).
\end{enumerate}
More generally, 
if $|H|>2r$ then there will be some $d\in H\setminus(\rng{\rho_1}\cup\rng{\rho_2})$, which makes impossible for locally fresh transitions in one system to be matched by globally fresh transitions in the other one.}

Thus, under certain circumstances {which include the fact that} $|H|\le2r$, local freshness can be captured by global freshness and some known-name transitions.
To accommodate this feature, we will design symbolic bisimulations with an additional component $h\in[0,2r]\cup\{\infty\}$ that will abstract the size of $|H|$. The value $h=\infty$ will signify that $|H|>2r$ and therefore local-fresh cannot be matched by global-fresh. On the other hand, $h\le2r$ will mean that $|H|=h\le2r$ and therefore extra cases need to be considered for fresh transitions.
For $h\le2r$, we will consider symbolic configurations $(q_i,S_i)$ ($i=1,2$) where $S_i\subseteq[1,3r]$ and $h=|S_i|$, related by bijections $\sigma:S_1\to S_2$.
\begin{itemize}
\item The component $S_i\cap[1,r]$ of $S_i$ will still represent the domain of $\rho_i$. 
\item The complementary part $S_i\setminus[1,r]$ will represent the remaining names, those that have passed but no longer reside in $\rho_i$ (i.e.\ $H\setminus\rng{\rho_i})$, in some canonical fashion. 
\end{itemize}
Effectively, the above will allow us to symbolically represent the history of each FRA, up to the size $2r$, in an ordered way.
It will also offer us a way to decide the simulation game for {locally fresh transitions.  Let us suppose that one system performs a transition} $q_1\xr{t,i\fre}q_1'$:
\begin{enumerate}[1.]
\item Such a transition can capture any name $d$ that is represented in some $i'\in S_1\setminus[1,r]$.
If $\sigma(i')\in[1,r]$ then
{the other system}
has the name in its registers and can (only) capture it by some 
$q_2\xr{t,\sigma(i')}q_2'$.
\item If $\sigma(i')\in S_2\setminus[1,r]$ then 
{the name is historical and the other system} does not currently have it in its registers. It is therefore obliged to simulate by some locally fresh transition $q_2\xr{t,j\fre}q_2'$.
\item The transition can also capture any name $d$ that is not in $H$ and, in this case,
{the other system}
can capture it by any $q_2\xr{t,j\fre/j\gfre}q_2'$. Moreover, such a simulation step would increase the size of $h$ by one.
\end{enumerate}
We therefore formulate symbolic bisimulation as follows.

\begin{defi}
Let $\mathcal{A}=\abra{Q,\Sigma,\delta}$ be an $r$-FRA($S\#_0$). 
We first set:
\[\begin{aligned}
\calU_0 &= Q\xtimes \calP([1,3r])\xtimes \is{3r}\xtimes Q\xtimes \calP([1,3r])\xtimes ([0,2r]{\cup}\{\infty\})\\[1mm]
\calU &=\{ (q_1,S_1,\sigma,q_2,S_2,h)\in\calU_0\ |\ 
\sigma\subseteq S_1\times S_2\land
(h\leq 2r\implies|\sigma|=|S_1|=|S_2|=h)\\
&\quad\qquad\qquad\quad\qquad\qquad\quad\qquad\qquad\quad\qquad\,{}\land (h=\infty\implies \sigma\in\is{r}\land  S_1,S_2\subseteq[1,r])\}
\end{aligned}\]
A \emph{symbolic simulation} on $\calA$ is a relation 
$R \subseteq\calU$,
with membership
$(q_1,S_1,\sigma,q_2,S_2,h)\in R$ often written $(q_1,S_1)\myRR[h]{\sigma}(q_2,S_2)$, such that
all $(q_1,S_1,\sigma,q_2,S_2,h)\in R$
satisfy the following \emph{fresh symbolic simulation conditions} \FSyS:\footnote{We say that 
  \emph{$(q_1,S_1,\sigma,q_2,S_2,h)$ satisfies the \FSyS\ conditions in $R$}.}\textsuperscript{,}\footnote{\us{Note how the \FSyS\ conditions are divided with respect to the value of $h$: conditions (a2), (b1), (b2), (c1) and (c2) all require $h\le 2r$; while conditions (a3), (b3) and (c3) are for $h=\infty$. On the other hand, (a1) applies to all $h$.}}
%
\begin{enumerate}[(a)]
\item for all $q_1\xr{t,i}q_1'$,
\begin{enumerate}[1.]
\item if $\sigma(i)\in[1,r]$ then there is $q_2\xr{t,\sigma(i)}q_2'$ with $(q_1',S_1)\myRR[h]{\sigma}(q_2',S_2)$,
\item if $\sigma(i)=j'\in[r{+}1,3r]$ then there is $q_2\xr{t,j\fre}q_2'$ with $(q_1',S_1)\myRR[h]{\xsw{j}{j'}\sigma}(q_2',S_2\xsw{j}{j'})$,
\item if $i\in S_1\setminus\dom{\sigma}$ then there is $q_2\xr{t,j\fre}q_2'$ with $(q_1',S_1)\myRR[h]{\sigma[i\mapsto j]}(q_2',S_2[j])$;
\end{enumerate}
\item for all $q_1\xr{t,i\fre}q_1'$, $i'\in S_1\setminus[1,r]$ and $j\in S_2\setminus\rng{\sigma}$,
\begin{enumerate}[1.]
\item
if  $\sigma(i')\in[1,r]$ then {there is}
$q_2\xr{t,\sigma(i')}q_2'$ with $(q_1',S_1\xsw{i}{i'})\myRR[h]{\sigma\xsw{i}{i'}}(q_2',S_2)$,
\item if $\sigma(i')=j'\in[r{+}1,3r]$ then there is 
$q_2\xr{t,j\fre}q_2'$ with 
\[
(q_1',S_1\xsw{i}{i'})\myRR[h]{\xsw{j}{j'}\sigma\xsw{i}{i'}}(q_2',S_2\xsw{j}{j'}),
\]
\item 
{there exists} $q_2\xr{t,j}q_2'$ with\ $(q_1',S_1[i])\myRR[h]{\sigma[i\mapsto j]}(q_2',S_2)$;
\end{enumerate}
\item for all $q_1\xr{t,\ell_{{1}}}q_1'$ with $\ell_{{1}}\in\{i\fre,i\gfre\}$ there is some $q_2\xr{t,\ell_{{2}}}q_2'$ with $\ell_{{2}}\in\{j\fre,j\gfre\}$ and,
\begin{enumerate}[1.]
\item if $h<2r$ then,
taking $i'=\min([r{+}1,3r]\setminus S_1)$ and $j'=\min([r{+}1,3r]\setminus S_2)$, we have\\
$(q_1',S_1[i']\xsw{i}{i'})\myRR[h+1]{\xsw{i}{i'}({\sigma[i'\mapsto j']})\xsw{j}{j'}}(q_2',S_2[j']\xsw{j}{j'})$;
\item if $h=2r$ then
$(q_1',S_1[i]\cap[1,r])\myRR[\infty]{{\sigma[i\mapsto j]}\cap{[1,r]^2}}(q_2',S_2[j]\cap[1,r])$;
\item if $h=\infty$ then $(q_1',S_1[i])\myRR[\infty]{\sigma[i\mapsto j]}(q_2',S_2[j])$ and if $\ell_{{1}}=i\fre$ then $\ell_{{2}}=j\fre$.
\end{enumerate}
\end{enumerate}
Define the inverse of $R$ by:
\[
R^{-1} = \{\,(q_2,S_2,\sigma^{-1},q_1,S_1,{h})\ |\ (q_1,S_1,\sigma,q_2,S_2,{h})\in R\,\}
\]
and call $R$ a \boldemph{symbolic bisimulation} if both $R$ and $R^{-1}$ are symbolic simulations. 
We let \emph{s-bisimilarity}, denoted $\sims$, be the union of all symbolic bisimulations.
\\
We define a sequence of \boldemph{indexed bisimilarity} relations ${\simih{i}} \subseteq\calU$ inductively as follows.
We let $\simih{0}$ be the whole of $\calU$.  Then, for all $i \in \omega$ and $h\in[0,2r]\cup\{\infty\}$,
 $(q_1,S_1)\, \simihn{i+1}{h}{\tau}\, (q_2,S_2)$ just if both $(q_1,S_1,\tau,q_2,S_2,h)$ and $(q_2,S_2,\tau^{-1},q_1,S_1,h)$  satisfy the \FSyS\ conditions in $\simih{i}$.
\end{defi}\smallskip

Let $\kappa_i=(q_i,\rho_i,H)$ ($i=1,2$) be configurations with common history $H$ and let $n=|H|$.
Their symbolic representation will depend on $n$. We take 
$\symb(\kappa_1,\kappa_2)\subseteq\calU$ to be:
\[
  \symb(\kappa_1,\kappa_2)=\begin{cases}
\{(q_1,\dom{\hat\rho_1},\hat\rho_1;\hat\rho_2^{-1}\!,q_2,\dom{\hat\rho_2},n)\in\calU\ |\ \theta(\hat{\rho}_1,\hat{\rho_2})\} &n\le 2r\\
\{(q_1,\dom{\rho_1},\rho_1;\rho_2^{-1}\!,q_2,\dom{\rho_2},\infty)\} & n>2r
\end{cases}\]
{%
where $\theta(\hat{\rho}_1,\hat{\rho}_2)$ is the condition stipulating that $\hat\rho_i$ range over all $3r$-register assignments of type $S\#_0$ such that $\rng{\hat\rho_i}=H$ and $\hat\rho_i\upharpoonright[1,r]=\rho_i$, for $i=1,2$.}
In particular, $\symb(\kappa_1,\kappa_2)$ is singleton in case $n>2r$ but not necessarily so if $n\le2r$. The following lemma ensures that,
 with respect to bisimilarity, the specific choice of element from 
$\symb(\kappa_1,\kappa_2)$ is not important.

\begin{lem}\label{l:symbchoice}
For all $\kappa_1,\kappa_2$ as above, if $|H|<2r$ then either $\symb(\kappa_1,\kappa_2)\subseteq{\sims}$ or
$\symb(\kappa_1,\kappa_2)\cap{\sims}=\emptyset$. 
\end{lem}

\begin{defi}
We say that $\kappa_1$ and $\kappa_2$ are \emph{s-bisimilar}, written $\kappa_1\sims\kappa_2$, if $\symb(\kappa_1,\kappa_2)\subseteq{\sims}$.
\end{defi}

\begin{rem}
The definition of symbolic bisimulation we give here is substantially more fine-grained than the one in~\cite{Tze11}. Although in \emph{loc.\,cit.} the symbolic bisimulation is {also} given parametrically to the size of the history $h$ (up to the given bound\footnote{In fact, the bound used in~\cite{Tze11} is smaller ($2r{-}1$), due to the fact that it examines bisimulation between configurations with common initial names.}), for $h\le 2r$ that formulation is simplistic in that it only keeps track of names that reside in registers of the automata,\footnote{that is, in $(q_1,S_1)\myRR[h]{\sigma}(q_2,S_2)$ we always have $S_1,S_2\subseteq[1,r]$.}
which in turn
prohibits us to derive
$(q_1,S_1)\myRR[h]{\sigma_1;\sigma_2}(q_3,S_3)$ from $(q_1,S_1)\myRR[h]{\sigma_1}(q_2,S_2)$ and $(q_2,S_2)\myRR[h]{\sigma_2}(q_3,S_3)$ and apply the group-theoretic approach.
\end{rem}

\ustwo{Using the intuition described above about the bounded representation of histories, we can show the following correspondence. Similarly to Lemma~\ref{l:sym1}, the proof of the next lemma is based on matching concrete and symbolic bisimulations and doing a careful, if somewhat tedious, case analysis of possible transitions in each case.}

\begin{lem}\label{l:sym1FRA}Let $\kappa_1$ and $\kappa_2$ be configurations of an $r$-FRA($S\#_0$). Then $\kappa_1\sim\kappa_2\iff\kappa_1\sims\kappa_2$.  
\end{lem}

\begin{lem}\label{lem:ftwiddle}
For all $i\in\omega$, ${\simih{i+1}}\subseteq{\simih{i}}$ and $(\bigcap_{i\in\omega}\simih{i})= {\sims}$.
\end{lem}

Similarly to symbolic bisimulations for RA($S\#_0$), we have the following closure properties.
\cutout{
Our next aim is to show that $\sims$ is closed under composition and extension of partial permutations.
The latter will enable us to represent it succinctly by appropriate choices of representatives.%
Given $S_1,S_2\subseteq[1,3r]$ and $\sigma,\sigma'\in\is{3r}$ we write $\sigma\leq_{S_1,S_2}\sigma'$ just if:
\[
\sigma\subseteq\sigma'\subseteq S_1\times S_2
\]
Moreover, for any 
$R\subseteq\calU$,
we define its \boldemph{closure}  $\Cl{R}$  to be
the smallest relation $P$ containing $R$ and closed under the following rules.
 \begin{gather*}
 \frac{(q_1,S_1,\sigma,h, q_2, S_2)\in R'}{ (q_2,S_2,\sigma^{-1},h, q_1, S_1)\in P}\;(\textsc{Sym}) \\[1mm]
 \frac{(q_1,S_1,\sigma,h, q_2, S_2)\in P\qquad \sigma\le_{S_1,S_2} \sigma'}{ (q_1,S_1,\sigma',h, q_2, S_2)\in P}\;(\textsc{Ext})
\\[1mm]
 \frac{(q_1,S_1,\sigma_1,h,q_2, S_2)\in P\qquad (q_2,S_2,\sigma_2,h, q_3, S_3)\in P}{(q_1, S_1, \sigma_1;\sigma_2,h, q_3, S_3)\in P}\;(\textsc{Tr})
\end{gather*}
}%
Given $R\subseteq\calU$ we split $R$ {into} \emph{components}:
\[
R = \sum\nolimits_{h\in[0,2r]\cup\{\infty\}}R^h
\]
where $R^h=\{(q_1,S_1,\sigma,q_2,S_2)\ |\ (q_1,S_1,\sigma,q_2,S_2,h)\in R\}$. 
We now write $\Cl{R}$ for the componentwise closure of $R$ with respect to {identity}, symmetry, transitivity and extension of partial permutations, i.e.\ 
$\Cl{R} = \sum\nolimits_{h\in[0,2r]\cup\{\infty\}}\Cl{R^h}$.

\newcommand\Rhat{\hat{R}}
\newcommand\Phat{\hat{P}}
\newcommand\myRhat[2][]{\Rhat^{#1}_{#2}}
\newcommand\myPhat[2][]{\Phat^{#1}_{#2}}

The following lemma will play a key role in the forthcoming technical development.
\ustwo{It is proved similarly to Lemma~\ref{l:sym2}, i.e.\ by showing that the \FSyS\ rules are compatible with the closure rules. While the proof is longer, there is no essential novelty: the approach is similar, only the case analysis required is more extensive.}

\begin{lem}\label{l:fsym2}
Let $R,P\subseteq\calU$. If all $\us{g\in R\cup R^{-1}}$ satisfy the \FSyS\ conditions in $P$ then all $g\in\Cl{R}$ satisfy the \FSyS\ conditions in $\Cl{P}$.
\end{lem}

\begin{prop}\label{p:closuresFRA}
Symbolic bisimilarity and indexed symbolic bisimilarity for FRA($S\#_0$) are closed.
\begin{enumerate}
\item $\Cl{\sims}={\sims}$\,;\quad
\item for all $i\in\omega$: ${\simi{i}} = \Cl{\simi{i}}$.
\end{enumerate}
\end{prop}
\begin{proof}
For $\Cl{\sims}={\sims}$, 
since $\sims$ is symmetric and satisfies the \FSyS\ conditions in itself, 
from the previous lemma we have that $\Cl{\sims}$ satisfies the \FSyS\ conditions in itself and is therefore a symbolic bisimulation. Thus, $\Cl{\sims}\subseteq{\sims}$.

{For $\Cl{\simi{i}}={\simi{i}}$ we
do induction on $i$.  
When $i = 0$ then the result follows from the fact that $\simi{0}$ is the universal relation.
For the inductive case, note first that $\simi{i+1}$ is symmetric by construction and all $g\in{\simi{i+1}}$ satisfy the \FSyS\ conditions in $\simi{i}$. Hence, by Lemma~\ref{l:fsym2}, all elements of $\Cl{\simi{i+1}}$ satisfy the \FSyS\ conditions in $\Cl{\simi{i}}$. By IH, $\Cl{\simi{i}}={\simi{i}}$ so ${\Cl{\simi{i+1}}}\subseteq{\simi{i+1}}$, as required.}
\end{proof}
More explicitly, {the last part of Proposition~\ref{p:closuresFRA}} means that, given 
$(q_1,S_1)~\simihn{i}{h}{\tau}~(q_2,S_2)$:
\begin{enumerate}
\item Then, $(q_2,S_2)~\simihn{i}{h}{\tau^{-1}}~(q_1,S_1)$.
\item For all $\tau'$, if $\tau\leq_{S_1,S_2}\tau'$ then 
$(q_1,S_1)~\simihn{i}{h}{\tau'}~(q_2,S_2)$.
\item For all $(q_2,S_2)~\simihn{i}{h}{\tau'}~(q_3,S_3)$,
$(q_1,S_1)~\simihn{i}{h}{\tau;\tau'}~(q_3,S_3)$.
\end{enumerate}

We therefore observe that the extension of symbolic representations to the size $3r$, and the ensuing history representation up to size $2r$ along with the extended symbolic bisimulation conditions, have paid off in yielding the desired closure properties.
The group-theoretic behaviour of {a closed relation $R$} 
differs between different components:
\begin{itemize}
\item  \cutout{$\sims^\infty$}{$R^\infty$} has the same structure as the {closed relations $R$}\cutout{$\sims$} examined in Section~\ref{sec:groups}.
\item For $h\in[0,2r]$, the tuples $(q_1,S_1,\sigma,q_2,S_2)\in{R^h}$\cutout{{\sims^h}} respect the condition $|S_1|=|S_2|=|\sigma|=h$. In particular, $\sigma$ is a bijection from $S_1$ to $S_2$ and, hence, in this case closure under extension is trivial, and so are characteristic sets ($X^p_S({R^h})=S$). Moreover, $\sigma\in\is{3r}$ and $S_1,S_2\subseteq[1,3r]$.
\end{itemize}
We can hence see that the same groups arise as in the case of RA($S\#_0$), and actually simpler in the case $h\in[0,2r]$, albeit parameterised over $h$.
This allows for a similar group-theoretic treatment.


\subsection{PSPACE bound for bisimulation game}


Before we come to the proof of the main result, recall Theorem~\ref{thm:babai} which says that, for $n \geq 2$, the length of every subgroup chain in $\sym{[1,n]}$ is at most $2n-3$.

\begin{lem}\label{lem:fplay-length}\us{
Let $h\in[0,2r]\cup\{\infty\}$, $S_1,S_2 \subseteq [1,3r]$ and $\calU_{S_1,S_2}^h = Q \times \{S_1\} \times \is{r} \times Q \times \{S_2\}\times\{h\}$.
Then the sub-chain $\{ \simih{i} \,|\, (\simi{i+1} \cap\ \calU_{S_1,S_2}^h) \subsetneq (\simi{i} \cap\ \calU_{S_1,S_2}^h)\}$ has size {$O(|Q|^2 + r^2|Q|)$}.	}
\end{lem}
\begin{proof}
\us{We argue that $\{\simi{i} \,|\, (\simi{i+1} \cap\ \calU_{S_1,S_2}^h) \subsetneq (\simi{i} \cap\ \calU_{S_1,S_2}^h)\}$ has at most
$|Q|^2 + r^2|Q| - 2r|Q|$ elements.
We shall say that $(q_1,S_1,h,q_2,S_2)$ is \emph{separated} in $\simi{i}$ if there is no $\sigma$ such that $(q_1,S_1)~\simihn{i}{h}{\sigma}~(q_2,S_2)$; we say it is \emph{unseparated} otherwise.
We claim that if $(\simi{i+1} \cap\ \calU_{S_1,S_2}^h) \subsetneq (\simi{i} \cap\ \calU_{S_1,S_2}^h)$ then there is some $q \in Q$ and $S \in \{S_1,S_2\}$ such that 
\begin{enumerate}[(i)]
\item either $X^q_S({\simihn{i}{h}{}})\subsetneq X^q_S({\simihn{i+1}{h}{}})$
\item or $\calG^q_S({\simih[h]{i+1}})$ is a strict subgroup of $\calG^q_S({\simih[h]{i}})$
\item or there is a tuple $(q_1,S_1,h,q_2,S_2)$ that is unseparated in $\simi{i}$ and becomes separated in $\simi{i+1}$.
\end{enumerate}
We  reason  as follows.  
If $(\simi{i+1} \cap\ \calU_{S_1,S_2}^h) \subsetneq (\simi{i} \cap\ \calU_{S_1,S_2}^h)$ then there are $q_1,q_2 \in Q$ and $\sigma$ such that $(q_1,S_1)~\simihn{i}{h}{\sigma}~(q_2,S_2)$ 
but not $(q_1,S_1)~\simihn{i+1}{h}{\sigma}~(q_2,S_2)$.  
From closure properties for $\simi{i}, \simi{i+1}$ it follows that
$(q_1,S_1)~\simihn{i}{h}{\sigma'}~(q_2,S_2)$ and 
not $(q_1,S_1)~\simihn{i+1}{h}{\sigma'}~(q_2,S_2)$, 
where $\sigma' = \sigma \cap (X^{q_1}_{S_1}({\simihn{i}{h}{}}) \times X^{q_2}_{S_2}({\simihn{i}{h}{}}))$.
Consequently, we can assume wlog that $\dom{\sigma} = X^{q_1}_{S_1}({\simihn{i}{h}{}})$ and $\rng{\sigma} = X^{q_2}_{S_2}({\simihn{i}{h}{}})$.
Now, suppose that, for all $q \in Q$, $S \in \{S_1,S_2\}$, we have $X^q_S({\simihn{i+1}{h}{}}) = X^q_S({\simihn{i}{h}{}})$ (i.e. not (i))
and that no previously unseparated tuple becomes separated in $\simi{i+1} \cap\ \calU_{S_1,S_2}^h$ (i.e. not (iii)).  
From the latter, It follows that there is some $\tau$ such that $(q_1,S_1)~\simihn{i+1}{h}{\tau}~(q_2,S_2)$.
Hence,  $\sigma;\tau^{-1} \in \calG^{q_1}_{S_1}(h,i)$ but $\sigma;\tau^{-1} \notin \calG^{q_1}_{S_1}(h,i+1)$ so that $\calG^{q_1}_{S_1}({\simihn{i}{h}{}}) > \calG^{q_1}_{S_1}({\simihn{i+1}{h}{}})$.

Because $X^q_S({\simihn{i}{h}{}})\subseteq X^q_S({\simihn{i+1}{h}{}})$, (i) may happen at most $2r|Q|$ times  in the whole chain.
For fixed $X^q_S({\simihn{i}{h}{}})$, by Theorem \ref{thm:babai},
(ii) may happen at most $2r-2$ times (we include the case $r=1$), which gives an upper bound of $2 r |Q| (2r-2)$ for the number of such changes inside the whole chain
(under the assumption that the changes are not of type (i), which have already been counted).
Finally, the remaining changes must be of type (iii) and may happen at most $|Q|^2$ times across the whole chain.
Overall, we obtain $2r |Q|+2 r|Q| (2r-2)+|Q|^2=|Q|^2 + 4r^2|Q| - 2r|Q|$ as a bound on the length of the given chain.}
\end{proof}

Given $S_1,S_2\subseteq[1,3r]$ and $h\in[0,2r]\cup\{\infty\}$, let us call the triple $(S_1,S_2,h)$ \boldemph{proper} just if: either $|S_1|=|S_2|=h$, or $h=\infty$ and $S_1,S_2\subseteq[1,r]$.
For such $(S_1,S_2,h)$, let us define: 
\[
\fmea{S_1,S_2,h}=
\begin{cases}
\gamma(S_1\cap[1,r],S_2\cap[1,r])+h &\text{if }h\in[0,2r]\\
\gamma(S_1,S_2)+2r+1 &\text{if }h=\infty
\end{cases}
\]
{The measure $\hat\gamma$ enables us to show the following bound for stabilising indexed bisimulation, proven similarly to Lemma~\ref{lem:bound}.}

\begin{lem}\label{lem:fbound}
\us{Let $\ell$ be the bound from Lemma~\ref{lem:fplay-length} and $B=(4r+2) \ell$.
For any proper $(S_1,S_2,h)$, we have ${\simih{B}} \cap \calU_{S_1,S_2}^h = {\sims}\cap \calU_{S_1,S_2}^h$.}
\end{lem}
\begin{proof}
\us{Observe that $0\le\fmea{S_1,S_2,h}\le 4r+1$.
For each $m\in [0,4r+1]$, let
\[
k_m = \min \{i \,\,|\,\, {\simi{i}}\cap\uni[h]{S_1,S_2} ={\sims}\cap\uni[h]{S_1,S_2} \textrm{ for any $S_1, S_2, h$ with $\fmea{S_1,S_2,h}\ge m$}\}.
\]
Consider $S_1,S_2,h$ with $\fmea{S_1,S_2,h}\ge m$, where $m<4r+1$.

Observe that, for $k\geq k_{m+1}$,  if
$\simi{k}\cap\ \uni[h]{S_1,S_2}=\ \simi{k+1}\cap\ \uni[h]{S_1,S_2}$, then we must have
$\simi{k}\cap\ \uni[h]{S_1,S_2}=\ \sims\cap\ \uni[h]{S_1,S_2}$, because {the
 \FSyS\ conditions for $(S_1,S_2,h)$ refer to either $(S_1,S_2,h)$ or $(S_1',S_2',h')$ 
 with $\fmea{S_1',S_2',h'}>\fmea{S_1,S_2,h}$.}
Consequently, if $\simi{k}\cap\ \uni[h]{S_1,S_2} \neq\ \sims\cap\ \uni[h]{S_1,S_2}$,
the sequence $(\simi{k}\cap\ \uni[h]{S_1,S_2})$ ($k=k_{m+1},k_{m+1}+1,\cdots$) must change in every step before stabilisation. By Lemma~\ref{lem:fplay-length}, 
at most  $\ell$ extra steps from $\simi{k_{m+1}}$ will be required to arrive at 
$\sims \cap\ \uni[h]{S_1,S_2}$, which implies $k_m \le k_{m+1}+\ell$.
By a similar argument, we can conclude that $k_{4r+1}\le \ell$.
Consequently, $k_0 \le (4r+2) \ell$, as required.}
\end{proof}

We can therefore establish solvability in polynomial space.

\begin{prop}\label{p:fbstrategy}
For any FRA($S\#_0$) bisimulation problem, if there is a winning strategy for Attacker then there is one of depth $O(r|Q|^2+r^3|Q|)$.
\end{prop}

\begin{prop}\label{fpsolv}
\FRAbs{S\#_0} is in PSPACE.
\end{prop}
\cutout{
\begin{proof}
In Remark~\ref{rem:tofinite} we established that bisimilarity for RA($M\#$) can be reduced  to the finite-alphabet case at the cost of prolonging the bisimulation game by a constant factor. Consequently, the polynomial bound from the preceding Proposition (for RA($S\#_0$))
is also valid after the reduction to the finite-alphabet case.

Thanks to the bound, it suffices to play the corresponding bisimulation games for polynomially many steps. The existence of a winning strategy can then be established by an alternating Turing machine in polynomial time. The PSPACE bounds then follows from APTIME\,$=$\,PSPACE.
\end{proof}
}

\subsection{Generating systems and NP routines}

We proceed to generating systems for FRA($SF$), which are $h$-parameterised versions of the ones for RA($SF$), {except that now they are built over $[1,3r]$ rather than $[1,r]$}.
{Since we again consider only characteristic sets and groups with relation parameter $R = \mathord{\sims}$, we will typically leave this argument implicit in what follows.}
We call a pair $(S,h)$ proper just if $(S,S,h)$ is proper.

\begin{defi}\label{def:gen-sys-for-twoFRA}
A \boldemph{generating system} $\clg{G}_{S,h}$ for proper $(S,h)$ {(in which case $|S|\le2r$)}, consists of:
\begin{itemize}[\;$\bullet$]
\item a partitioning of $Q$ into $P_1, \cdots, P_k$;
\item for each partition $P_i$, a single representative $p_i\in P_i$ and:
\begin{itemize}
\item a characteristic set $X_{S,h}^{p_i}\subseteq S$;
\item a set $G_{S,h}^{p_i}$, of up to $\max(2,\!\sr{r})$ permutations $\sigma\in\sym{X_{S,h}^{p_i}}$\!;
\item for each $q\in P_i\setminus\{p_i\}$, a partial permutation $\ray{p_i}{q}\in\is{S}$ such that
 $\dom{\ray{p_i}{q}}=X_{S,h}^{p_i}$; for technical convenience, we also add $\ray{p_i}{p_i}= \id{X^{p_i}_{S,h}}$.
 \end{itemize}
 \end{itemize}
 We write $\reps{\clg{G}_{S,h}}$ for the set $\{p_1,\cdots, p_k\}$ of representatives.
\\
From $\clg{G}_{S,h}$ we generate $\gen{\clg{G}_{S,h}}\subseteq (Q\times\{S\}\times\is{3r}\times Q\times \{S\})$ by 
setting
 \[\begin{aligned}%
\base{\clg{G}_{S,h}}& =
 \{ (p_i,S,\sigma,p_i,S)\,|\, p_i\in\reps{\clg{G}_{S,h}}\land \sigma\in G_{S,h}^{p_i}\} \\
&\;\quad\cup\, 
\{(p_i,S,\ray{p_i}{q}, q, S)\,|\, p_i\in \reps{\clg{G}_{S,h}}\land q\in P_i\}
\end{aligned} \]
and taking $\gen{\clg{G}_{S,h}}=\Cl{\base{\clg{G}_{S,h}}}$.
\cutout{
A generating system for a proper $(S_1,S_2,h)$ is a triple $(\clg{G}_1, \clg{G}_{2}, \Phi)$, where $\clg{G}_{i}$ is a generating system for $(S_i,h)$ ($i=1,2$)
and $\Phi$ is a set of \emph{arcs}: 
\[
\Phi\;\subseteq\; \reps{\clg{G}_1}\times\{S_1\}\times(S_1\overset{\cong}{\rightharpoonup}S_2)\times \reps{\clg{G}_2}\times\{S_2\}
\]
such that, for all $q_1'\!\in\!\reps{\clg{G}_1}, q_2'\!\in\!\reps{\clg{G}_2}$, there is at most one $(q_1',S_1,\sigma',q_2',S_2)\in\Phi$.
}
\end{defi}

The following lemma, proved in the same way as Lemmata~\ref{lem:genex} and~\ref{lem:polymem}, enables us to prove an NP upper bound for bisimilarity in FRA($SF$).

\begin{lem}\label{lem:fgenex}\label{lem:fpolymem}
\begin{enumerate}
\item
For any proper $(S,h)$ there exists a generating system $\clg{G}_{S,h}$ such that
$\gen{\clg{G}_{S,h}}= {\sims}\cap\uni[h]{S,S}$.
\item
For any generating system $\clg{G}_{S,h}$, membership in $\gen{\clg{G}_{S,h}}$ can be determined in polynomial time.
\end{enumerate}
\end{lem}
\cutout{
\begin{proof}
To determine whether $(q_1,S,\sigma, q_2,S)\in \gen{\clg{G}_S}$, we proceed as follows.
If $q_1,q_2$ belong to different partitions we return NO.
Suppose $q_1,q_2\in P_i$. Recall that $\base{\clg{G}_S}$ contains $(p_i,S, \ray{p_i}{q_j}, q_j,S)$ with $\dom{\ray{p_i}{q_j}}=X_{p_i}$.
Then $(q_1,S,\sigma,q_2,S)\in \gen{\clg{G}_S}$ is equivalent to $(p_i, S,\sigma', p_i,S)\in \gen{\clg{G}_S}$, where $\sigma'=\ray{p_i}{q_1};\sigma;(\ray{p_i}{q_2})^{-1}$.
This  is in turn equivalent to $\sigma'\cap(X_{p_i}\times X_{p_i})$ being generated from permutations in $G^{p_i}_S$.
That the latter problem is solvable in polynomial time is a well-known result in computational group theory~\cite{FHL80}.
\end{proof}
}
\begin{thm}
\FRAbs{SF} is  in NP.
\end{thm}
\begin{proof}
Given an input tuple $(q_1,S_1,\sigma,q_2,S_2,h^0)$, note first that $[1,r]\subseteq S_1,S_2$ (by $F$) and $|S_1|=|S_2|$.
We can therefore convert it to an {equivalent}
$(q_1,S_1',\sigma',q_2,S_2,h^0)$, with $S_1'=S_2$, by applying a permutation on the indices in $S_1\setminus[1,r]$. Hence, we can assume wlog that our input is some $(q_1,S^0,\sigma,q_2,S^0,h^0)$.
Moreover, because the expansion of $S$ in the symbolic bisimulation game (when $h\in[0,2r]$) always occurs in its first free register ($\min([r{+}1,3r]\setminus S)$), we can compute the sequence $(S^0,h^0,S^0),(S^1,h^0{+}1,S^1),\cdots$ of distinct triples considered in the game (in the $h\in[0,2r]$ phase), which must thence be bounded in length by $2r$.
Including the final bisimulation phase ($h=\infty$), this gives us $2r+1$ phases. 
We first generate for each of them a generating system, say $\clg{G}_{S^i,h^i}$, and then verify whether 
each $\gen{\clg{G}_{S^i,h^i}}$ is a symbolic bisimulation, similarly to Theorem~\ref{t:SFisNP}.
Note that each such check can be achieved in polynomial time.
If the guess leads to some $\gen{\clg{G}_{S^i,h^i}}$ being a non-symbolic-bisimulation, we return NO.
Otherwise, we use another membership test for $\gen{\clg{G}_{S^0,h^0}}$ to check whether
the given instance of the bisimilarity problem belongs to $\gen{\clg{G}_{S^0,h^0}}$. We return the outcome of that test as the final result.
\end{proof}

%% file: vpdra.tex
\newcommand\bra[3]{ #1\triangleleft #2\triangleright #3}
\newcommand\eot{\mathsf{end}}
\newcommand{\myparaa}[1]{\vspace{1.25mm}\noindent\emph{#1}.}

\section{Visibly pushdown automata with single assignment and filled registers (VPDRA($SF$))}\label{sec:vpdra}

Finally, we consider a variant of register automata
with visible pushdown storage~\cite{AM04}. We only consider the most restrictive register discipline ($SF$), as undecidability 
will be shown to apply already in this case.
\begin{defi}
A {\boldemph{visibly pushdown $r$-register automaton}} ($r$-VPDRA($SF$)) $\clg{A}$
is a tuple 
\[
\abra{Q,\Sigma_C,\Sigma_N,\Sigma_R,{\Gamma},\delta},\] 
where:  
\begin{itemize}
\item $Q$ is a finite set of states;
\item $\Sigma_C$, $\Sigma_N$, $\Sigma_R$ are disjoint finite sets of
 \emph{push-}, \emph{no-op-} and \emph{pop-tags} respectively;
\item $\Gamma$ is a finite set of \emph{stack tags};
\item $\delta = \delta_C\cup \delta_N\cup \delta_R$, the transitions,
have ${\oprr}=\{1,\ldots,r\}\cup\{1^\bullet,\ldots, r^\bullet\}$ and:
\smallskip
\begin{itemize}[\quad $\circ$\ ]
\item
$\delta_C\subseteq  Q\times \Sigma_C\times {\oprr}\times\Gamma\times\{1,\cdots,r\}\times Q$
\item $\delta_N \subseteq Q\times\Sigma_N \times {\oprr}\times Q$
\item $\delta_R \subseteq Q\times\Sigma_R \times {\oprr}\times \Gamma\times \{1,\cdots,r,\bullet\}\times Q$
\end{itemize}
\cutout{\[\begin{array}{rcl}
\delta_C&\subseteq & Q\times \Sigma_C\times {\oprr}\times\Gamma\times\{1,\cdots,r\}\times Q\\
\delta_N &\subseteq &Q\times\Sigma_N \times {\oprr}\times Q\\
\delta_R & \subseteq & Q\times\Sigma_R \times {\oprr}\times \Gamma\times \{1,\cdots,r,\bullet\}\times Q
\end{array}\]}
\end{itemize}
Configurations of $r$-VPDRA($SF$) are triples $(q,\rho,s)$, where
$q\in Q$, $\rho$ is a register assignment and $s\in(\Gamma\times {\D})^\ast$ is the stack.
An LTS arises by having a labelled edge
$
(q_1,\rho_1, s_1) \stackrel{(t,d)}{\longrightarrow} 
(q_2,\rho_2,s_2)
$
just if there exist $i\in[1, r]$ and $l\in\{i,i^\bullet\}$ such that:
\begin{enumerate}
\item $\rho_1(x)=\rho_2(x)$ for all $x\neq i$;
\item if $l=i$ then $\rho_1(i)=\rho_2(i)$, otherwise {$\rho_2(i)\not\in\rng{\rho_1}$};
\end{enumerate}
and (iii) one of the following conditions holds:
\begin{itemize}
\item $(q_1,t,l,t',j,q_2)\in \delta_C$ and $s_2=(t',\rho_2(j))s_1$,
\item $(q_1,t,l,q_2)\in \delta_N$ {and $s_2 = s_1$},
\item $(q_1,t,l,t',j,q_2)\in\delta_R$, $s_1=(t',d')s_2$,
\end{itemize}
where if $j\in[1,r]$ then $d'=\rho_2(j)$, otherwise $d'\not\in\rng{\rho_2}$.
\end{defi}

We show that even the visibly pushdown with $SF$ register discipline is undecidable.
To do so, we reduce from the undecidable emptiness problem for {(one-way)} \emph{universal}
register automata with two registers (\uratwo)~\cite{DL09}.

\begin{defiC}[\cite{DL09}]
A one-way universal $n$-register automaton (URA$_n$) is a tuple $\abra{\Sigma,Q,q_I, n,\delta}$ such that $\Sigma$ is a finite alphabet, 
$Q$ is a finite set of states, $q_I\in Q$ is the initial state and $\delta:Q\rightarrow \Delta(\Sigma,Q,n)$ 
is the transition function, where 
\[
\begin{array}{rcl}
\Delta(\Sigma,Q,n)&=& \{\, \bot,\,\, \top,\,\,q\wedge q',\,\, \bra{q}{\beta}{q'}, \,\, X q,\,\, \overline{X} q,\,\, \downarrow_r q \\
&| &\,\,q,q'\in Q,\,\, r\in\{1,\cdots,n\},\,\, \beta\in B(\Sigma,n)\,\,\}\\[2mm]
B(\Sigma,n)&=& \{ a, \eot\}\cup\{ \uparrow_r\,|\, r\in\{1,\cdots,n\}\}
\end{array}
\]
\end{defiC}
The emptiness problem for \uratwo\ is undecidable~\cite{DL09}. We shall reduce it to bisimilarity testing.
We first sketch the argument and then later give all the details.

Given a \uratwo\ $U$, we shall devise a $2$-VPDRA $\clg{A}_U$ with
two configurations $\kappa_1, \kappa_2$ such that $U$ accepts a word iff $\kappa_1\not\sim \kappa_2$.
$\clg{A}_{U}$ is constructed to induce a bisimulation game in which Attacker  gets a chance to choose
a word to be accepted by $U$ and simulate an accepting run (if one exists).
It consists of two nearly identical components, which are linked by the Defender Forcing circuit in places.
Other differences between 
them stem from the need to arrange for non-bisimilarity, in cases when the bisimulation game
reaches a stage indicating acceptance or Attacker tried to cheat while simulating a run.
We sketch  the design of the components.

\myparaa{Input stage} Initially, we want Attacker to start choosing input letters and pushing them on the stack.
This is to continue until Attacker decides to finish the input phase. Defender will simply copy the moves
in other component.
Technically, both kinds of choices can be implemented by deterministic push transitions that
cover the range of input in both components.  Observe that, in order to win (uncover non-bisimilarity), Attacker will 
eventually need to  abandon the input stage to avoid infinite copying.

\myparaa{Transitions} Once the input phase is over, the automaton enters the simulation stage.
Recall that the input word chosen by Attacker will be available on the stack in both components. 
The top of the stack will play the role of the head of $U$ and we can use the two registers
of $\clg{A}_U$ to emulate the two registers  of $U$.
To make transitions, we need to be able to access the tag at the top of the stack
as well as compare the corresponding data value with the content of registers.
The only way of inspecting the top of the stack is by popping, but then we could lose
the data value if it does not already occur in a register (the value might be needed
later, e.g. the automaton might want to move it into a register). To avoid such a loss,
we will let Attacker guess the outcome of the comparisons. However, Defender
will be allowed to verify the correctness of such guesses (via Defender Forcing).
During the verification the top of the stack will indeed be popped, but we shall
be no longer concerned about losing it, because it will survive {in} a different branch
of the game, which will carry on simulating the run. In order to implement the 
detection of incorrect guesses, 
we will need to break symmetry between the components and arrange for non-bisimilarity
if Attacker's guess is correct.

\myparaa{Universal states} To simulate these, we can delegate the choice to Defender
through Forcing. This will allow Defender to direct the game towards a failing branch, if one exists.

\myparaa{Head movements} To advance the tape, we simply use one of the pop-instructions.

\myparaa{Register reassignment} To move the currently scanned data value into a register, let us assume that the symbol is not in a register yet.
Then we can refresh the content of the relevant register (to guess the data value at the top of the stack) and 
then perform a pop. Note that a wrong guess by Attacker will lead to a deadlock (no ability to pop),
which gives Attacker the necessary incentive to guess correctly.

\myparaa{Accepting/rejecting states} If the simulation reaches a rejecting state, we arrange for bisimilarity (to attract Defender there). 
In accepting states, we arrange non-bisimilarity.

\cutout{
\uratwo\ are equipped with two registers and an input
tape, which moves forward after an explicit instruction. 
While the head
is scanning the current input symbol, the automaton can compare it
with its register content and branch accordingly. Consequently, 
an input symbol can be read and processed without being stored in registers.}
\cutout{
\footnote{This does make a difference.
As we show in the Appendix, universality for $2$-register machines in the style of~\cite{NSV04} (without a separate input tape) 
can be reduced to alternating $1$-register machines of~\cite{DL09}, which implies decidability.}}
\begin{thm}\label{thm:bisim}
VPDRA(SF) bisimiliarity is undecidable.
\end{thm}

\cutout{
\begin{proof}[Proof (sketch)]
Given a \uratwo\ $U$, we devise a $2$-VPDRA $\clg{A}_U$ with
two configurations $\kappa_1, \kappa_2$ such that $U$ accepts a word iff $\kappa_1\not\sim \kappa_2$.
$\clg{A}_{U}$ is constructed to induce a bisimulation game in which Attacker  gets a chance to choose
a word to be accepted by $U$ and simulate an accepting run (if one exists).
The stack of $\clg{A}_U$ is used to store the word that Attacker has chosen,
with the top of the stack playing the role of the head of $U$ and the two registers
of $\clg{A}_U$ emulating the two registers of $U$.
To simulate a transition we arrange for Attacker to guess the outcome of the comparison of the top of stack with the current register contents whilst allowing Defender to verify the correctness of such guesses via Defender forcing.
Transitions from universal states are chosen by Defender, again using Defender forcing.
\end{proof}}

\begin{proof}
Given a \uratwo\ $U=\abra{\Sigma,Q,q_I, 2,\delta}$, we shall construct a $2$-VRPDA $\clg{A}_{U}$ 
such that $\kappa_1\sim \kappa_2$ if and only if ${U}$ does not accept any input,
where $\kappa_j=(\init^j, {\tau_I}, \epsilon)$ ($j=1,2$) and $\init^1,\init^2$ are states.
$\clg{A}_{U}$ will be constructed so as to induce a bisimulation game in which Attacker gets a chance 
to choose a word to be accepted and simulate an accepting run (if one exists).
{Without loss of generality, we shall assume injectivity of register assignments
and that, whenever $\downarrow_r$ is used, the ${\D}$-value on the 
tape is not present in registers (these conditions can be enforced
by modifying the transition function with the help of the finite control and appropriate book-keeping). }
Moreover, to avoid complications with 
borderline cases, we shall assume that $U$ does not accept the empty word.

$\clg{A}_{U}$ will consist of two mostly identical components involving superscripted states from ${U}$
as well as a number of auxiliary states  implicit in the definitions below. The only connections between
the two components will be due to the use of the Defender Forcing circuit. The only differences between 
the components will stem from the need to arrange for non-bisimilarity, in cases when the bisimulation game
reaches a stage indicating acceptance or when Attacker makes a simulation mistake.

Below we explain the design of $\clg{A}_{U}$ at various stages of  simulating $U$. We use arrows to define
transitions according to the following conventions.
\begin{itemize}
\item $\xymatrix{q_1\ar[r]^{(t,l)/(t',j)}& q_2}$ stands for $(q_1,t,l,t',j,q_2)\in \delta_C$
\item $\xymatrix{q_1\ar[r]^{(t,l)}& q_2}$ stands for $(q_1,t,l,q_2)\in \delta_N$
\item $\xymatrix{q_1\ar[r]^{(t,l),(t',j)} &q_2}$ stands for $(q_1,t,l,t',j,q_2)\in \delta_R$
\end{itemize}
Given $q\in Q$, we write $q^j$ ($j=1,2$) for its superscripted variants to be included in $\clg{A}_{U}$.
We shall rely on the following sets of tags.
\[\begin{array}{rcl}
\Sigma_C &=&\{\bot \}+\Sigma\\
\Sigma_N&=&\{\ustwo{t_0}, t_1, t_2\}\\
\Sigma_R&=&\{t_R\}
\end{array}\]
For the stack alphabet, we shall have $\Gamma=\Sigma_C$.

\bigskip

We start off by introducing new states $\init^1, \init^2$ that will be used to start the initial phase in which
Attacker can choose an input word and push it on the stack.

\subsection*{Input Phase}

When drawing a diagram featuring states superscripted with $j$, we mean to say that 
\emph{two} copies of the design should be included into $\clg{A}_\clg{U}$, one for $j=1$ and another for $j=2$.
We use $\circ$, $\square,\triangle, \diamond,\odot$ to indicate auxiliary states to be included in each component.
We shall reuse them in different cases on the understanding that they refer to \emph{different} states in each case.
\pagebreak
\[\xymatrix{
&\init^j \ar[d]_{(\bot,1)/(\bot,1)}&\\
&\circ^j \ar[d]_{(\ustwo{t_0},1^\bullet)}     \ar@/^3mm/[dd]^{(a,1)/(a,1)}      \ar@/^20mm/[dd]^{(a,2)/(a,2)}& \\
&\square^j \ar[d]_{(a,1)/(a,1)} &\\
&\triangle^j \ar[d]_{(\ustwo{t_0},1^\bullet)}\ar@/^24mm/[uu]^{(\ustwo{t_0},1)}&\\
&\diamond^j\ar[d]_{(\ustwo{t_0},2^\bullet)}\\
& \odot^j\ar@{-->}[ld]\ar@{-->}[rd]&\\
\mathit{test}^j && q_I^j
}
\]
$a$ ranges over $\Sigma$ above. Consequently,
if the bisimulation game starts from $(\kappa_1,\kappa_2)$ then the above design gives Attacker 
a chance to pick a data word and push it on the stack.  
{The three outgoing transitions from state $\circ^j$ correspond to (from left to right) Attacker picking for the next data value: a fresh data value not currently in either register, the data value currently stored in register 1 or the data value currently stored in register 2.}
The stack content in both copies will be
the same. {Attacker also decides when to end the input selection phase and 
proceed to $(\diamond^1,\diamond^2)$. The transition sequence $(\ustwo{t_0},1^\bullet) (\ustwo{t_0},2^\bullet)$
is intended to give Attacker a chance to pick the right initial register assignment to support
the simulation. For a match with URA, we need the initial values to be different
from any data values present in the selected input word.
Once Attacker generates the values and $(\odot^1,\odot^2)$ is reached, Defender will have an option to challenge the choice 
or to proceed with the simulation to $(q_I^1, q_I^2)$. This will be achieved through Defender Forcing,
represented by dashed lines. We shall return to the exact design of $\mathit{test}^j$, 
after we apply Defender Forcing in simpler cases.}

The subsequent part of the construction corresponds to checking that the selected word
is accepted (we want Attacker to win {iff} this is the case).
We analyze each kind of transition in turn.

\subsection*{Transitions}

\subsubsection*{$\delta(q)=\bot$ (rejection)}

\[ q^j \]

We do not add any transitions from $q^1$ or $q^2$. This ensures bisimilarity,
should the game enter configurations with states $q^1, q^2$ respectively.

\subsubsection*{$\delta(q)=\top$ (acceptance)}

\[ \xymatrix{q^1\ar[d]^{(\ustwo{t_0},1)} \\ \circ^1} \qquad \xymatrix{q^2} \]

Note that we do not add any transitions from $q^2$ in order to generate non-bisimilar configurations.

\subsubsection*{$\delta(q)=q_1\wedge q_2$ (universal choice)}

We will let Defender choose the state ($q_1$ or $q_2$)
that should be pursued. Note that this is consistent with the goal of
relating emptiness with bisimilarity. To that end, we  use the Defender Forcing
circuit from Section \ref{sec:bisim-defns} (Figure~\ref{fig:dforcing}).
Recall that in order for the technique to work with {VPDRA}, 
we need to be sure that the stacks and registers are used in the same way by each of the components.
This is an easily verifiable property of our constructions. In order to implement DF
we need two different labels, e.g. $(t_1,1)$ and $(t_2,1)$.

For brevity, in what follows, we shall write 
\[
\xymatrix{&q\ar@{-->}[ld]\ar@{-->}[rd]&\\
q_1&& q_2}
\]
to refer to the use of $DF(q^1, q^2, \ustwo{(t_1,1),(t_1,1),(t_2,1)}, q_1^1, q_1^2, q_2^1, q_2^2)$.

\bigskip

\subsubsection*{$\delta(q)=\bra{q_1}{\beta}{q_2}$}

Here we shall let Attacker choose between $q_1$ and $q_2$ but
the Defender will later be able to challenge the decision (and check
whether it is consistent with $\beta$). For this purpose we use

\[
\xymatrix@C=1em{ & & &q^j \ar[lld]_{(t_1,1)}\ar[rrd]^{(t_2,1)}& & &\\
&\circ_L^j\ar@{-->}[ld]\ar@{-->}[rd]&& &&\circ_R^j\ar@{-->}[ld]\ar@{-->}[rd]&\\
\beta^j&& q_1^j& &{\neg\beta}^j&& q_2^j
}
\]
where $\beta^1, \beta^2, {\neg\beta}^1, {\neg\beta}^2$ will be constructed so that
the first two induce bisimilarity iff $\beta$ fails and the last two induce bisimilarity iff $\beta$ holds.
We do case analysis on $\beta$.

\mypara{$\beta=a$ (stack tag comparison)}
To handle $\beta^j$, we introduce
\[
\xymatrix{\beta^j \ar[d]^{(t_R,1),(a,2)}\ar@/_20mm/[d]_{(t_R,1),(a,1)}\ar@/^25mm/[d]^{(t_R,1),(a,\bullet)}\\
\circ^j}
\]
and
\[\xymatrix{\circ^1\ar[d]^{(\ustwo{t_0},1)}\\
\square^1}\]

{\noindent  We explain the idea behind this first gadget, the rest are similar.  
If Defender was correct to challenge Attacker because Attacker cheated, i.e. the letter under the head (top of stack) is not tagged by $a$ (despite Attacker's claim), then Attacker will not be able to play any transition from $\beta_j$ and hence Defender will win.  
If Defender challenged Attacker incorrectly, then Attacker will be able to play exactly one of the transitions, according to the current register assignment, and Defender will copy the move. 
However, in the following move Attacker will win, since Attacker will play the only transition out of $\circ^1$ and Defender cannot match this in $\circ^2$, since it has no available transitions.}

For $\neg\beta^j$ we can take
\[
\xymatrix{{\neg}\beta^j \ar[d]^{(t_R,1),(a',2)}\ar@/_20mm/[d]_{(t_R,1),(a',1)}\ar@/^25mm/[d]^{(t_R,1),(a',\bullet)}\\
\circ^j}
\]
where $a'$ ranges over $\Sigma\setminus\{a\}$, and:
\[\xymatrix{\circ^1\ar[d]^{(\ustwo{t_0},1)}\\
\square^1}\]

\mypara{$\beta=\uparrow_r$ (stack $\D$-value comparison)}
To handle $\beta^j$, we introduce
\[
\xymatrix{\beta^j \ar[d]^{(t_R,1),(a,r)}\\
\circ^j}
\]
with $a$ ranging over $\Sigma$,
and:
\[\xymatrix{\circ^1\ar[d]^{(\ustwo{t_0},1)}\\
\square^1}\]

For $\neg\beta^j$ we can take
\[
\xymatrix{{\neg}\beta^j \ar@/_10mm/[d]_{(t_R,1),(a,3-r)}\ar@/^10mm/[d]^{(t_R,1),(a,\bullet)}\\
\circ^j}
\]
with $a$ ranging over $\Sigma$, and
\[\xymatrix{\circ^1\ar[d]^{(\ustwo{t_0},1)}\\
\square^1}\]

\mypara{$\beta=\eot$ (last tape-symbol)}
To handle $\beta^j$, we introduce
\[
\xymatrix{\beta^j \ar[d]^{(t_R,1),(a,2)}\ar@/_20mm/[d]_{(t_R,1),(a,1)}\ar@/^25mm/[d]^{(t_R,1),(a,\bullet)}\\
\square^j}
\]
with $a$ ranging over $\Sigma$,
\[
\xymatrix{\square^j \ar[d]^{(t_R,1),(\bot,2)}\ar@/_20mm/[d]_{(t_R,1),(\bot,1)}\ar@/^25mm/[d]^{(t_R,1),(\bot,\bullet)}\\
\circ^j}
\]
and:
\[\xymatrix{\circ^1\ar[d]^{(\ustwo{t_0},1)}\\
\triangle^1}\]

For $\neg\beta^j$ we can take

\[
\xymatrix{{\neg}\beta^j \ar[d]^{(t_R,1),(a,2)}\ar@/_20mm/[d]_{(t_R,1),(a,1)}\ar@/^25mm/[d]^{(t_R,1),(a,\bullet)}\\
\square^j}
\]
with $a$ ranging over $\Sigma$,
\[
\xymatrix{\square^j \ar[d]^{(t_R,1),(a,2)}\ar@/_20mm/[d]_{(t_R,1),(a,1)}\ar@/^25mm/[d]^{(t_R,1),(a,\bullet)}\\
\circ^j}
\]
with $a$ ranging over $\Sigma$ again,
and:
\[\xymatrix{\circ^1\ar[d]^{(\ustwo{t_0},1)}\\
\triangle^1}\]

\subsubsection*{$\delta(q)=\downarrow_r q_1$}

We add
\[
\xymatrix{&q^j \ar[d]^{(\ustwo{t_0},r^\bullet)}&\\
&\square^j\ar@{-->}[ld]\ar@{-->}[rd] &\\
\beta^j && q_1^j}
\]
and also add outgoing transitions for $\beta^j$, as for the case $\beta=\uparrow_r$.
Note that the arrangement forces Attacker to guess the $\D$-value stored on top of the stack
(and place it in register $r$).

{
\subsubsection*{Freshness testing ($\mathit{test}^j$)}

We design $\mathit{test}^1$ and $\mathit{test}^2$ in such a way
that they will lead to non-bisimilarity iff Attacker guessed an initial register assignment
that does not contain any data values encountered during the input phase. $a$ ranges over $\Sigma$.
\[
\xymatrix{
\mathit{test}^1\ar@(ru,rd)^{(t_R,1),(a,\bullet)}\ar[d]_{(t_R,1),(\bot,\bullet)}\\
\diamond^1 }\qquad
\xymatrix{
\mathit{test}^2\ar@(ru,rd)^{(t_R,1),(a,\bullet)}
}
\]
}

\subsubsection*{$\delta(q)=X q_1$ (move head right/reject)} To take advantage of previous cases, we represent the
 transition as $\bra{u_1}{\eot}{u_2}$ with $\delta(u_1)=\bot$ and $\delta(u_2)=X q_1$.
This makes sure that $X$ is only invoked when we are not at the end of the word. Consequently,
we can reuse the previous constructions for $\bra{u_1}{\beta}{u_2}$ and $\bot$ cases. 
To handle $u_2$, we can now add
\[
\xymatrix{u_2^j \ar[d]^{(t_R,1),(a,2)}\ar@/_20mm/[d]_{(t_R,1),(a,1)}\ar@/^25mm/[d]^{(t_R,1),(a,\bullet)}\\
q_1^j}
\]
with $a$ ranging over $\Sigma$.

\subsubsection*{$\delta(q)=\overline{X} q_1$ (move head right/accept)}

This is nearly the same as the previous case: now we decompose the transition into $\bra{u_1}{\eot}{u_2}$ with $\delta(u_1)=\top$ and $\delta(u_2)=X q_1$.

\bigskip
Altogether, we obtain $\kappa_1\sim \kappa_2$ if and only if $U$ does not accept any words. This implies Theorem~\ref{thm:bisim}.
\end{proof}
 
The argument  above {also} reduces URA$_1$ emptiness to $1$-VPDRA,
which implies a non-primitive-recursive lower bound for $1$-VPDRA.



%% file: conclusion.tex

\section{Conclusion}

We have demonstrated bounds on the bisimilarity problem for broad classes of (fresh-)register 
automata, which include those studied in the literature.
The ability to start with empty registers, erase their contents (or equivalently, store duplicate values) and use of a stack all affect the inherent problem complexity.  
Global freshness, however, does not {seem to} affect complexity.
Except for the $SF$ discipline, all bounds are tight.

Although our problem formulation is with respect to two configurations of a single automaton, extending our results to problems concerning two automata is unproblematic.  
If the automata have different numbers of registers, the game can be played on an automaton with a number equal to the larger of the two, with additional registers initialised (and left) empty.  
Even in $F$ register disciplines our arguments show that, since these extra registers are never assigned to, the system can be treated as a $\#_0$ system without change in complexity.

%% file: apx-se-hardness.tex
\input{apx-framh}
\medskip

%% file: apx-rasihard.tex
\input{apx-twiddle}

%% file: apx-ram.tex

\newcommand\birng[2]{\rng{#1}\cup\rng{#2}}
\newcommand\trim[2][\rho_1,\rho_2]{\lceil  #2 \rceil_{#1}^N}
\newcommand\canrest[1][N]{\triangleright_{#1}}
\newcommand\perm{\mathsf{Perm}_\D}

\section{Proofs from Section~\ref{sec:ram}}

\us{
Given $d,d'\in\mathcal{D}$, let us write $(d\ d')$ for the bijection on $\D$ defined as:
\[
(d\ d')(x) = \begin{cases}
d' & \text{if }x=d\\
d & \text{if }x=d'\\
x & \text{otherwise}
\end{cases}
\]
{In what follows, we will consider various finite sets that involve elements of $\D$, e.g. finite subsets of $\D$, register assignments and tuples thereof.
Given such a finite set $X$, we write $(d\ d')\cdot X$ for the result of applying $(d\ d')$ recursively to the elements of $X$.
Put otherwise, $(d\ d')\cdot X$ will be $X$, where $d$ and $d'$ have been swapped.\footnote{Formally, this can be defined as an action of the group of permutations on a nominal set; see~\cite{Pit13} for a detailed exposition.}}
In particular,
\begin{itemize}
  \item if $X$ does not involve names, then $(d\ d')\cdot X=X$;
  \item if $X\subseteq \mathcal{D}$, then $(d\ d')\cdot X=\{(d\ d')(x)\mid x\in X\}$;
  \item if $X$ is some register assignment, then $(d\ d')\cdot X=\{(i,(d\ d')(X(i)))\mid X(i)\in\mathcal{D}\}\cup\{(i,\#)\mid X(i)=\#\}$;
    \item if $X$ is some tuple $(X_1,\dots,X_n)$ then $(d\ d')\cdot X=((d\ d')\cdot X_1,\dots,(d\ d')\cdot X_n)$.
    \end{itemize}
\us{Moreover, we shall consider finite name-permutations, i.e.\ ones taken from the set:
    \[
\perm = \{\, \pi:\D\xrightarrow{\cong}\D\mid \exists X\subseteq\D.\ X\text{ finite}\land \forall d\in\D\setminus X.\,\pi(d)=d\,\}
    \]}%
    and use $\pi$ to range over them. Each $\pi\in\perm$ can be decomposed as $\pi=(d_1\ d_1')\circ\dots\circ(d_n\ d_n')$, for some $n$ and $d_1,d_1',\dots,d_n,d_n'\in\D$. We then define $\pi\cdot X=(d_1\ d_1')\cdot\ldots\cdot(d_n\ d_n')\cdot X$.

Finally, given $\rho_1,\rho_2,H$ with ${\rng{\rho_{1}}\cup\rng{\rho_{2}}}\subseteq H\subseteq N$ or ${\rng{\rho_{1}}\cup\rng{\rho_{2}}}\subseteq N\subseteq H$, we extend the trim operation for $H$ as:
\[
  \lceil H \rceil_{\rho_1,\rho_2}^N   = \begin{cases} H & \text{ if }H\subsetneq N\\ {N}\setminus\{\min(N\setminus(\rng{\rho_1}\cup\rng{\rho_2})\} & \text{ if $N\subseteq H$}\end{cases}
\]
In either case, $\rng{\rho_{1}}\cup\rng{\rho_{2}}\subseteq \lceil H \rceil_{\rho_1,\rho_2}^N\subseteq N\cap H$ and $\lceil H \rceil_{\rho_1,\rho_2}^N\subsetneq N$.
Moreover, given $\rho_1,\rho_2,H,\hat H$, we say that \emph{$H$ can restrict to $(\rho_1,\rho_2,\hat H)$}, written $H\triangleright_N   (\rho_1,\rho_2,\hat H)$, if:
    \begin{itemize}
    \item ${\rng{\rho_{1}}\cup\rng{\rho_{2}}}\subseteq H\subsetneq N$ and $\hat H =  H$, or
    \item ${\rng{\rho_{1}}\cup\rng{\rho_{2}}}\subseteq N\subseteq  H$ and $\hat H=N\setminus\{d\}$ for some $d\in N\setminus(\rng{\rho_1}\cup\rng{\rho_2})$.
      \end{itemize}
Note that, in either case, $\rng{\rho_{1}}\cup\rng{\rho_{2}}\subseteq \hat H \subseteq N\cap  H$ and $\hat H\subsetneq N$. In particular, when $\lceil  H \rceil_{\rho_1,\rho_2}^N$ is well defined, 
we have $ H\triangleright_N(\rho_1,\rho_2,\lceil  H \rceil_{\rho_1,\rho_2}^N)$.

\begin{proof}[Proof of Lemma~\ref{l:Nbisim}]
    We show a correspondence between bisimulations and $N$-bisimulations from which the result follows.
%
%
    %
\medskip

\paragraph{bisim\,$\to$\,$N$-bisim}
Let $R$ be a bisimulation on $\calA$ that is closed in the following manner: for all permutations $\pi$,
if $ (q_1,\rho_1,H)\, R\, (q_2,\rho_2,H)$ then $(\pi\cdot(q_1,\rho_1, H))\,R\,(\pi\cdot(q_2,\rho_2,H))$.
We claim that the relation $\hat R\subseteq\mathbb{C}_{\calA,N}\times\mathbb{C}_{\calA,N}$, defined by
\[
\hat R = \{\,
((q_1,\rho_1,\hat H),(q_2,\rho_2,\hat H)) \mid \exists H.\ 
(q_1,\rho_1, H)\,R\,(q_2,\rho_2, H) \land
H\canrest(\rho_1,\rho_2,\hat H)
\,\},
\]
is an $N$-bisimulation.\\
Let $(q_1,\rho_1,\hat H)\,\hat R\,(q_2,\rho_2,\hat H)$, due to some 
$(q_1,\rho_1, H)\, R\, (q_2,\rho_2, H)$, and suppose $(q_1,\rho_1,\hat H)\xrightarrow{(t,d)}(q_1',\rho_1',\hat H')$ for some $t,d,q_1',\rho_1',\hat H'$. {Next we reason by case analysis.}
\begin{enumerate}[(a)]
\item Suppose $d\in\rng{\rho_1}\cup\rng{\rho_2}$. Then, $\hat H'=\hat H$ and, since $R$ is a bisimulation, we have   $(q_2,\rho_2, H)\xrightarrow{(t,d)}(q_2',\rho_2', H)$ for some $q_2',\rho_2'$ such that $(q_1',\rho_1', H)\, R\, (q_2',\rho_2', H)$. Consequently, $(q_2,\rho_2,\hat H)\xrightarrow{(t,d)}(q_2',\rho_2',\hat H)$.
It suffices to show that $(q_1',\rho_1',\hat H)\,\hat R\,(q_2',\rho_2',\hat H)$, i.e. $ H\canrest(\rho_1',\rho_2',\hat H)$. But this follows from 
$ H\canrest(\rho_1,\rho_2,\hat H)$ and  $\rng{\rho_1'}\cup\rng{\rho_2'}\subseteq\rng{\rho_1}\cup\rng{\rho_2}$.
\item Suppose $d=\min(\hat H\setminus(\rng{\rho_1}\cup\rng{\rho_2}))$.
  Then, again $\hat H'=\hat H$ and, reasoning as in the previous case,
  $(q_2,\rho_2, \hat H)\xrightarrow{(t,d)}(q_2',\rho_2',\hat H)$ for some $q_2',\rho_2'$ such that 
$(q_1',\rho_1', H)\,R\,(q_2',\rho_2', H)$. 
Since $\rng{\rho_1'}\cup\rng{\rho_2'}\subseteq\rng{\rho_1}\cup\rng{\rho_2}\cup\{d\}$, it follows that
$H\canrest(\rho_1',\rho_2',\hat H)$ and, thus,
$(q_1',\rho_1',\hat H)\,\hat R\,(q_2',\rho_2',\hat H)$, as required.
\item[(c1)] Suppose $d=\min(N\setminus \hat H)$ and $\hat H= H\subsetneq N$.
  Then, $\hat H'=\hat H\uplus\{d\}$ and, since $R$ is a bisimulation, we have 
  $(q_2,\rho_2,\hat H)\xrightarrow{(t,d)}(q_2',\rho_2',\hat H')$ for some $q_2',\rho_2'$ with $(q_1',\rho_1',\hat H')\,R\,(q_2',\rho_2',\hat H')$. Moreover, 
$\hat H'\canrest(\rho_1',\rho_2',\trim[\rho_1',\rho_2']{\hat H'})$
and, thus, $(q_1',\rho_1',\trim[\rho_1',\rho_2']{\hat H'})\,\hat R\,(q_2',\rho_2',\trim[\rho_1',\rho_2']{\hat H'})$.
\item[(c2)] Suppose $d=\min(N\setminus \hat H)$ and $\hat H= N\setminus\{\hat d\}$ for some $\hat d\in(N\setminus(\birng{\rho_1}{\rho_2})$, and $N\subseteq H$. Clearly, $d=\hat d$ and $\hat H'=N$.
Since $d$ is a fresh name for $\hat H$, the transition from $(q_1,\rho_1,\hat H)$ must be a globally fresh one, i.e. $\rho_1'=\rho_1[i\mapsto d]$. 
This implies that $(q_1,\rho_1,H)\xrightarrow{(t,d')} (q_1',\rho_1[i\mapsto d'],H\uplus\{d'\})$ for some fresh $d'$ and, therefore,
$(q_2,\rho_2,H)\xrightarrow{(t,d')}
(q_2',\rho_2[j\mapsto d'],H\uplus\{d'\})$ with $(q_1',\rho_1[i\mapsto d'],H\uplus\{d'\})\,R\,(q_2',\rho_2[j\mapsto d'],H\uplus\{d'\})$, for some $q_2',j$.
Moreover, $(q_2,\rho_2,\hat H)\xrightarrow{(t,d)}
(q_2',\rho_2[j\mapsto d],N)$. Let us set $\rho_2'=\rho_2[j\mapsto d]$.
By closure of $R$ under permutations of $\D$, we also have that
$(q_1',\rho_1',H\uplus\{d'\})\,R\,(q_2',\rho_2',H\uplus\{d'\})$,
so it suffices to show that $(H\uplus\{d'\})\canrest(\rho_1',\rho_2',\trim[\rho_1',\rho_2']{N})$,
which holds by definition.
\end{enumerate}

\paragraph{$N$-bisim\,$\to$\,bisim}
Let $R$ be an $N$-bisimulation on $\calA$. We claim that the relation $\hat R\subseteq\mathbb{C}_{\calA}\times\mathbb{C}_{\calA}$, defined by
\[
\hat R =  \{\,
  \pi\cdot((q_1,\rho_1,H),(q_2,\rho_2,H)) \mid {\pi\in\perm} \land\exists\hat H.
  (q_1,\rho_1,\hat H)\, R\,(q_2,\rho_2,\hat H) \land
H\canrest(\rho_1,\rho_2,\hat H)
  \,\}
\]
is a bisimulation.\\
Let $(q_1,\rho_1,H)\,\hat R\,(q_2,\rho_2,\hat H)$, due to some 
$(q_1,\rho_1,\hat H)\,R\,(q_2,\rho_2, \hat H)$, so WLOG assume that $\pi$ is the identity, and suppose $(q_1,\rho_1,H)\xrightarrow{(t,d)}(q_1',\rho_1',H')$ for some $t,d,q_1',\rho_1', H'$. 
{Next we reason by case analysis.}
\begin{enumerate}[(a)]
\item Suppose $d\in\rng{\rho_1}\cup\rng{\rho_2}$. Then, $ H'= H$ and, since $R$ is an $N$-bisimulation, we have   $(q_2,\rho_2, \hat H)\xrightarrow{(t,d)}(q_2',\rho_2',\hat H)$ for some $q_2',\rho_2'$ such that $(q_1',\rho_1',\hat H)\, R\, (q_2',\rho_2',\hat H)$. Hence, $(q_2,\rho_2, H)\xrightarrow{(t,d)}(q_2',\rho_2', H)$.
We need to show that $(q_1',\rho_1', H)\,\hat R\,(q_2',\rho_2', H)$. For this, it suffices that $ H\canrest(\rho_1',\rho_2',\hat H)$, which follows from 
$H\canrest(\rho_1,\rho_2,\hat H)$ and  $\rng{\rho_1'}\cup\rng{\rho_2'}\subseteq\rng{\rho_1}\cup\rng{\rho_2}$.
\item Suppose $d\in H\setminus(\rng{\rho_1}\cup\rng{\rho_2})$ and let $d'=\min(\hat H\setminus(\rng{\rho_1}\cup\rng{\rho_2}))$. Note that $d'\in H$. Then, we also have 
$(q_1,\rho_1,H)\xrightarrow{(t,d')}(q_1',(d\ d')\cdot\rho_1',H)$ and, hence, 
$(q_1,\rho_1,\hat H)\xrightarrow{(t,d')}(q_1',(d\ d')\cdot\rho_1',\hat H)$. By $N$-bisimulation, we get $(q_2,\rho_2,\hat H)\xrightarrow{(t,d')}(q_2',\rho_2',\hat H)$ and $(q_1,(d\ d')\cdot\rho_1',\hat H)\,R\,(q_2,\rho_2', \hat H)$.
But then
$(q_2,\rho_2,H)\xrightarrow{(t,d')}(q_2',\rho_2',H)$ and therefore
$(q_2,\rho_2,H)\xrightarrow{(t,d)}(q_2',(d\ d')\cdot\rho_2',H)$,
so it suffices to show that 
$(q_1',\rho_1', H)\,\hat R\,(q_2',(d\ d')\cdot\rho_2', H)$.
Note that $H\canrest(\rho_1,\rho_2,\hat H)$ implies that $H\canrest((d\ d')\cdot\rho_1',\rho_2',\hat H)$, thus $(q_1',(d\ d')\cdot\rho_1', H)\,\hat R\,(q_2',\rho_2', H)$, and hence $(q_1',\rho_1', H)=((d\ d')\cdot(q_1',(d\ d')\cdot\rho_1', H))\ \hat R\ ((d\ d')\cdot(q_2',\rho_2', H))=(q_2',(d\ d')\cdot\rho_2', H)$.

\item[(c1)] Suppose $d\notin H$, $\hat H=H\subsetneq N$,  and pick
  $d'=\min(N\setminus \hat H)$.
  Then, $ H'= H\uplus\{d\}$ and we also have
$(q_1,\rho_1,H)\xrightarrow{(t,d')}(q_1',(d\ d')\cdot\rho_1',(d\ d')\cdot H')$ and hence, 
  since $R$ is an $N$-bisimulation, we get
  $(q_2,\rho_2, H)\xrightarrow{(t,d')}(q_2',\rho_2',(d\ d')\cdot H')$ for some $q_2',\rho_2'$ with $(q_1',(d\ d')\cdot\rho_1',\hat H')\,R\,(q_2',\rho_2',\hat H')$
  and $\hat H'=\trim[(d\ d')\cdot\rho_1',\rho_2']{(d\ d')\cdot H'}$.
The latter implies that $(q_1',(d\ d')\cdot\rho_1',(d\ d')\cdot H')\,\hat R\,(q_2',\rho_2',(d\ d')\cdot H')$ and, by closure of $\hat R$,
$(q_1',\rho_1', H')\,\hat R\,(q_2',(d\ d')\cdot\rho_2', H')$. We conclude by noting that also   $(q_2,\rho_2, H)\xrightarrow{(t,d)}(q_2',(d\ d')\cdot\rho_2', H')$.
\item[(c2)]
  Suppose $d\notin H$, $N\subseteq H$ and $\hat H=N\setminus\{ d'\}$ for some $ d'\in N\setminus(\rng{\rho_1}\cup\rng{\rho_2})$, so $d'\in H$.
  Then, $ H'= H\uplus\{d\}$ and we also have
  $(q_1,\rho_1, \hat H)\xrightarrow{(t,d')}(q_1',(d\ d')\cdot\rho_1', N)$ and hence,
  since $R$ is an $N$-bisimulation, 
  $(q_2,\rho_2, \hat H)\xrightarrow{(t,d')}(q_2',\rho_2', N)$ for some $q_2',\rho_2'$ with
  $(q_1',(d\ d')\cdot\rho_1', \hat H')\,R\,(q_2',\rho_2',\hat H')$ and $\hat H'=\trim[(d\ d')\cdot\rho_1',\rho_2']{N}$. 
  But then   $(q_2,\rho_2, H)\xrightarrow{(t,d)}(q_2',(d\ d')\cdot\rho_2', H')$, so it suffices to show that 
  $(q_1',\rho_1',  H')\,\hat R\,(q_2',(d\ d')\cdot\rho_2',H')$.
  Noting that $H'\canrest((d\ d')\cdot\rho_1',\rho_2',\hat H')$, we get 
  $(q_1',(d\ d')\cdot\rho_1',  H')\,\hat R\,(q_2',\rho_2',H')$, from which the claim follows by closure of $\hat R$.
\end{enumerate}\medskip

\noindent
The lemma follows from the two reductions above, using the fact that bisimilarity satisfies the permutation-closure assumption used in the first reduction.
  \end{proof}
}

%% file: apx-rash.tex

\section{Proofs from Section~\ref{s:SI}}

\us{
\begin{proof}[Proof of Lemma~\ref{lem:updates}]
We do a case analysis on $i,j$ being 0 or not. Assume first that $i,j\neq0$. Then:
\begin{align*}
(\rho_1;\rho_2^{-1})[i\mapsto j] &=\{(i,j)\}\cup\{(i',j')\in[1,r]^2\ |\ i'\neq i\land j'\neq j\land\exists a'.\,\rho_1(i')=\rho_2(j')=a'\}\\
 &= \{(i,j)\}\cup\{(i',j')\in\rho_1[i\mapsto a];\rho_2[j\mapsto a]^{-1}\ |\ i'\neq i\land j'\neq j \}\\
                                 &= \rho_1[i\mapsto a];\rho_2[j\mapsto a]^{-1}
\end{align*}
On the other hand, if $i=j=0$ then the claim is trivial. Suppose now $i=0,j\neq0$. Then:
\begin{align*}
  (\rho_1;\rho_2^{-1})[i\mapsto j] &=\{(i',j')\in[1,r]^2\ |\ j'\neq j\land\exists a'.\,\rho_1(i')=\rho_2(j')=a'\}\\
  &=\{(i',j')\in[1,r]^2\ |\ j'\neq j\land\exists a'\neq a.\,\rho_1(i')=\rho_2(j')=a'\}\quad(\text{as } a\notin\rng{\rho_1})\\
                                 &= \rho_1;\rho_2[j\mapsto a]^{-1}= \rho_1[i\mapsto a];\rho_2[j\mapsto a]^{-1}
\end{align*}
Finally, if $i\neq0,j=0$ then we can show that $(\rho_1;\rho_2^{-1})[i\mapsto j]^{-1}=(\rho_1[i\mapsto a];\rho_2[j\mapsto a]^{-1})^{-1}$ using the previous case above.
\end{proof}}

\begin{proof}[Proof of Lemma~\ref{l:sym1}]
  We show a correspondence between bisimulations and symbolic bisimulations from which the result follows.
  
\paragraph{bisim\,$\to$\,s-bisim}
Let $R$ be a bisimulation on $\calA$. We claim that the relation $R'\subseteq\calU$,
\[\begin{aligned}
R' = \{\,
(q_1,S_1,\sigma,q_2,S_2)\ &|\ \exists\rho_1,\rho_2.\ 
(q_1,\rho_1)R(q_2,\rho_2) \land
\sigma=\rho_1;\rho_2^{-1}\land \dom{\rho_i}=S_i 
\,\}
\end{aligned}\]
is a symbolic bisimulation. For the latter (by symmetry in the definition) it suffices to show that $R'$ is a symbolic simulation.
So suppose that 
$(q_1,S_1,\sigma,q_2,S_2)\in R'$ due to some
$(q_1,\rho_1)R(q_2,\rho_2)$.
Let
$q_1\xr{t,i}q_1'$ for some $i\in S_1$. Then, $(q_1,\rho_1)\xr{t,a}(q_1',\rho_1)$ with $a=\rho_1(i)$ and, hence, $(q_2,\rho_2)\xr{t,a}(q_2',\rho_2')$ with $(q_1',\rho_1)R(q_2',\rho_2')$.
\begin{itemize}
\item If $i\in\dom{\sigma}$ then $a=\rho_2(\sigma(i))$ and therefore the above transition is due to some $q_2\xr{t,\sigma(i)}q_2'$, and $\rho_2'=\rho_2$. Hence, $(q_1',S_1)R'_\sigma(q_2',S_2)$.
\item If $i\notin\dom{\sigma}$ then the transition is due to some $q_2\xr{t,j\fre}q_2'$, and $\rho_2'=\rho_2[j\mapsto a]$. Hence, 
since $\sigma[i\mapsto j]=\rho_1;(\rho_2[j\mapsto a])^{-1}$ and $\dom{\rho_2'}=S_2[j]$, we have
$(q_1',S_1)R'_{\sigma[i\mapsto j]}(q_2',S_2[j])$.
\end{itemize}
Now let $q_1\xr{t,i\fre}q_1'$. For each $a\notin\rng{\rho_1}$, $(q_1,\rho_1)\xr{t,a}(q_1',\rho_1')$ with $\rho_1'=\rho_1[i\mapsto a]$ and, hence, there is some $(q_2,\rho_2)\xr{t,a}(q_2',\rho_2')$ with $(q_1',\rho_1')R(q_2',\rho_2')$.
\begin{itemize}
\item Select some $a\notin\rng{\rho_2}$. Then, the transition above is due to some $q_2\xr{t,j\fre}q_2'$, and $\rho_2'=\rho_2[j\mapsto a]$. Moreover, 
since $\sigma[i\mapsto j]=\rho_1[i\mapsto a];(\rho_2[j\mapsto a])^{-1}$,
$\dom{\rho_1'}=S_1[i]$ and $\dom{\rho_2'}=S_2[j]$, 
we have $(q_1',S_1[i])R'_{\sigma[i\mapsto j]}(q_2',S_2[j])$.
\item Let $j\in S_2\setminus\rng{\sigma}$. Then, we can take $a$ to be $\rho_2(j)$, so the transition is due to some $q_2\xr{t,j}q_2'$, and $\rho_2'=\rho_2$.
We moreover have $(q_1',S_1[i])R'_{\sigma[i\mapsto j]}(q_2',S_2)$.
\end{itemize}
\paragraph{s-bisim\,$\to$\,bisim}
Let $R$ be a symbolic bisimulation on $\calA$. We claim that the relation
\[\begin{aligned}
R' = \{\,
((q_1,\rho_1),(q_2,\rho_2))\ &|\
(q_1,S_1)R_{\sigma}(q_2,S_2)
\\ &\;\,
\land
\sigma=\rho_1;\rho_2^{-1}
\land
S_i=\dom{\rho_i}
\,\}
\end{aligned}\]
is a bisimulation, for which it suffices to show that $R'$ is a simulation.
So suppose that 
\[
((q_1,\rho_1),(q_2,\rho_2))\in R'
\] 
due to some 
$(q_1,S_1)R_{\sigma}(q_2,S_2)$, and let
$(q_1,\rho_1)\xr{t,a}(q_1',\rho_1')$ for some $(t,a)\in\Sigma\times\D$. 
If $a\in\rng{\rho_1}$, say $a=\rho_1(i)$, then 
$q_1\xr{t,i}q_1'$ and
$\rho_1'=\rho_1$. We distinguish two cases:
\begin{itemize}
\item If $a\in\rng{\rho_2}$ then $i\in\dom{\sigma}$, so $q_2\xr{t,\sigma(i)}q_2'$ and $(q_1',S_1)R_\sigma(q_2',S_2)$. Hence, $(q_2,\rho_2)\xr{t,a}(q_2',\rho_2)$ and $(q_1',\rho_1)R'(q_2',\rho_2)$.
\item If $a\notin\rng{\rho_2}$ then $i\in S_1\setminus\dom{\sigma}$, so 
$q_2\xr{t,j\fre}q_2'$ and $(q_1',S_1)R_{\sigma[i\mapsto j]}(q_2',S_2[j])$. Hence, $(q_2,\rho_2)\xr{t,a}(q_2',\rho_2[j\mapsto a])$ and $(q_1',\rho_1)R'(q_2',\rho_2[j\mapsto a])$.
\end{itemize}
If $a\notin\rng{\rho_1}$ then there is $q_1\xr{t,i\fre}q_1'$ such that $\rho_1'=\rho_1[i\mapsto a]$.
\begin{itemize}
\item If $a\notin\rng{\rho_2}$ then, since $q_2\xr{t,j\fre}q_2'$ with $(q_1',S_1[i])R_{\sigma[i\mapsto j]}(q_2',S_2[j])$, we obtain
$(q_2,\rho_2)\xr{t,a}(q_2',\rho_2[j\mapsto a])$ and $(q_1',\rho_1')R'(q_2',\rho_2[j\mapsto a])$.
\item If $a\in\rng{\rho_2}$, say $a=\rho_2(j)$, then
$j\in S_2\setminus\rng{\sigma}$. Hence,
$q_2\xr{t,j}q_2'$ with 
\[
(q_1',S_1[i])R_{\sigma[i\mapsto j]}(q_2',S_2),
\] 
from which we get $(q_2,\rho_2)\xr{t,a}(q_2',\rho_2)$ and $(q_1',\rho_1')R'(q_2',\rho_2)$. \qedhere
\end{itemize}
\end{proof}

\begin{proof}[Proof of Lemma~\ref{l:sym1b}]
By induction on $i$ we prove that, for all $i\in\omega$, ${\simi{i+1}} \subseteq {\simi{i}}$.
When $i=0$, the result is trivial as $\simi{i}$ is the universe.
Let us assume ${\simi{i+1}} \subseteq {\simi{i}}$ (IH) and $(q_1,S_1) \simi{i+2}_\tau (q_2,S_2)$.  
It follows by definition that $(q_1,S_1,\tau,q_2,S_2)$ and $(q_2,S_2,\tau^{-1},q_1,S_1)$ satisfy the \SyS\ conditions in $\simi{i+1}$.  
Because ${\simi{i+1}} \subseteq {\simi{i}}$,  the tuples also satisfy the \SyS\ conditions in $\simi{i}$, whence $(q_1,S_1) \simi{i+1}_\tau (q_2,S_2)$, as needed.
\\
We next show that ${\bigcap_{i\in\omega} \simi{i}} = {\sims}$.
We start with the $\supseteq$ direction and argue that, for all $i \in \omega$, ${\simi{i}}\supseteq{\sims}$.
The proof is by induction on $i$.  When $i=0$ the result is trivial.
Let us assume ${\simi{i}}\supseteq{\sims}$ (IH) and $(q_1,S_1) \sims_\tau (q_2,S_2)$.  
We wish to show that $(q_1,S_1,\tau,q_2,S_2)$ and its inverse satisfy the \SyS\ conditions in $\simi{i}$.  
By definition, they satisfy the \SyS\ conditions in $\sims$.  
Because ${\simi{i}}\supseteq{\sims}$, the tuples satisfy the \SyS\ conditions in $\simi{i}$. 
Hence, ${\simi{i+1}}\supseteq {\sims}$.

For the $\subseteq$ direction, we argue that the left-hand side is a symbolic bisimulation.
To see this, assume $(q_1,S_1,\tau,q_2,S_2) \in {\bigcap_{i\in\omega}\simi{i}}$ so that   
$(q_1,S_1,\tau,q_2,S_2)$ and its inverse satisfy the \SyS\ conditions in $\simi{i}$, for all $i \in \omega$.  
\us{The satisfaction of the \SyS\ conditions in $\simi{i}$ by $(q_1,S_1,\tau,q_2,S_2)$ (and, analogously, by its inverse)
is witnessed by a subset $C_i\subseteq{\simi{i}}\subseteq \calU$ for each $i$. Because $\calU$ is finite, there exists $C$ such that $C=C_i$ for infinitely many $i$.
Consequently, in view of ${\simi{i+1}}\subseteq {\simi{i}}$, $C$ witnesses satisfaction of the \SyS\ conditions in $({\bigcap_{i\in\omega}\simi{i}})$.}
\end{proof}

\begin{proof}[Proof of Lemma~\ref{l:sym2}]
\us{We first observe that $\Cl{R}=\Cl[-]{R\cup R^{-1}}$ where, for any relation $X$, we let $\Cl[-]{X}$ be the smallest relation that contains $X$ and is closed under the rules (\textsc{Id}), (\textsc{Tr}) and (\textsc{Ext}) above.
Let $R'=\Cl[-]{R\cup R^{-1}}$ and $P'=\Cl{P}$.
\cutout{Note first that $R'=R'^{-1}$ and, for all
$(q_1,S_1)\,R'_\sigma\,(q_2,S_2)$,%
\begin{enumerate}[1$'$.]
\item for all $\sigma'$, if $\sigma\leq_{S_1,S_2}\sigma'$ then 
$(q_1,S_1)\,R'_{\sigma'}\,(q_2,S_2)$;
\item for all $(q_2,S_2)\,R'_{\sigma'}\,(q_3,S_3)$ we have
$(q_1,S_1,)\,R'_{\sigma;\sigma'}\,(q_3,S_3)$.
\end{enumerate}
Item~1$'$ is clear from the definition; for~2$'$, we note that if $\sigma_1;,\cdots;\sigma_n\leq_{S_1,S_{3}}\sigma$ and $\sigma_1';\cdots;\sigma_{n'}'\leq_{S_{3},S_{n+n'+1}}\sigma$ then $\sigma_1;,\cdots;\sigma_n;\sigma_1';\cdots;\sigma_{n'}'\leq_{S_1,S_{n+n'+1}}\sigma;\sigma'$.
\\}
We show that all elements in $R'$ satisfy the \SyS\ conditions in $P'$, by rule induction on $\Cl[-]{R\cup R^{-1}}$. 
\\
For the base cases, either the element is in $R\cup R^{-1}$ or is an identity.  In both cases the result is clear. }%
For the inductive step,
consider the rule:
\[
 \frac{(q_1,S_1,\sigma_1,q_2, S_2)\in R'\qquad (q_2,S_2,\sigma_2, q_3, S_3)\in R'}{(q_1, S_1, \sigma_1;\sigma_2, q_3, S_3)\in R'}\;(\textsc{Tr})
\]
and assume that the premises satisfy the \SyS\ conditions in $P'$.
Let us write $\sigma$ for $\sigma_1;\sigma_2$.
Suppose $q_1\xr{t,i}q_1'$ {with $i\in S_1$}.
\begin{itemize}[$\bullet$]
\item 
If $i\in\dom{\sigma_1}$ and $j=\sigma_1(i)\in\dom{\sigma_2}$ then $q_2\xr{t,j}q_2'$ with $(q_1',S_1)\myPP[']{\sigma_1}(q_2',S_2)$, 
and $q_{3}\xr{t,k}q_{3}'$ with $(q_2',S_2)\myPP[']{\sigma_2}(q_{3}',S_{3})$ and $k=\sigma_2(j)=\sigma(i)$. By~(\textsc{Tr}) we obtain $(q_1',S_1)\myPP[']{\sigma}(q_{3}',S_{3})$.
\item 
If $i\in\dom{\sigma_1}$ and $j=\sigma_1(i)\notin\dom{\sigma_2}$ then $q_2\xr{t,j}q_2'$ with $(q_1',S_1)\myPP[']{\sigma_1}(q_2',S_2)$, 
and $q_{3}\xr{t,k\fre}q_{3}'$ with $(q_2',S_2)\myPP[']{\sigma_2[j\mapsto k]}(q_{3}',S_{3}[k])$ for some $k$. By~(\textsc{Tr}) we obtain $(q_1',S_1)\myPP[']{\sigma[i\mapsto k]}\allowbreak(q_{3}',S_{3}[k])$.
\item 
If $i\notin\dom{\sigma_1}$ then $q_2\xr{t,j\fre}q_2'$ with $(q_1',S_1)\myPP[']{\sigma_1[i\mapsto j]}(q_2',S_2[j])$, for some $j$, so 
$q_{3}\xr{t,k\fre}q_{3}'$ with $(q_2',S_2[j])\myPP[']{\sigma_2[j\mapsto k]}(q_{3}',S_{3}[k])$ for some $k$. By~(\textsc{Tr,Ext}), using $\sigma_1[i\mapsto j];\sigma_2[j\mapsto k]\leq_{S_1,S_{3}[k]}\sigma[i\mapsto k]$, we get
$(q_1',S_1)\myPP[']{\sigma[i\mapsto k]}(q_{3}',S_{3}[k])$.
\end{itemize}
Now suppose $q_1\xr{t,i\fre}q_1'$.
\begin{itemize}[$\bullet$]
\item 
Then, $q_2\xr{t,j\fre}q_2'$ with $(q_1',S_1[i])\myPP[']{\sigma_1[i\mapsto j]}(q_2',S_2[j])$, for some $j$, so 
$q_{3}\xr{t,k\fre}q_{3}'$ with 
\[
(q_2',S_2[j])\myPP[']{\sigma_2[j\mapsto k]}(q_{3}',S_{3}[k])
\] for some $k$. By~(\textsc{Tr,Ext}), 
$(q_1',S_1[i])\myPP[']{\sigma[i\mapsto k]}(q_{3}',S_{3}[k])$.
\item
If $k\in\rng{\sigma_2}$ and $j=\sigma_2^{-1}(k)\notin\rng{\sigma_1}$ then $q_2\xr{t,j}q_2'$ with $(q_1',S_1[i])\myPP[']{\sigma_1[i\mapsto j]}(q_2',S_2)$, 
and $q_{3}\xr{t,k}q_{3}'$ with $(q_2',S_2)\myPP[']{\sigma_2}(q_{3}',S_{3})$. By~(\textsc{Tr})  obtain $(q_1',S_1[i])\myPP[']{\sigma[i\mapsto k]}(q_{3}',S_{3})$.
\item
If $k\in S_{3}\setminus\rng{\sigma_2}$ then $q_2\xr{t,j\fre}q_2'$ with $(q_1',S_1[i])\myPP[']{\sigma_1[i\mapsto j]}(q_2',S_2[j])$, for some $j$, 
and so $q_{3}\xr{t,k}q_{3}'$ with $(q_2',S_2[j])\myPP[']{\sigma_2[j\mapsto k]}(q_{3}',S_{3})$. By~(\textsc{Tr,Ext}) we obtain $(q_1',S_1[i])\myPP[']{\sigma[i\mapsto k]}(q_{3}',S_{3})$.
\end{itemize}
Consider now the rule:
\[
 \frac{(q_1,S_1,\sigma, q_2, S_2)\in R'\qquad \sigma\le_{S_1,S_2} \sigma'}{ (q_1,S_1,\sigma', q_2, S_2)\in R'}\;(\textsc{Ext})
\]
and assume $(q_1,S_1,\sigma, q_2, S_2)$ satisfies the \SyS\ conditions in $P'$.
Suppose $q_1\xr{t,i}q_1'$ with $i\in S_1$.
\begin{itemize}[$\bullet$]
\item 
If $i\in\dom{\sigma}$ then $q_{2}\xr{t,\sigma(i)}q_{2}'$ and $(q_1',S_1)\myPP[']{\sigma}(q_{2}',S_{2})$. Since $\sigma\subseteq\sigma'$, we have $\sigma(i)=\sigma'(i)$ and $(q_1',S_1)\myPP[']{\sigma'}(q_{2}',S_{2})$.
\item 
If $i\notin\dom{\sigma'}$ then also $i\notin\dom{\sigma}$ and therefore
$q_{2}\xr{t,j\fre}q_{2}'$, for some $j$, and 
\[
(q_1',S_1)\myPP[']{{\sigma[i\mapsto j]}}(q_{2}',S_{2}[j]).
\] 
From $\sigma\leq_{S_1,S_{2}}\sigma'$ we obtain $\sigma[i\mapsto j]\leq_{S_1,S_{2}[j]}\sigma'[i\mapsto j]$, 
so $(q_1',S_1)\myPP[']{\sigma'[i\mapsto j]}(q_{2}',S_{2}[j])$.
\item 
If $i\in\dom{\sigma'}\setminus\dom{\sigma}$ then we reason as follows. Let $\sigma'(i)=j\in S_{2}$.
\begin{enumerate}[I.]
\item Since $i\notin\dom{\sigma}$, there is some $q_{2}\xr{t,{j'}\fre}q_{2}''$ with $(q_1',S_1)\myPP[']{\sigma[i\mapsto j']}(q_{2}'',S_{2}[j'])$;
\item hence, there is some $q_1\xr{t,{i'}\fre}q_1''$ with 
$(q_1'',S_1[i'])\myPP[']{\sigma[i'\mapsto j']}(q_{2}'',S_{2}[j'])$;
\item then, there is some $q_{2}\xr{t,j}q_{2}'$ with 
$(q_1'',S_1[i'])\myPP[']{\sigma[i'\mapsto j]}(q_{2}',S_{2})$.
\end{enumerate}
Taking stock (and using symmetry of $P'$),
\[
(q_1',S_1)\myPP[']{\sigma[i\mapsto j']}(q_{2}'',S_{2}[j'])\myPP[']{\sigma^{-1}[j'\mapsto i']}
(q_1'',S_1[i'])
\myPP[']{\sigma[i'\mapsto j]}(q_{2}',S_{2})
\]
and thus, since $\sigma[i\mapsto j'];\sigma^{-1}[j'\mapsto i'];\sigma[i'\mapsto j]\leq_{S_1,S_{2}}\sigma[i\mapsto j]$, we have 
\[
(q_1',S_1)\myPP[']{\sigma[i\mapsto j]}(q_{2}',S_{2}).
\]
\end{itemize}
Suppose now $q_1\xr{t,i\fre}q_1'$.
\begin{itemize}[$\bullet$]
\item 
Then, $q_{2}\xr{t,j\fre}q_{2}'$ and $(q_1',S_1[i])\myPP[']{\sigma[i\mapsto j]}(q_{2}',S_{2}[j])$. Since $\sigma[i\mapsto j]\leq_{S_1[i],S_{2}[j]}\sigma'[i\mapsto j]$, we have $(q_1',S_1[i])\myPP[']{\sigma'[i\mapsto j]}(q_{2}',S_{2}[j])$.
\item 
If $j\in S_{2}\setminus\rng{\sigma'}$ then $j\notin\rng{\sigma}$, hence $q_{2}\xr{t,j}q_{2}'$ and $(q_1',S_1[i])\myPP[']{\sigma[i\mapsto j]}(q_{2}',S_{2})$. Again, we obtain $(q_1',S_1[i])\myPP[']{\sigma'[i\mapsto j]}(q_{2}',S_{2})$.
\end{itemize}
Hence, all elements of $R'$ satisfy the \SyS\ conditions in $P'$.
\end{proof}

\begin{proof}[Proof of Lemma~\ref{lem:bijections}]
  {(First part). Since $R$ is closed,} $(p,S) \relR_{\sigma;\sigma^{-1}} (p,S)$.
Because $\sigma;\sigma^{-1}=\id{X}$ for some $X\subseteq S$, 
we have $X\supseteq X_S^p{(R)}$. 
Moreover,  $\dom{\sigma}\supseteq \dom{\sigma;\sigma^{-1}}=X$, hence $\dom{\sigma}\supseteq X_S^p{(R)}$.
A symmetric argument establishes that $\dom{\sigma^{-1}}\supseteq X_S^q{(R)}$. 

(Second part).
By definition, we have that $\dom{\sigma'}\subseteq X_S^p{(R)}$ and $\rng{\sigma'}\subseteq X_S^q{(R)}$.
Observing that $\sigma'=\id{X_S^p{(R)}};\sigma;\id{X_S^q{(R)}}$, by closure of $R$ we get $(p,S)\relR_{\sigma'}(q,S)$. 
By the first part, $\dom{\sigma'}\supseteq X_S^p{(R)}$ and $\rng{\sigma'}\supseteq X_S^q{(R)}$,
hence $\dom{\sigma'}=X_S^p{(R)}$ and $\rng{\sigma'}=X_S^q{(R)}$.
The final claim follows from the fact that $(p,S)\relR_{\id{S}}(p,S)$.
\end{proof}

\begin{proof}[Proof of Lemma~\ref{lem:bound}]
Let us write $\mea{S_1,S_2}$ for $|S_1|+|S_2|$, i.e. $0\le \mea{S_1,S_2}\le 2r$.
For each $m\in [0,2r]$, let
\[
k_m = \min \{i \,|\, {\simi{i}}\cap\uni{S_1,S_2} ={\sims}\cap\uni{S_1,S_2} \textrm{ for any $S_1, S_2$ with $\mea{S_1,S_2}\ge m$}\}.
\]
Consider $S_1,S_2$ with $\mea{S_1,S_2}\ge m$, where $m<2r$.

Observe that, for $k\geq k_{m+1}$,  if
$\simi{k}\cap\ \uni{S_1,S_2}=\ \simi{k+1}\cap\ \uni{S_1,S_2}$, then we must have
$\simi{k}\cap\ \uni{S_1,S_2}=\ \sims\cap\ \uni{S_1,S_2}$, because {the
 \SyS\ conditions for $(S_1,S_2)$ refer to either $(S_1,S_2)$ or $(S_1',S_2')$ 
 with $\mea{S_1',S_2'}>\mea{S_1,S_2}$.}
Consequently, if $\simi{k}\cap\ \uni{S_1,S_2} \neq\ \sims\cap\ \uni{S_1,S_2}$,
the sequence $(\simi{k}\cap\ \uni{S_1,S_2})$ ($k=k_{m+1},k_{m+1}+1,\cdots$) will have to change in every step before stabilisation. By Lemma~\ref{lem:play-length}, 
at most  $\ell$ extra steps from $(\simi{k_{m+1}})$ will be required to arrive at 
$\sims \cap\ \uni{S_1,S_2}$, which implies $k_m \le k_{m+1}+\ell$.
By a similar argument, we can conclude that $k_{2r}\le \ell$.
Consequently, $k_0 \le (2r+1) \ell$, as required.\end{proof}

%% file: apx-frash.tex
\section{Proofs from Section~\ref{sec:frash}}

\us{
\begin{proof}[Proof of Lemma~\ref{lem:updates2}]
  For the first claim, note that it suffices to consider the case where the product $ii'jj'$ is not 0 as e.g.\ if $ii'=0$ then $\xsw{i}{i'}=\xsw{1}{1}$. In this case, the claim follows by composition of partial permutations, noting that $\rho_2^{-1};\sw{j}{j'}=(\sw{j}{j'};\rho_2)^{-1}$.\\
  Claim 2 then follows as
     Lemma~\ref{lem:updates} implies that \
  $(\rho_1;\rho_2^{-1})[i'\mapsto j'] = \rho_1[i'\mapsto a];\linebreak[5]\rho_2[j'\mapsto a]^{-1}$.
\end{proof}

\begin{proof}[Proof of Lemma~\ref{l:symbchoice}]
Let $(q_1,S_1,\sigma,q_2,S_2,h)$, $(q_1,S_1',\sigma',q_2,S_2',h)\in\symb(\kappa_1,\kappa_2)$ be distinct and produced from $\hat\rho_i$ and $\hat\rho_i'$ respectively (for $i=1,2$). Let us assume that $(q_1,S_1',\sigma',q_2,S_2',h)\in{\sims}$.
Take $\sigma_i=\hat\rho_i;\hat\rho_i'^{\us{-1}}$. By definition, $\sigma_i\upharpoonright[1,r]=\id{S_i\cap[1,r]}$, 
and we can verify that $(q_i,S_i)\simihn{\mathsf{s}}{h}{\sigma_i}(q_i,S_i')$. Hence, 
$(q_1,S_1)\,\simihn{\mathsf{s}}{h}{\sigma_1}\,(q_1,S_1')\,\simihn{\mathsf{s}}{h}{\sigma'}\,(q_2,S_2')\,\simihn{\mathsf{s}}{h}{\sigma_2^{-1}}\,(q_2,S_2)$ and,
using Proposition~\ref{p:closuresFRA} (which does not depend on this lemma), we get 
{$(q_1,S_1,\sigma_1;\sigma';\sigma_2^{-1},\allowbreak q_2,S_2,h)=(q_1,S_1,\sigma,q_2,S_2,h)\in{\sims}$.}
\end{proof}

}

\begin{proof}[Proof of Lemma~\ref{l:sym1FRA}]
Let $\clg{A}$ be an $r$-FRA($S\#_0$).
We show a correspondence between bisimulations and symbolic bisimulations for $\clg{A}$ from which the result follows.

\paragraph{bisim\,$\to$\,s-bisim}
Let $R$ be a bisimulation on $\calA$. We claim that the relation $\us{P}\subseteq\calU$,
\[\begin{aligned}
P = \bigcup\{\,\symb(\kappa_1,\kappa_2)\ |\ \us{(\kappa_1,\kappa_2)\in R}\land\kappa_i=(q_i,\rho_i,H_i)\land H_1=H_2\}
\end{aligned}\]
is a symbolic bisimulation. For the latter (by symmetry) it suffices to show that $P$ is a symbolic simulation, which reduces to showing the \FSyS\ conditions true.
So suppose that 
$(q_1,S_1,\sigma,q_2,S_2)\in P^h$ due to some
$(q_1,\rho_1,H)R(q_2,\rho_2,H)$. If $h\le2r$ then let $\hat\rho_i$ be \us{some}
$3r$-register assignment of type $S\#_0$ used by $\symb$ (for $i=1,2$), so
$\hat\rho_i\upharpoonright[1,r]=\rho_i$, $S_i=\dom{\hat\rho_i}$, $\rng{\hat\rho_i}=H$ and $\sigma=\hat\rho_1;\hat\rho_2^{-1}$.

Let
$q_1\xr{t,i}q_1'$ for some $i\in S_1\cap[1,r]$. Then, $(q_1,\rho_1,H)\xr{t,a}(q_1',\rho_1,H)$ with $a=\rho_1(i)\in H$ and, hence, $(q_2,\rho_2,H)\xr{t,a}(q_2',\rho_2',H)$ with $(q_1',\rho_1,H)R(q_2',\rho_2',H)$.
\begin{itemize}
\item If $\sigma(i)\in[1,r]$ then $a=\rho_2(\sigma(i))$ and therefore the above transition is due to some $q_2\xr{t,\sigma(i)}q_2'$, and $\rho_2'=\rho_2$. Hence, $(q_1',S_1)P^h_\sigma(q_2',S_2)$.
\item If $\sigma(i)=j'\in[r{+}1,3r]$ then $a=\hat\rho_2(j')\notin\rng{\rho_2}$ and the above transition is due to some $q_2\xr{t,j\fre}q_2'$, and $\rho_2'=\rho_2[j\mapsto a]$. 
Now, taking $\hat\rho_2'=\hat\rho_2\us{\xsw{j}{j'}}$,
we have $(q_1',S_1,\hat\rho_1;\hat\rho_2'^{-1},q_2',\dom{\hat\rho_2'})\in P^h$.
Since $\hat\rho_1;\hat\rho_2'^{-1}=\us{\xsw{j}{j'}}\sigma$
 and $\dom{\hat\rho_2'}= S_2'\us{\xsw{j}{j'}}$, we
obtain 
\[
(q_1',S_1)P^h_{\us{\xsw{j}{j'}}\sigma}(q_2', S_2\us{\xsw{j}{j'}}).
\]
\item If $i\notin\dom{\sigma}$ then $h=\infty$ and the transition is due to some $q_2\xr{t,j\fre}q_2'$, and $\rho_2'=\rho_2[j\mapsto a]$. Hence, 
since $\sigma[i\mapsto j]=\rho_1;(\rho_2[j\mapsto a])^{-1}$ and $\dom{\rho_2'}=S_2[j]$, we have
\[
(q_1',S_1)P^h_{\sigma[i\mapsto j]}(q_2',S_2[j]).
\]
\end{itemize}
Let $q_1\xr{t,i\fre}q_1'$. For each $a\in H\setminus\rng{\rho_1}$, $(q_1,\rho_1,H)\xr{t,a}(q_1',\rho_1',H)$ with $\rho_1'=\rho_1[i\mapsto a]$ and, hence, there is some $(q_2,\rho_2,H)\xr{t,a}(q_2',\rho_2',H)$ with $(q_1',\rho_1',H)R(q_2',\rho_2',H)$. 
Now, let $a=\hat\rho_1(i')$ for $i'\in S_1\setminus[1,r]$ (if $h\le2r$), and $a=\hat\rho_2(j)$ for $j\in S_2\setminus\rng{\sigma}$ (if $h=\infty$); in the former case, set $\hat\rho_1'=\hat\rho_1\us{\xsw{i}{i'}}$.
\begin{itemize}
\item 
If $\sigma(i')\in[1,r]$ then $a=\rho_2(\sigma(i'))$ so the transition above is due to some $q_2\xr{t,\sigma(i')}q_2'$ and $\rho_2'=\rho_2$. Thus, $(q_1',\dom{\hat\rho_1'})P^h_{\hat\rho_1';\hat\rho_2^{-1}}(q_2',S_2)$
i.e.\ 
$(q_1', S_1\us{\xsw{i}{i'}})P^h_{\sigma\us{\xsw{i}{i'}}}(q_2',S_2)$.
\item 
If $\sigma(i')=j'\in[r{+}1,3r]$ then $a=\hat\rho_2(j')\notin\rng{\rho_2}$ so the transition above is due to some $q_2\xr{t,j\fre}q_2'$ and $\rho_2'=\rho_2[j\mapsto a]$. Thus, setting 
$\hat\rho_2'=\hat\rho_2\us{\xsw{j}{j}}$, we obtain
\[
(q_1',\dom{\hat\rho_1'})P^h_{\hat\rho_1';\hat\rho_2'^{-1}}(q_2',\dom{\hat\rho_2'}),
\]
i.e.\ 
$(q_1',S_1\us{\xsw{i}{i'}})P^h_{\us{\xsw{j}{j'}}\sigma\us{\xsw{i}{i'}}}(q_2', S_2\us{\xsw{j}{j'}})$.
\item 
For $a=\hat\rho_2(j)$ with $j\in S_2\setminus\rng{\sigma}$, 
the transition is due to some $q_2\xr{t,j}q_2'$, and $\rho_2'=\rho_2$.
We moreover have $(q_1',S_1[i])P^h_{\sigma[i\mapsto j]}(q_2',S_2)$.
\end{itemize}
Finally, let 
$q_1\xr{t,\ell_i}q_1'$ with $\ell_i\in\{i\fre,i\gfre\}$. For each $a\notin H$, we have $(q_1,\rho_1,H)\xr{t,a}(q_1',\rho_1',H')$ with $\rho_1'=\rho_1[i\mapsto a]$ and $H'=H\cup\{a\}$ and, hence, there is some $(q_2,\rho_2,H)\xr{t,a}(q_2',\rho_2',H')$ with $(q_1',\rho_1',H')R(q_2',\rho_2',H')$.
The latter must be due to $q_2\xr{t,\ell_j}q_2'$, for some $\ell_j\in\{j\fre,j\gfre\}$, in which case $\rho_2'=\rho_2[j\mapsto a]$.
\begin{itemize}
\item 
If $h<2r$ then let $\hat\rho'_1=\hat\rho_1[i'\mapsto a]\us{\xsw{i}{i'}}$ and 
$\hat\rho'_2=\hat\rho_2[j'\mapsto a]\us{\xsw{j}{j'}}$, where $i'=\min([r{+}1,3r]\setminus\dom{\hat\rho_1})$
and $j'=\min([r{+}1,3r]\setminus\dom{\hat\rho_2})$.
We have $\rho_1'=\hat\rho_1'\upharpoonright[1,r]$, similarly for $\rho_2'$, 
and $\hat\rho_1';\hat\rho_2'^{-1}=\us{\xsw{j}{j'}}(\sigma[i'\mapsto j'])\us{\xsw{i}{i'}}$,
so 
\[
(q_1', S_1\us{\xsw{i}{i'}})P^{h+1}_{\us{\xsw{j}{j'}}(\sigma[i'\mapsto j'])\us{\xsw{i}{i'}}}(q_2',S_2\us{\xsw{j}{j'}}).
\]
\item 
If $h=2r$ then $(q_1',\dom{\rho_1'})P^\infty_{\rho_1';\rho_2'^{-1}}(q_2',\dom{\rho_2'})$.
Now observe that $\hat\rho_1[i\mapsto a]\upharpoonright[1,r]=\rho_1'$, similarly for $\rho_2'$,
and hence $\sigma[i\mapsto j]\cap[1,r]^2=\rho_1';\rho_2'^{-1}$.
\item 
If $h=\infty$ then, 
since $\sigma[i\mapsto j]=\rho_1[i\mapsto a];(\rho_2[j\mapsto a])^{-1}$,
$\dom{\rho_1'}=S_1[i]$ and $\dom{\rho_2'}=S_2[j]$, 
we have $(q_1',S_1[i])P^h_{\sigma[i\mapsto j]}(q_2',S_2[j])$.
Moreover, if $\ell_i=i\fre$ then, since $|H|>|\rng{\rho_1}|+|\rng{\rho_2}|$, there is some $a'\in H\setminus(\rng{\rho_1}\cup\rng{\rho_2})$. We can therefore pick $a=a'$ and the latter would impose $\ell_j=j\fre$.
\end{itemize}
Hence, $P$ is a symbolic bisimulation.

\paragraph{s-bisim\,$\to$\,bisim}
Let $R$ be a symbolic bisimulation on $\calA$ such that, for all pairs of configurations $\kappa_1,\kappa_2$, either $\symb(\kappa_1,\kappa_2)\subseteq R$ or 
$\symb(\kappa_1,\kappa_2)\cap R=\emptyset$.
We claim that the relation
\[
R'= \{\,
(\kappa_1,\kappa_2)\ |\ \kappa_i=(q_i,\rho_i,H_i)\land H_1=H_2
\land \symb(\kappa_1,\kappa_2)\subseteq R\,\}
\]
is a bisimulation, for which it suffices to show that $R'$ is a simulation.
So suppose that 
\[
((q_1,\rho_1,H),(q_2,\rho_2,H))\in R'
\] and let $(q_1,S_1,\sigma,q_2,S_2,h)\in \symb((q_1,\rho_1,H),(q_2,\rho_2,H))\subseteq R$, 
and if $h\le2r$ let $\hat\rho_i$ be some $3r$-extension of $\rho_i$ used by $\symb$.
Let $(q_1,\rho_1,H)\xr{t,a}(q_1',\rho_1',H')$ for some $(t,a)\in\Sigma\times\D$. 

If $a\in\rng{\rho_1}$, say $a=\rho_1(i)$, then 
$q_1\xr{t,i}q_1'$ and
$\rho_1'=\rho_1$. We distinguish three cases:
\begin{itemize}
\item If $a\in\rng{\rho_2}$ then $\sigma(i)\in[1,r]$, so $q_2\xr{t,\sigma(i)}q_2'$ and $(q_1',S_1)R^h_\sigma(q_2',S_2)$. Hence, $(q_2,\rho_2,H)\xr{t,a}(q_2',\rho_2,H)$ and $(q_1',\rho_1,H)R'(q_2',\rho_2,H)$.
\item If $a\notin\rng{\rho_2}$ and $h\le2r$ then $\sigma(i)=j'\in[r{+}1,3r]$, so $q_2\xr{t,j\fre}q_2'$ and 
\[
(q_1',S_1)R^h_{\us{\xsw{j}{j'}}\sigma}(q_2',S_2\us{\xsw{j}{j'}}),
\] 
for some $j$. 
Hence, $(q_2,\rho_2,H)\xr{t,a}(q_2',\rho_2[j\mapsto a],H)$ and, taking $\hat\rho_2'=\hat\rho_2\us{\xsw{j}{j'}}$ (so $\hat\rho_2'\upharpoonright[1,r]=\rho_2[j\mapsto a]$), we have 
$(q_1',\dom{\hat\rho_1})R^h_{\hat\rho_1;\hat\rho_2'^{-1}}(q_2',\dom{\hat\rho_2'})$
hence 
\[
(q_1',\rho_1,H)R'(q_2',\rho_2[j\mapsto a],H).
\]
\item If $a\notin\rng{\rho_2}$ and $h=\infty$ then $i\in S_1\setminus\dom{\sigma}$, so 
$q_2\xr{t,j\fre}q_2'$ and $(q_1',S_1)R^h_{\sigma[i\mapsto j]}(q_2',S_2[j])$. Hence, $(q_2,\rho_2,H)\xr{t,a}(q_2',\rho_2[j\mapsto a],H)$ and this $(q_1',\rho_1,H)R'(q_2',\rho_2[j\mapsto a],H)$.
\end{itemize}
If $a\in H\setminus\rng{\rho_1}$, 
and either $h\leq2r$ (so $a=\hat\rho_1(i')$ for some $i'>r$) or $h=\infty$ and $a\in\rng{\rho_2}$,
then $H'=H$ and there is some $q_1\xr{t,i\fre}q_1'$ and $\rho_1'=\rho_1[i\mapsto a]$.
\begin{itemize}
\item 
If $h\le2r$ and 
$\sigma(i')\in[1,r]$ then
$q_2\xr{t,\sigma(i')}q_2'$ 
and $(q_1',S_1\us{\xsw{i}{i'}})R^h_{\sigma\us{\xsw{i}{i'}}}(q_2',S_2)$.
Thus, since  $\rho_2(\sigma(i'))=a$, $(q_2,\rho_2,H)\xr{t,a}(q_2',\rho_2,H)$ and, setting $\hat\rho_1'=\hat\rho_1\us{\xsw{i}{i'}}$, we obtain $(q_1',\rho_1',H)R'(q_2',\rho_2,H)$.
\item
If $h\le2r$ and 
$\sigma(i')=j'\in[r{+}1,r]$ 
then $q_2\xr{t,j\fre}q_2'$ with $(q_1',S_1\us{\xsw{i}{i'}})R^h_{\us{\xsw{j}{j'}}\sigma\us{\xsw{i}{i'}}}\allowbreak(q_2',S_2\us{\xsw{j}{j'}})$, for some $j$.
Thus, since $a\notin\rng{\rho_2}$, 
$(q_2,\rho_2,H)\xr{t,a}(q_2',\rho_2[j\mapsto a],H)$ and, setting $\hat\rho_1'=\hat\rho_1\us{\xsw{i}{i'}}$ and $\hat\rho_2'=\hat\rho_2\us{\xsw{j}{j'}}$, we obtain $(q_1',\rho_1',H)R'(q_2',\rho_2[j\mapsto a],H)$.
\item If $h=\infty$ and $a\in\rng{\rho_2}$, say $a=\rho_2(j)$, then
$j\in S_2\setminus\rng{\sigma}$. Hence,
$q_2\xr{t,j}q_2'$ with $(q_1',S_1[i])R_{\sigma[i\mapsto j]}(q_2',S_2)$, from which we get
$(q_2,\rho_2,H)\xr{t,a}(q_2',\rho_2,H)$ and 
\[
(q_1',\rho_1',H)R'(q_2',\rho_2,H).
\]
\end{itemize}
If either $h\le2r$ and $a\notin H$, or $h=\infty$ and $a\notin\rng{\rho_1}\cup\rng{\rho_2}$
then $q_1\xr{t,\ell_i}q_1'$, for some $\ell_i\in\{i\fre,i\gfre\}$, and $H'=H\cup\{a\}$ and $\rho_1'=\rho_i[i\mapsto a]$.
Thus,  $q_2\xr{t,\ell_j}q_2'$ for some $\ell_j\in\{j\fre,j\gfre\}$. Let $\rho_2'=\rho_2[j\mapsto a]$.
\begin{itemize}
\item If $h<2r$ then,
taking $i'=\max([r{+}1,3r]\setminus S_1)$ and $j'=\max([r{+}1,3r]\setminus S_2)$,
we have $(q_1',S_1\us{\xsw{i}{i'}})R^{h+1}_{\us{\xsw{j}{j'}}\sigma[i'\mapsto j']\us{\xsw{i}{i'}}}(q_2', S_2\us{\xsw{j}{j'}})$.
Setting $\hat\rho_1'=\hat\rho_1[i'\mapsto a]\us{\xsw{i}{i'}}$ and 
$\hat\rho_2'=\hat\rho_2[j'\mapsto a]\us{\xsw{j}{j'}}$, we obtain
$(q_1',\rho_1',H')R'(q_2',\rho_2',H')$.
\item If $h=2r$ then $(q_1',S_1[i]\cap[1,r])R^{\infty}_{\sigma[i\mapsto j]\cap[1,r]^2}(q_2',S_2[j]\cap[1,r])$, from which we obtain
\[
(q_1',\rho_1',H')R'(q_2',\rho_2',H').
\]
\item If $h=\infty$ then 
$(q_1',S_1[i])R^h_{\sigma[i\mapsto j]}(q_2',S_2[j])$.
In particular, if $a\in H$ then $\ell_i=i\fre$ and therefore $\ell_j=j\fre$. Thus, in each case,
$(q_2,\rho_2,H)\xr{t,a}(q_2',\rho_2[j\mapsto a],H')$ and $(q_1',\rho_1',H')R'(q_2',\rho_2',H')$.
\end{itemize}
Hence, $R'$ is a bisimulation.

Thus, to prove Lemma~\ref{l:sym1FRA}, given such $\kappa_1$ and $\kappa_2$, if $\kappa_1\sims\kappa_2$ then we can construct a symbolic bisimulation $P$ such that $\symb(\kappa_1,\kappa_2)\subseteq P$. Conversely, if $\kappa_1\sims\kappa_2$ then, using also Lemma~\ref{l:symbchoice}, there is a bisimulation $R'$ such that $\kappa_1 R'\kappa_2$.
\end{proof}

\begin{proof}[Proof of Lemma~\ref{lem:ftwiddle}]
\us{For the first part, we argue by induction on $i$.
For $i=0$ we need to show ${\simih{1}}\subseteq{\simih{0}}$, which is true because ${\simi{0}}=\calU$.
Next, assuming ${\simih{i+1}}\subseteq{\simih{i}}$, we argue that ${\simih{i+2}}\subseteq{\simih{i+1}}$.
Suppose $(q_1,S_1)~\simihn{i+\us{2}}{h}{\tau}~(q_2,S_2)$.  It follows by definition that $(q_1,S_1,\tau,q_2,S_2,h)$ and $(q_2,S_2,\tau^{-1},q_1,S_1,h)$ satisfy the \us{\FSyS}\ conditions in $\simi{i+1}$.
Because  ${\simih{i+1}}\subseteq{\simih{i}}$,  the tuples also satisfy the \us{\FSyS}\ conditions in $\simi{i}$, which implies $(q_1,S_1)~\simihn{i+1}{h}{\tau}~(q_2,S_2)$.}

\us{
For the second part,
we start with $\supseteq$ and argue that, for all $i \in \omega$, ${\simi{i}}\supseteq{\sims}$.
The proof is by induction on $i$.  For $i=0$ the result is trivial, because ${\simi{0}}=\calU$.
Next, assuming ${\simi{i}}\supseteq{\sims}$, we will show ${\simi{i+1}}\supseteq{\sims}$.
Suppose $(q_1,S_1) (\sims)_\tau^h (q_2,S_2)$.  
We wish to show that $(q_1,S_1,\tau,q_2, S_2,h)$ and its inverse  satisfy the \us{\FSyS}\ conditions in $\simi{i}$.  By definition, they satisfy the \us{\FSyS}\ conditions in $\sims$.  
Because of ${\simi{i}}\supseteq{\sims}$, this implies that they satisfy the \us{\FSyS}\ conditions in $\simi{i}$. Hence, ${\simi{i+1}}\supseteq {\sims}$, as required.

For the $\subseteq$ direction, we argue that the left-hand side is a symbolic bisimulation.
To see this, assume $(q_1,S_1,\tau,q_2,S_2,h) \in {\bigcap_{i\in\omega}\simi{i}}$ so that   
$(q_1,S_1,\tau,q_2,S_2,h)$ and its inverse satisfy the \us{\FSyS}\ conditions in $\simi{i}$, for all $i \in \omega$.  
The satisfaction of the \us{\FSyS}\ conditions in $\simi{i}$ by $(q_1,S_1,\tau,q_2,S_2,h)$ (and, analogously, by its inverse)
is witnessed by a subset $C_i\subseteq{\simi{i}}\subseteq \calU$ for each $i$. Because $\calU$ is finite, there exists $C$ such that $C=C_i$ for infinitely many $i$.
Consequently, in view of ${\simi{i+1}}\subseteq {\simi{i}}$, $C$ witnesses satisfaction of the \us{\FSyS}\ conditions in $({\bigcap_{i\in\omega}\simi{i}}q)$.}
\end{proof}

\begin{proof}[Proof of Lemma~\ref{l:fsym2}]
\us{We first observe that $\Cl{R}=\Cl[-]{R\cup R^{-1}}$ where, for any relation $X$, we let $\Cl[-]{X}$ be the smallest relation that contains $X$ and is closed under the rules (\textsc{Id}), (\textsc{Tr}) and (\textsc{Ext}) above.
Let $\hat R=\Cl[-]{R\cup R^{-1}}$ and $\hat P=\Cl{P}$.
We show that all elements in $\hat R$ satisfy the \FSyS\ conditions in $\hat P$, by rule induction on $\Cl[-]{R\cup R^{-1}}$. 
\\
For the base cases, either the element is in $R\cup R^{-1}$ or is an identity.  In both cases the result is clear. }%
For the inductive step,
consider the rule:
\[
 \frac{(q_1,S_1,\sigma_1,q_2, S_2)\in \Rhat^h\qquad (q_2,S_2,\sigma_2, q_3, S_3)\in \Rhat^h}{(q_1, S_1, \sigma_1;\sigma_2, q_3, S_3)\in \Rhat^h}\;(\textsc{Tr})
\]
and assume that the premises satisfy the \FSyS\ conditions in $\Phat$.
Let us write $\sigma$ for $\sigma_1;\sigma_2$.
Suppose $q_1\xr{t,i_1}q_1'$.
\begin{itemize}[$\bullet$]
\item 
If $\sigma_1(i_1)=i_2\in[1,r]$ then, by the \FSyS\ conditions on $(q_1,S_1,\sigma_1,h,q_2, S_2)$, 
we have $q_2\xr{t,i_2}q_2'$ with $j_2=\sigma_1(j_1)$ and $(q_1',S_1)\myPhat[h]{\sigma_1}(q_2',S_2)$.
\begin{itemize}
\item If $\sigma_2(i_2)=i_3\in[1,r]$ then 
 $q_{3}\xr{t,i_3}q_{3}'$ with $(q_2',S_2)\myPhat[h]{\sigma_2}(q_{3}',S_{3})$. By~(\textsc{Tr}), $(q_1',S_1)\myPhat[h]{\sigma}(q_{3}',S_{3})$.
\item If $\sigma_2(i_2)=i_3'\in[r{+}1,3r]$
then $q_{3}\xr{t,i_3\fre}q_{3}'$ with 
$(q_2',S_2)\myPhat[h]{\us{\xsw{i_3}{i_3'}}\sigma_2}(q_{3}',S_{3}\us{\xsw{i_3}{i_3'}})$. By~(\textsc{Tr}), and using also Lemma~\ref{lem:updates1}, we obtain $(q_1',S_1)\myPhat[h]{\us{\xsw{i_3}{i_3'}}\sigma}(q_{3}',S_{3}\us{\xsw{i_3}{i_3'}})$, as required.
\item If $i_2\in S_2\setminus\dom{\sigma_2}$ then $q_{3}\xr{t,i_3\fre}q_{3}'$ with 
$(q_2',S_2)\myPhat[h]{\sigma_2[i_2\mapsto i_3]}(q_{3}',S_{3}[i_3])$. By~(\textsc{Tr}), we obtain $(q_1',S_1)\myPhat[h]{\sigma_1;\sigma_2[i_2\mapsto i_3]}(q_{3}',S_{3}[i_3])$, which is what is required since 
$\sigma[i_1\mapsto i_3]=\sigma_1;\allowbreak\sigma_2[i_2\mapsto i_3]$.
\end{itemize}
\item
If $\sigma_1(i_1)=i_2'\in[r{+}1,3r]$ then
$q_2\xr{t,i_2\fre}q_2'$ with $(q_1',S_1)\myPhat[h]{\us{\xsw{i_2}{i_2'}}\sigma_1}(q_2',S_2\us{\xsw{i_2}{i_2'}})$.
\begin{itemize}
\item
If $\sigma_2(i_2')=i_3\in[1,r]$ then 
 $q_{3}\xr{t,i_3}q_{3}'$ with $(q_2',S_2\us{\xsw{i_2}{i_2'}})\myPhat[h]{\sigma_2\us{\xsw{i_2}{i_2'}}}(q_{3}',S_{3})$. By~(\textsc{Tr}) we obtain $(q_1',S_1)\myPhat[h]{\sigma}(q_{3}',S_{3})$.
\item 
If $\sigma_2(i_2')=i_3'\in[r{+}1,3r]$
then $q_{3}\xr{t,i_3\fre}q_{3}'$ with 
$(q_2', S_2\us{\xsw{i_2}{i_2'}})\myPhat[h]{\us{\xsw{i_3}{i_3'}}\sigma_2\us{\xsw{i_2}{i_2'}}}\linebreak[5](q_{3}',S_{3}\us{\xsw{i_3}{i_3'}})$. By~(\textsc{Tr}) we have $(q_1',S_1)\myPhat[h]{\us{\xsw{i_3}{i_3'}}\sigma}(q_{3}', S_{3}\us{\xsw{i_3}{i_3'}})$.
\end{itemize}
\item 
If $i_1\in S_1\setminus\dom{\sigma_1}$ then we have $h=\infty$ and $q_2\xr{t,i_2\fre}q_2'$ with $(q_1',S_1)\myPhat[h]{\sigma_1[i_1\mapsto i_2]}(q_2',S_2[i_2])$, for some $i_2$, so 
$q_{3}\xr{t,\us{i_3}\fre}q_{3}'$ with $(q_2',S_2[i_2])\myPhat[h]{\sigma_2[i_2\mapsto i_3]}(q_{3}',S_{3}[i_3])$ for some $i_3$. By~(\textsc{Tr,Ext}), using $\sigma_1[i_1\mapsto i_2];\sigma_2[i_2\mapsto i_3]\leq_{S_1,S_{3}[i_3]}\sigma[i_1\mapsto i_3]$, we get
$(q_1,S_1)\myPhat[h]{\sigma[i_1\mapsto i_3]}(q_{3},S_{3}[i_3])$.
\end{itemize}
Now suppose $q_1\xr{t,i_1\fre}q_1'$ and let $i_1'\in S_1\setminus[1,r]$ (so $h\le 2r$). 
\begin{itemize}[$\bullet$]
\item If $\sigma_1(i_1')=i_2\in[1,r]$ then $q_2\xr{t,i_2}q_2'$ with $(q_1',S_1\us{\xsw{i_1}{i_1'}})\myPhat[h]{\sigma_1\us{\xsw{i_1}{i_1'}}}(q_2',S_2)$.
\begin{itemize}
\item
If $\sigma_2(i_2)=i_3\in[1,r]$ then 
 $q_{3}\xr{t,i_3}q_{3}'$ with $(q_2', S_2)\myPhat[h]{\sigma_2}(q_{3}',S_{3})$. By~(\textsc{Tr}) we obtain 
 \[
 (q_1',S_1\us{\xsw{i_1}{i_1'}})\myPhat[h]{\sigma\us{\xsw{i_1}{i_1'}}}(q_{3}',S_{3}).
 \]
\item 
If $\sigma_2(i_2)=i_3'\in[r{+}1,3r]$
then $q_{3}\xr{t,i_3\fre}q_{3}'$ with 
$(q_2', S_2)\myPhat[h]{\us{\xsw{i_3}{i_3'}}\sigma_2}(q_{3}',S_{3}\us{\xsw{i_3}{i_3'}})$. By~(\textsc{Tr}) we have $(q_1',S_1\us{\xsw{i_1}{i_1'}})\myPhat[h]{\us{\xsw{i_3}{i_3'}}\sigma\us{\xsw{i_1}{i_1'}}}(q_{3}',S_{3}\us{\xsw{i_3}{i_3'}})$.
\end{itemize}
\item If $\sigma_1(i_1')=i_2'\in[r{+}1,3r]$ then $q_2\xr{t,i_2\fre}q_2'$ with $(q_1',S_1\us{\xsw{i_1}{i_1'}})\myPhat[h]{\us{\xsw{i_2}{i_2'}}\sigma_1\us{\xsw{i_1}{i_1'}}}\linebreak[5](q_2', S_2\us{\xsw{i_2}{i_2'}})$.
\begin{itemize}
\item
If $\sigma_2(i_2')=i_3\in[1,r]$ then 
 $q_{3}\xr{t,i_3}q_{3}'$ with $(q_2',S_2\us{\xsw{i_2}{i_2'}})\myPhat[h]{\sigma_2\us{\xsw{i_2}{i_2'}}}(q_{3}',S_{3})$. By~(\textsc{Tr}) we obtain $(q_1', S_1\us{\xsw{i_1}{i_1'}})\myPhat[h]{\sigma \us{\xsw{i_1}{i_1'}}}(q_{3}',S_{3})$.
\item 
If $\sigma_2(i_2')=i_3'\in[r{+}1,3r]$
then $q_{3}\xr{t,i_3\fre}q_{3}'$ with 
$(q_2', S_2\us{\xsw{i_2}{i_2'}})\myPhat[h]{\us{\xsw{i_3}{i_3'}}\sigma_2\us{\xsw{i_2}{i_2'}}}\linebreak[5](q_{3}', S_{3}\us{\xsw{i_3}{i_3'}})$. By~(\textsc{Tr}), $(q_1', S_1\us{\xsw{i_1}{i_1'}})\myPhat[h]{\us{\xsw{i_3}{i_3'}}\sigma\us{\xsw{i_1}{i_1'}}}(q_{3}', S_{3}\us{\xsw{i_3}{i_3'}})$.
\end{itemize}
\end{itemize}
On the other hand, if $q_1\xr{t,i_1\fre}q_1'$ and $i_3\in S_3\setminus\rng{\sigma}$ (so $h=\infty$). 
\begin{itemize}[$\bullet$]
\item
If $i_3\in\rng{\sigma_2}$ and $i_2=\sigma_2^{-1}(i_3)\notin\rng{\sigma_1}$ then $q_2\xr{t,i_2}q_2'$ with $(q_1',S_1[i_1])\myPhat[h]{\sigma_1[i_1\mapsto i_2]}(q_2',S_2)$, 
and so $q_{3}\xr{t,i_3}q_{3}'$ with $(q_2',S_2)\myPhat[h]{\sigma_2}(q_{3}',S_{3})$. By~(\textsc{Tr})  obtain $(q_1',S_1[i])\myPhat[h]{\sigma[i_1\mapsto i_3]}(q_{3}',S_{3})$.
\item
If $i_3\in S_{3}\setminus\rng{\sigma_2}$ then, since $q_2\xr{t,i_2\fre}q_2'$ with $(q_1',S_1[i_1])\myPhat[h]{\sigma_1[i_1\mapsto i_2]}(q_2,S_2[i_2])$ for some $i_2$, 
we also have $q_{3}\xr{t,i_3}q_{3}'$ with $(q_2',S_2[i_2])\myPhat[h]{\sigma_2[i_2\mapsto i_3]}(q_{3}',S_{3})$. By~(\textsc{Tr,Ext}) we obtain 
\[
(q_1',S_1[i_1])\myPhat[h]{\sigma[i_1\mapsto i_3]}(q_{3}',S_{3}).
\]
\end{itemize}
Finally, let $q_1\xr{t,i_1\fre/i_1\gfre}q_1'$. Then, $q_2\xr{t,i_2\fre/i_2\gfre}q_2'$ and 
$q_3\xr{t,i_3\fre/i_3\gfre}q_3'$ with $(q_1',S_1')\myPhat[h']{\sigma_1'}(q_{2}',S_{2}')$
and $(q_2',S_2')\myPhat[h']{\sigma_2'}(q_{3}',S_{3}')$.
\begin{itemize}[$\bullet$]
\item 
If $h<2r$ then 
$h'=h+1$
and
$i_k'=\min([r{+}1,3r]\setminus S_k)$, 
$S_k'= S_k[i'_k]\us{\xsw{i_k}{i_k'}}$ and $\sigma_k'=\us{\xsw{i_{k+1}}{i_{k+1}'}}(\sigma_k[i_k'\mapsto i_{k+1}'])\us{\xsw{i_{k}}{i_{k}'}}$, for $k=1,2,3$. By~(\textsc{Tr}), we have
$(q_1',S_1')\myPhat[h]{\sigma_1';\sigma_2'}(q_{3}',S_{3}')$, which is as required since $\sigma_1';\sigma_2'=\us{\xsw{i_3}{i_3'}}(\sigma[i_1'\mapsto i_3'])\us{\xsw{i_1}{i_1'}}$.
\item 
If $h=2r$ then
$h'=\infty$ and
$S_k'=S_k[i_k]\cap[1,r]$ and $\sigma_k'=\sigma_k[i_k\mapsto i_{k+1}]\cap[1,r]^2$. By~(\textsc{Tr}), we have
$(q_1',S_1')\myPhat[h]{\sigma_1';\sigma_2'}(q_{3}',S_{3}')$ and, hence, by~(\textsc{Ext}) we obtain the required result since $\sigma_1';\sigma_2'\,\le_{S_1',S_2'}\,\sigma[i_1\mapsto i_3]\cap[1,r]^2$.
\item 
If $h=\infty$ then $h'=\infty$ and $S_k'=S_k[i_k]$ and $\sigma_k'=\sigma_k[i_k\mapsto i_{k{+}1}]$.
By~(\textsc{Tr,Ext}), 
\[
(q_1',S_1[i_1])\myPhat[h]{\sigma[i_1\mapsto i_3]}(q_{3}',S_{3}[i_3]). 
\]
Moreover, if the transition from $q_1$ to $q_1'$ is localy fresh then so is the one from $q_2$ to $q_2'$, and from $q_3$ to $q_3'$.
\end{itemize}
\cutout{Finally, suppose $q_1\xr{t,i_1\gfre}q_1'$.
\begin{itemize}[$\bullet$]
\item 
Then, $q_2\xr{t,i_2\gfre}q_2'$ with $(q_1',S_1[i_1])\myPhat[h]{\sigma_1[i_1\mapsto i_2]}(q_2',S_2[i_2])$, \hbox{some $i_2$, so
$q_{3}\xr{t,i_3\gfre}q_{3}'$ with $(q_2',S_2[i_2])\myPhat[h]{\sigma_2[i_2\mapsto i_3]}(q_{3}',S_{3}[i_3])$} \hbox{for some $i_3$. By~(\textsc{Tr,Ext}), 
$(q_1',S_1[i_1])\myPhat[h]{\sigma[i_1\mapsto i_3]}(q_{3}',S_{3}[i_3])$.}
\end{itemize}}
\smallskip

\noindent
We now consider the rule:
\[
 \frac{(q_1,S_1,\sigma,q_2, S_2)\in \Rhat^h\qquad \sigma\le_{S_1,S_2} \sigma'}{ (q_1,S_1,\sigma',q_2, S_2)\in \Rhat^h}\;(\textsc{Ext})
\]
and assume $(q_1,S_1,\sigma,h, q_2, S_2)$ satisfies the \FSyS\ conditions in $\Phat$.
Note that if $h<\infty$ then $h=|\sigma|=|\sigma'|$, hence $\sigma=\sigma'$ and the required result is trivial. So let us assume $h=\infty$.
Suppose $q_1\xr{t,i}q_1'$.
\begin{itemize}[$\bullet$]
\item 
If $i\in\dom{\sigma}$ then $q_{2}\xr{t,\sigma(i)}q_{2}'$ and $(q_1',S_1)\myPhat[\infty]{\sigma}(q_{2}',S_{2})$. Since $\sigma\subseteq\sigma'$, we have $\sigma(i)=\sigma'(i)$ and $(q_1',S_1)\myPhat[\infty]{\sigma'}(q_{2}',S_{2})$.
\item 
If $i\notin\dom{\sigma'}$ then also $i\notin\dom{\sigma}$ and therefore
$q_{2}\xr{t,j\fre}q_{2}'$, for some $j$, and 
\[
(q_1',S_1)\myPhat[\infty]{{\sigma[i\mapsto j]}}(q_{2}',S_{2}[j]).
\] 
From $\sigma\leq_{S_1,S_{2}}\sigma'$ we obtain $\sigma[i\mapsto j]\leq_{S_1,S_{2}[j]}\sigma'[i\mapsto j]$, 
so $(q_1',S_1)\myPhat[\infty]{\sigma'[i\mapsto j]}(q_{2}',S_{2}[j])$.
\item 
If $i\in\dom{\sigma'}\setminus\dom{\sigma}$ then we reason as follows. Let $\sigma'(i)=j\in S_{2}$.
\begin{enumerate}[I.]
\item Since $i\notin\dom{\sigma}$, there is some $q_{2}\xr{t,{j'}\fre}q_{2}''$ with $(q_1',S_1)\myPhat[\infty]{\sigma[i\mapsto j']}(q_{2}'',S_{2}[j'])$;
\item hence, there is some $q_1\xr{t,{i'}\fre}q_1''$ with 
$(q_1'',S_1[i'])\myPhat[\infty]{\sigma[i'\mapsto j']}(q_{2}'',S_{2}[j'])$;
\item then, there is some $q_{2}\xr{t,j}q_{2}'$ with 
$(q_1'',S_1[i'])\myPhat[\infty]{\sigma[i'\mapsto j]}(q_{2}',S_{2})$.
\end{enumerate}
Taking stock (and using symmetry of $\Phat$),
\[
(q_1',S_1)\myPhat[\infty]{\sigma[i\mapsto j']}(q_{2}'',S_{2}[j'])\myPhat[\infty]{\sigma^{-1}[j'\mapsto i']}
(q_1'',S_1[i'])
\myPhat[\infty]{\sigma[i'\mapsto j]}(q_{2}',S_{2})
\]
and thus, since $\sigma[i\mapsto j'];\sigma^{-1}[j'\mapsto i'];\sigma[i'\mapsto j]\leq_{S_1,S_{2}}\sigma[i\mapsto j]\leq_{S_1,S_{2}}\sigma'$, we have $(q_1',S_1)\myPhat[\infty]{\sigma'}(q_{2}',S_{2})$.
\end{itemize}
Suppose now $q_1\xr{t,i\fre}q_1'$.
\begin{itemize}[$\bullet$]
\item 
Then, $q_{2}\xr{t,j\fre}q_{2}'$ and $(q_1',S_1[i])\myPhat[\infty]{\sigma[i\mapsto j]}(q_{2}',S_{2}[j])$. Since $\sigma[i\mapsto j]\leq_{S_1[i],S_{2}[j]}\sigma'[i\mapsto j]$, we have $(q_1',S_1[i])\myPhat[\infty]{\sigma'[i\mapsto j]}(q_{2}',S_{2}[j])$.
\item 
If $j\in S_{2}\setminus\rng{\sigma'}$ then $j\notin\rng{\sigma}$, hence $q_{2}\xr{t,j}q_{2}'$ and $(q_1',S_1[i])\myPhat[\infty]{\sigma[i\mapsto j]}(q_{2}',S_{2})$. Again, we obtain $(q_1',S_1[i])\myPhat[\infty]{\sigma'[i\mapsto j]}(q_{2}',S_{2})$.
\end{itemize}
Finally, let $q_1\xr{t,i\gfre}q_1'$.
\begin{itemize}[$\bullet$]
\item 
Then, $q_{2}\xr{t,j\gfre}q_{2}'$ and $(q_1',S_1[i])\myPhat[\infty]{\sigma[i\mapsto j]}(q_{2}',S_{2}[j])$. Since $\sigma[i\mapsto j]\leq_{S_1[i],S_{2}[j]}\sigma'[i\mapsto j]$, we have $(q_1',S_1[i])\myPhat[\infty]{\sigma'[i\mapsto j]}(q_{2}',S_{2}[j])$.
\end{itemize}
\smallskip

\noindent
Hence, $\Rhat$ satisfies the \FSyS\ conditions in $\Phat$.
\end{proof}


%% file: apx-pspace.tex

\section{PSPACE completeness of inverse subsemigroup membership\label{apx-pspace}}

Given an inverse semigroup $\calG$, 
an \emph{inverse subsemigroup} of $\calG$ is some inverse semigroup $\mathcal{H}\subseteq\calG$.
The problem of inverse subsemigroup membership of $\calG$:
\begin{quote}
For a set $G$ of elements of $\calG$ and a distinguished element $g$ of $\calG$, does $g\in\abra{G}$?
\end{quote}
where $\abra{G}$ is the inverse semigroup generated by the members of $G$ via composition and inversion.
In this section we prove the following result.

\begin{thm}\label{t:pspace}
Checking membership in inverse subsemigroups of $\is{n}$ is PSPACE-complete.
\end{thm}

Note first that PSPACE membership follows from Kozen's corresponding PSPACE result for functions,
as members of $\is{n}$ can be seen as functions on $[1,n]\cup\{\#\}$.

\begin{thmC}[\cite{K77}]
Checking whether a function $h:[1,n]\to[1,n]$ can be generated from given functions $f_1,\cdots,f_k:[1,n]\to[1,n]$ is PSPACE-complete.
\end{thmC}

For hardness, we shall make use of a result of Lewis and Papadimitriou which shows that PSPACE computations correspond to computations performed in polynomial space by Turing machines with symmetric transitions.

\begin{defiC}[\cite{LewisPap82}]
A \emph{symmetric Turing Machine} is a tuple $\calM=\abra{Q,q_0,\delta,F}$ where:
\begin{itemize}
\item $Q$ is a set of states, $q_0\in Q$ is initial and $F\subseteq Q$ are final,
\item $\delta\subseteq (Q\times\{0,1\}\times\{0\}\times\{0,1\}\times Q)\cup
(Q\times\{0,1\}^2\times\{-1,+1\}\times\{0,1\}^2\times Q)$ is the transition relation,
\end{itemize}
such that $\delta=\delta^{-1}$, where $\delta^{-1}=\{t^{-1}\ |\ t\in\delta\}$ and:
\begin{itemize}
\item $(q,a,0,b,q')^{-1}=(q',b,0,a,q)$,
\item $(q,a,b,A,c,d,q')^{-1}=(q',c,d,-A,a,b,q)$.
\end{itemize}
\end{defiC}

Note that our machines have input and tape alphabet $\{0,1\}$. Moreover, since we are only examining machines running in polynomial space, we assume a single tape (i.e.\ no separate input/work tapes), which is initially empty.\footnote{Lewis\,\&\,Papadimitriou work with multi-tape automata, which they reduce to 2-tape automata with one tape for input and one work tape. The same procedure can be used to reduce to just one tape, retaining the same space complexity if the initial complexity is at least polynomial.}
A symmetric TM $\calM$ operates just as a TM, with the
feature that $\calM$ can look 2 symbols ahead:\footnote{This feature does not add expressiveness to a TM but allows one to define symmetric machines.}
e.g.\ a transition
$(q,a,b,+1,c,d,q')$ means that, if the automaton is at state $q$, with the tape symbol at the head being $a$ and the tape symbol to the right of the head being $b$, then the automaton will rewrite those symbols to $c,d$ respectively, move the head to the right and go to state $q'$.
In a transition $(q,a,b,-1,c,d,q')$ we have the dual behaviour: the automaton looks one symbol to the left ahead, and moves the head to the left. Transitions of the form $(q,a,0,b,q')$ leave the head unmoved.

Given $f:\mathbb{N}\to\mathbb{N}$,
we let SSPACE$(f)$ be the class of problems decided by a symmetric TM in space $O(f)$.

\begin{thmC}[\cite{LewisPap82}]
For any $f:\mathbb{N}\to\mathbb{N}$,
\[
\mathrm{DSPACE}(f)\subseteq \mathrm{SSPACE}(f)\subseteq \mathrm{NSPACE}(f).
\]
\end{thmC}
Hence, setting
$
\mathrm{SPSPACE} = \bigcup_{i\in\mathbb{N}}\mathrm{SSPACE}(n^i),
$
using also Savitch's theorem we have $\mathrm{SPSPACE}=\mathrm{PSPACE}$.

\begin{proof}[Proof of Theorem~\ref{t:pspace}]
It suffices to show that the problem is PSPACE-hard.
Suppose that $\calM$ is a symmetric TM with set of states $Q=[1,K]$ and a tape of size $N$. 
By convention, we assume that the initial state is 1, the initial head position is 1 and the unique final state is $K$.
We will simulate its computation using partial permutations from $\is{n}$, where
$n = 2N + N + K+1$.

The first $2N$ numbers in $n$ are used for modelling the tape, the next $N$ numbers for storing the position of the head on the tape, and the last $K+1$ ones for storing the current state, where we include an extra dummy state ($K+1$) to be used at the beginning of the simulation. 
The way we model these data (tape, head, state) is by employing $N + 1 + 1$ ``tokens'' which we distribute among our $n$ numbers as follows:
\begin{itemize}
\item
One token is shared between $2i-1$ and $2i$, for each $i\in[1,N]$.
This token represents the value of bit $i$ of the tape. E.g.\ if the tape is $10\cdots 0$, then we can think of the tokens being on numbers $2,3,5,\cdots,2N-1$.
\item
One token is shared between the numbers $2N+1,\cdots,3N$. This token represents the position of the head. E.g.\ if the tape is on position 5, then this token will be on number $2N+5$.
\item
One token is shared between the numbers $3N+1,...,3N+K+1$. This token represents the current state.
\end{itemize}
Initially, we will require all tokens to be on positions $2i-1$ ($i\in[1,N]$), $2N+1$ and $3N+K+1$.
The latter means that the last token is initially placed on the dummy state $K+1$.

We model transitions as partial permutations that pass on the $2N+2$ tokens. E.g.\ consider the transition $t=(3,0,0,+1,1,0,5)$.\footnote{i.e.\ from state 3, if the head of the tape and its right-successor read 00 then write 10 to them, move right and go to state 5.}
Then,
$t$ is modelled by partial permutations:
\begin{align*}
\pi_t^i &= \{ (2i-1, 2i) \} \cup 
\{ (2(i+1)-1, 2(i+1)-1) \} \\
&\quad\cup\{ (j, j) \in [1,2N]\times[1,2N]\ |\ j\not=2i-1,2i,2i+1,2i+2 \}\\
       &\quad \cup\; \{ (2N+i, 2N+i+1) \}\\ &\quad\cup\; \{ (3N+3, 3N+5) \}
\end{align*}
for $i\in[1,N-1]$. The first line above says ``at position $i$, read 0 and write 1" and 
``at position $i+1$, read 0 and write 0";
the second line ``leave the remaning cells unchanged'';
third line ``move right"; and the fourth one ``from state 3 go to state 5". 
This can be generalised to all of $\delta$:
\begin{itemize}
\item for all $t=(x,a,b,A,c,d,y)$ and $i\in[1,N]$ such that $i+A\in[1,N]$, set \\
$\pi_t^i=\{ (2i-2+A+a, 2i-2+A+c) \} \cup 
\{ (2i+A+b, 2i+A+d) \}
\cup\{ (j, j) \in [1,2N]\times[1,2N]\ |\ j\notin[2i-2+A,2i+1+A] \}
       \cup \{ (2N+i, 2N+i+A) \} \cup \{ (3N+x, 3N+y) \}$
\item for all $t=(x,a,0,b,y)$ and $i\in[1,N]$, set $\pi_t^i=\{ (2i-1+a, 2i-1+b) \} \cup 
\{ (j, j) \in [1,2N]\times[1,2N]\ |\ j\notin[2i-1,2i] \}
       \cup \{ (2N+i, 2N+i) \}\cup \{ (3N+x, 3N+y) \}$
\end{itemize} 
Note that, in the latter case, $(\pi_t^i)^{-1}=\pi^i_{t^{-1}}$ and, in the former one,
$(\pi_t^{i})^{-1}=\pi_{t^{-1}}^{i+A}$.

Let us write $X$ for the set of all such partial permutations. If $\calM$ has $d$ many transitions then the size of $X$ is at most $d\cdot N$. Let us also select $Y$ to be a minimal set of generators for the group of partial permutations of the form:
\[
\pi' = \pi_1\cup\pi_2\cup\{(3N+K,3N+K)\}
\]
where $\pi_1:[1,2N]\overset{\cong}{\to}[1,2N]$ 
and $\pi_2:[2N+1,3N]\overset{\cong}{\to}[2N+1,3N]$.
%
Note that $|Y|\leq 3n/2$. Moreover, let us take 
\[\begin{split}
\pi_0=\{(2i-1,2i-1)\ |\ i\in[1,N]\}\cup\{(2N+1,2N+1)\}\\
\cup\{(3N+K+1,3N+1)\}
\end{split}\]
to be a permutation setting up the initial positions of the tokens.
We then have that:
\begin{equation}\label{eq:calM}
\calM\text{ terminates } \iff \pi_{\calM} \in \abra{X\cup Y\cup\{\pi_0\}}
\end{equation}
where $\pi_{\calM}$ is the partial permutation:
\[\begin{split}
\pi_{\calM} = \{ (2i-1, 2i-1)\ |\ i \in [1,N] \}
     \cup \{ (2N+1, 2N+1) \}\\
     \cup \{ (3N+K+1, 3N+K) \}
\end{split}\]
To prove~\eqref{eq:calM}, note first that any accepting run of $\calM$, say 
\[
(q_0,H_0,\alpha_0)\xr{t_1}(q_1,H_1,\alpha_1)\cdots\xr{t_k}(q_k,H_k,\alpha_k)
\]
where $q_0=1$, $H_0=1$, $\alpha_0=0^N$ and $q_k=K$, yields a permutation $\pi=\pi_0;\pi^{H_0}_{t_1};\cdots;\pi_{t_k}^{H_k}$ with the property that 
$\dom{\pi}=\dom{\pi_0}$ and
$\pi(3N+1)=3N+K$. 
We can now select some $\pi'\in\abra{Y}$ such that $\pi'\upharpoonright\dom{\pi}=(\pi\upharpoonright[1,3N])\cup\{(3N+K,3N+K)\}$
and, hence, $\pi;\pi'^{-1}=\pi_{\calM}$.

Conversely, suppose that $\pi_\calM\in\abra{X\cup Y\cup\{\pi_0\}}$ and in particular let $\pi_{\calM}=\pi_0;\pi_1;\cdots;\pi_k$ be a production (so each $\pi^i$ is in $X\cup Y\cup\{\pi_0\}\cup X^{-1}\cup Y^{-1}\cup\{\pi_0^{-1}\}$). 
Note that, because $\pi_0$ is the only generator with $3N+K+1$ in its domain, it must be the leftmost one in the production. 
Let $k'\leq k$ be the least index such that 
$\pi_{k'}\notin Y\cup Y^{-1}$ and, for all $j>k'$, $\pi_j\in Y\cup Y^{-1}$, and assume the production is minimal with respect to the value $(k',k)$ (in the lexicographic ordering).
We first claim that there is no $\pi_j$ with $j< k'$ such that $\pi_j\in Y\cup Y^{-1}$.
Because if that were the case then $\pi'=\pi_0;\cdots;\pi_{j-1}$ would satisfy $\dom{\pi'}=\dom{\pi_{\calM}}$ and $\pi'(3N+K+1)=3N+K$ so there would be some $\pi''\in\abra{Y}$ such that $\pi_{\calM}=\pi_0;\cdots;\pi_{j-1};\pi''$, and the latter would lead to a production with size $(j-1,\cdots)$ which would be smaller than $(k',k)$.
Moreover, if $\pi_i=\pi_0$ for some $i>0$ then we must have $\pi_{i-1}=\pi_0^{-1}$. Because $\pi_0^{-1};\pi_0=\id{{\rng{\pi_0}}}$ 
and $|\pi_0^{-1};\pi_0|=|\pi_{\calM}|=N+2$, we have that $\pi_0^{-1};\pi_0$ can be safely removed from the production of $\pi$, thus contradicting the minimality of the latter.
For similar reasons, $\pi_i\not=\pi_0^{-1}$, for all $i\in[1,k]$. Hence, $\pi_0$ only occurs at the beginning of the production and $\pi_0^{-1}$ does not occur at all. Summing up, $\pi=\pi_0;\pi_A;\pi_B$ with $\pi_A\in\abra{X}$ and $\pi_B\in\abra{Y}$. We can now see that $\pi_A$ represents a computation of $\calM$ from $1$ to $K$.
\end{proof}
